\newif\ifpageseparation
\pageseparationtrue

\newif\ifonecol
\onecolfalse
\ifonecol
  \documentclass[journal,onecolumn]{IEEEtran} 
\else
  \documentclass[journal,twocolumn]{IEEEtran} 
\fi

\usepackage{balance}
\usepackage[export]{adjustbox}

\usepackage{amssymb}
\usepackage{amsmath}
\usepackage{amsthm}
\usepackage{thmtools}

\usepackage{tikz}
\usetikzlibrary{calc, arrows, decorations.markings, shapes, shapes.geometric, positioning, 3d}
\usetikzlibrary{matrix}
\usepackage{xifthen}

\usepackage{calc} %
\usepackage{subfigure}
\usepackage{graphicx}
\usepackage{comment}
\usepackage{color}
\usepackage{xcolor}

\date{}

\newcommand{\F}{\mathbb F}
\newcommand{\R}{\mathbb R}
\newcommand{\img}{ {\rm im\,}}

\newcommand{\A}{{\cal A}}
\newcommand{\setC}{{\cal C}}
\newcommand{\D}{{\cal D}}
\newcommand{\E}{{\cal E}}
\newcommand{\G}{{\cal G}}

\newcommand{\N}{{\cal N}}
\newcommand{\setP}{{\cal P}}
\newcommand{\V}{{\cal V}}
\newcommand{\setS}{{\cal S}}
\newcommand{\s}{{\cal S}}
\newcommand{\T}{{\cal T}}
\newcommand{\X}{{\cal X}}
\newcommand{\Y}{{\cal Y}}
\newcommand{\vnfg}{{\cal F}}
\newcommand{\B}{{\cal B}}
\newcommand{\code}{{\cal C}}
\newcommand{\I}{{\cal I}}
\newcommand{\eqpunc}{\,}

\newcommand{\Ising}{I}
\newcommand{\Nisingbeta}{\N_{\mathrm{\Ising}(\beta)}}
\newcommand{\Nisingbetatilde}{\N_{\mathrm{\Ising}(\widetilde{\beta})}}
\newcommand{\Nisingbetatildeh}{\N_{\mathrm{\Ising}(\widetilde{\beta})}^{\mathrm{h}}}
\newcommand{\Nisingbetatildev}{\N_{\mathrm{\Ising}(\widetilde{\beta})}^{\mathrm{v}}}
\newcommand{\Nisingbetatildehv}{\N_{\mathrm{\Ising}(\widetilde{\beta})}^{\mathrm{hv}}}
\newcommand{\FTNisingbeta}{\widehat{\N}_{\mathrm{\Ising}(\beta)}}

\newcommand*{\qedb}{\hfill\ensuremath{\blacksquare}}%
\newcommand*{\qede}{\hfill\ensuremath{\triangle}}%
\newcommand*{\qedd}{\hfill\ensuremath{\blacksquare}}%
\newcommand*{\qeda}{\hfill\ensuremath{\blacksquare}}%
\newcommand*{\qedt}{\hfill\ensuremath{\blacksquare}}%
\newcommand*{\qedl}{\hfill\ensuremath{\blacksquare}}%

\newtheorem{theorem}{Theorem}
\newtheorem{example}[theorem]{Example}
\newtheorem{remark}[theorem]{Remark}
\newtheorem{definition}[theorem]{Definition}
\newtheorem{lemma}[theorem]{Lemma}
\newtheorem{assumption}[theorem]{Assumption}

\newcommand{\refrecipe}{Construction}

\newcommand{\defeq}{\triangleq}

\specialcomment{comment}{\begingroup\color{blue}}{\endgroup}
\specialcomment{commentold}{\begingroup\color{blue}}{\endgroup}

\def\angle{30}
\pgfmathsetmacro\cosine{cos(\angle)}
\pgfmathsetmacro\sine{sin(\angle)}

\newcommand\cube[3]{
	\path[shift={(0:#1)},shift={(90:#2)}, shift={(\angle:#3*\cosine)}]
		(0,0)  coordinate (A)  (1,0)     coordinate (B) %
		+(\angle:\cosine) coordinate (C)  (\angle:\cosine) coordinate (D)
		(0,1)  coordinate (E) +(\angle:\cosine) coordinate (H)   %
		(1,1)  coordinate (F) +(\angle:\cosine) coordinate (G)  ;
	\draw[draw=black]
		(A) -- (E) (B) -- (F) (C) -- (G)  %
		(A) -- (B) -- (C) %
		(E) -- (F) -- (G) -- (H) -- (E); %
	\draw[fill=red!40, draw=black]
		(D) -- (H) (D) -- (C) (D) -- (A);  
		\def\gap{.99}
		\node[inner sep=0,minimum size=0](cf) at (barycentric cs:A=1,B=1,F=1,E=1){};
		\path(cf) -- coordinate[pos=\gap] (Af) (A.0);
		\path(cf) -- coordinate[pos=\gap] (Bf) (B);
		\path(cf) -- coordinate[pos=\gap] (Ff) (F);
		\path(cf) -- coordinate[pos=\gap] (Ef) (E);
		\draw(F)--(Ff);
		\node[inner sep=0,minimum size=0](cr) at (barycentric cs:B=1,C=1,G=1,F=1){};
		\path(cr) -- coordinate[pos=\gap] (Br) (B);
		\path(cr) -- coordinate[pos=\gap] (Cr) (C);
		\path(cr) -- coordinate[pos=\gap] (Gr) (G);
		\path(cr) -- coordinate[pos=\gap] (Fr) (F);
		\node[inner sep=0,minimum size=0](ct) at (barycentric cs:E=1,F=1,G=1,H=1){};
		\path(ct) -- coordinate[pos=\gap] (Et) (E);
		\path(ct) -- coordinate[pos=\gap] (Gt) (G);
		\path(ct) -- coordinate[pos=\gap] (Ft) (F);
		\path(ct) -- coordinate[pos=\gap] (Ht) (H);

	\path[fill=white,fill opacity=.6]
		(Af)--(Bf) -- (Ff)-- (Ef)
		(Et)--(Ft) -- (Gt)-- (Ht)
		(Br)--(Cr) -- (Gr)-- (Fr);
	\foreach \vertex in {A,...,C} {\node[vertex](v\vertex) at (\vertex) {};}
	\foreach \vertex in {E,...,H} {\node[vertex](v\vertex) at (\vertex) {};}
	\node[vertex,fill=gray](vD) at (D) {};       
}

\newcommand\cubedirected[3]{
	\cube{#1}{#2}{#3}
	\draw[>=latex, ->] (vA)--(vE);
	\draw[>=latex, ->] (vA) -- (vB);
	\draw[>=latex, ->,gray] (vA) -- (vD);
	\draw[>=latex, ->] (vB) -- (vC); 
	\draw[>=latex, ->]	(vB) -- (vF);
	\draw[>=latex, ->]	(vC) -- (vG);
	\draw[>=latex, ->, gray]	(vD) -- (vC);
	\draw[>=latex, ->, gray]	(vD) -- (vH);
	\draw[>=latex, ->]	(vE) -- (vF);
	\draw[>=latex, ->]	(vE) -- (vH);
	\draw[>=latex, ->]	(vF) -- (vG);
	\draw[>=latex, ->]	(vH) -- (vG);
}

\newcommand\topdirected[3]{
	\phantom{\cube{#1}{#2}{#3}}
	\foreach \vertex in {E,...,H}{
		\draw[->,>=latex] (v\vertex)	--++(90		:.3) ;
	}
}
\newcommand\bottomdirected[3]{
	\phantom{\cube{#1}{#2}{#3}}
	\foreach \vertex in {A,...,D}{
		\draw[<-,>=latex] (v\vertex)	--++(90		:-.3) ;
	}
}
\newcommand\leftdirected[3]{
	\phantom{\cube{#1}{#2}{#3}}
	\foreach \vertex in {A,D,E,H}{
		\draw[<-,>=latex] (v\vertex)	--++(00		:-.3) ;
	}
}
\newcommand\rightdirected[3]{
	\phantom{\cube{#1}{#2}{#3}}
	\foreach \vertex in {B,C,F,G}{
		\draw[->,>=latex] (v\vertex)	--++(00		: .3) ;
	}
}
\newcommand\frontdirected[3]{
	\cubecoord{#1}{#2}{#3}
	\foreach \vertex in {A,B,E,F}{
		\draw[<-,>=latex](v\vertex)  --++(180+\angle:.3);
	}
}
\newcommand\backdirected[3]{
	\cubecoord{#1}{#2}{#3}
	\foreach \vertex in {C,D,G,H}{
		\draw[->,>=latex](v\vertex)  --++(\angle:.3);
	}
}

\tikzset{vertex/.style={
			shape=circle, %
			minimum size=1mm,
			fill=black,
			inner sep = 0pt}}
\def\extend{.4}
\def\extendclr{gray!50}
\definecolor{aqua}{rgb}{0.0, 1.0, 1.0}
\definecolor{brightcerulean}{rgb}{0.11, 0.67, 0.84}
\definecolor{darkpastelgreen}{rgb}{0.01, 0.75, 0.24}
\def\extendclr{brightcerulean}
\def\extendtype{dashed}

\def\angle{30}
\pgfmathsetmacro\cosine{cos(\angle)}
\pgfmathsetmacro\sine{sin(\angle)}
\newcommand\cubeorientext[3]{
	\path[shift={(0:#1)},shift={(90:#2)}, shift={(\angle:#3*\cosine)}]
		(0,0)  coordinate (A)  (1,0)     coordinate (B) %
		+(\angle:\cosine) coordinate (C)  (\angle:\cosine) coordinate (D)
		(0,1)  coordinate (E) +(\angle:\cosine) coordinate (H)   %
		(1,1)  coordinate (F) +(\angle:\cosine) coordinate (G)  ;
	\draw[draw=black]
		(A) -- (E) (B) -- (F) (C) -- (G)  %
		(A) -- (B) -- (C) %
		(E) -- (F) -- (G) -- (H) -- (E); %
	\draw[fill=red!40, draw=black]
		(D) -- (H) (D) -- (C) (D) -- (A);  
	\draw[dashed,>=latex, \extendclr] (D)edge[<- ]++(180			:\extend);
	\draw[dashed,>=latex, \extendclr] (D)edge[ ->]++(    \angle	:\extend);
	\draw[dashed,>=latex, \extendclr] (D)edge[<- ]++(-90			:\extend) ;

		\def\gap{.99}
		\node[inner sep=0,minimum size=0](cf) at (barycentric cs:A=1,B=1,F=1,E=1){};
		\path(cf) -- coordinate[pos=\gap] (Af) (A.0);
		\path(cf) -- coordinate[pos=\gap] (Bf) (B);
		\path(cf) -- coordinate[pos=\gap] (Ff) (F);
		\path(cf) -- coordinate[pos=\gap] (Ef) (E);
		\draw(F)--(Ff);
		\node[inner sep=0,minimum size=0](cr) at (barycentric cs:B=1,C=1,G=1,F=1){};
		\path(cr) -- coordinate[pos=\gap] (Br) (B);
		\path(cr) -- coordinate[pos=\gap] (Cr) (C);
		\path(cr) -- coordinate[pos=\gap] (Gr) (G);
		\path(cr) -- coordinate[pos=\gap] (Fr) (F);
		\node[inner sep=0,minimum size=0](ct) at (barycentric cs:E=1,F=1,G=1,H=1){};
		\path(ct) -- coordinate[pos=\gap] (Et) (E);
		\path(ct) -- coordinate[pos=\gap] (Gt) (G);
		\path(ct) -- coordinate[pos=\gap] (Ft) (F);
		\path(ct) -- coordinate[pos=\gap] (Ht) (H);

	\node[inner sep=0,minimum size=0](cback) at (barycentric cs:D=1,C=1,G=1,H=1){};
	\node[inner sep=0,minimum size=0](cl) at (barycentric cs:A=1,D=1,H=1,E=1){};
	\node[inner sep=0,minimum size=0](cbottom) at (barycentric cs:A=1,B=1,C=1,D=1){};
	\node(back)[red] at (cback){\clockxy{}};
	\node(bottom)[red, rotate=.5*\angle] at (cbottom){\clockxz{}};
	\node(left)[red, rotate=2*\angle] at (cl){\clockyz{}};
	\path[fill=white,fill opacity=.6]
		(Af)--(Bf) -- (Ff)-- (Ef)
		(Et)--(Ft) -- (Gt)-- (Ht)
		(Br)--(Cr) -- (Gr)-- (Fr);

	\node[red] at (cf){\clockxy{}};
	\node[red, rotate=2*\angle] at (cr){\clockyz{}};
	\node[red, rotate=.5*\angle] at (ct){\clockxz{}};
	\foreach \vertex in {A,...,C} {\node[vertex](v\vertex) at (\vertex) {};}
	\foreach \vertex in {E,...,H} {\node[vertex](v\vertex) at (\vertex) {};}
	\node[vertex,fill=gray](vD) at (D) {};       
	\draw[>=latex, ->] (vA)--(vE);
	\draw[>=latex, ->] (vA) -- (vB);
	\draw[>=latex, ->,gray] (vA) -- (vD);
	\draw[>=latex, ->] (vB) -- (vC); 
	\draw[>=latex, ->]	(vB) -- (vF);
	\draw[>=latex, ->]	(vC) -- (vG);
	\draw[>=latex, ->, gray]	(vD) -- (vC);
	\draw[>=latex, ->, gray]	(vD) -- (vH);
	\draw[>=latex, ->]	(vE) -- (vF);
	\draw[>=latex, ->]	(vE) -- (vH);
	\draw[>=latex, ->]	(vF) -- (vG);
	\draw[>=latex, ->]	(vH) -- (vG);

	\draw[dashed,>=latex, ->, \extendclr]
		(vA)edge[<- ]++(180			:\extend)
		(vA)edge[<- ]++(180+\angle	:\extend)
		(vA)edge[<- ]++(-90			:\extend)

		(vB)edge[ ->]++(0				:\extend)
		(vB)edge[<- ]++(180+\angle	:\extend)
		(vB)edge[<- ]++(-90			:\extend)

		(vC)edge[ ->]++(0				:\extend)
		(vC)edge[ ->]++(    \angle	:\extend)
		(vC)edge[<- ]++(-90			:\extend)

		(vE)edge[<- ]++(180			:\extend)
		(vE)edge[<- ]++(180+\angle	:\extend)
		(vE)edge[ ->]++(90				:\extend)

		(vF)edge[ ->]++(0				:\extend)
		(vF)edge[<- ]++(180+\angle	:\extend)
		(vF)edge[ ->]++(90				:\extend)

		(vG)edge[ ->]++(0				:\extend)
		(vG)edge[ ->]++(    \angle	:\extend)
		(vG)edge[ ->]++(90				:\extend)

		(vH)edge[<- ]++(180			:\extend)
		(vH)edge[ ->]++(    \angle:\extend)
		(vH)edge[ ->]++(90			:\extend)
		;

}

\newcommand\cubeorient[3]{
	\path[shift={(0:#1)},shift={(90:#2)}, shift={(\angle:#3*\cosine)}]
		(0,0)  coordinate (A)  (1,0)     coordinate (B) %
		+(\angle:\cosine) coordinate (C)  (\angle:\cosine) coordinate (D)
		(0,1)  coordinate (E) +(\angle:\cosine) coordinate (H)   %
		(1,1)  coordinate (F) +(\angle:\cosine) coordinate (G)  ;
	\draw[draw=black]
		(A) --coordinate(ae) (E) (B) --coordinate(bf) (F) (C) --coordinate(cg) (G)  %
		(A) --coordinate(ab) (B) --coordinate(bc) (C) %
		(E) --coordinate(ef) (F) --coordinate(fg) (G) --coordinate(gh) (H) --coordinate(eh) (E); %
	\draw[fill=red!40, draw=black]
		(D) --coordinate(dh) (H) (D) --coordinate(cd) (C) (D) --coordinate(ad) (A);  

		\def\gap{.99}
		\node[inner sep=0,minimum size=0](cf) at (barycentric cs:A=1,B=1,F=1,E=1){};
		\path(cf) -- coordinate[pos=\gap] (Af) (A.0);
		\path(cf) -- coordinate[pos=\gap] (Bf) (B);
		\path(cf) -- coordinate[pos=\gap] (Ff) (F);
		\path(cf) -- coordinate[pos=\gap] (Ef) (E);
		\draw(F)--(Ff);
		\node[inner sep=0,minimum size=0](cr) at (barycentric cs:B=1,C=1,G=1,F=1){};
		\path(cr) -- coordinate[pos=\gap] (Br) (B);
		\path(cr) -- coordinate[pos=\gap] (Cr) (C);
		\path(cr) -- coordinate[pos=\gap] (Gr) (G);
		\path(cr) -- coordinate[pos=\gap] (Fr) (F);
		\node[inner sep=0,minimum size=0](ct) at (barycentric cs:E=1,F=1,G=1,H=1){};
		\path(ct) -- coordinate[pos=\gap] (Et) (E);
		\path(ct) -- coordinate[pos=\gap] (Gt) (G);
		\path(ct) -- coordinate[pos=\gap] (Ft) (F);
		\path(ct) -- coordinate[pos=\gap] (Ht) (H);

	\path[fill=white,fill opacity=.6]
		(Af)--(Bf) -- (Ff)-- (Ef)
		(Et)--(Ft) -- (Gt)-- (Ht)
		(Br)--(Cr) -- (Gr)-- (Fr);

	\node[red] at (cf){\clockxy{}};
	\node[red, rotate=2*\angle] at (cr){\clockyz{}};
	\node[red, rotate=.5*\angle] at (ct){\clockxz{}};
	\foreach \vertex in {A,...,C} {\node[vertex](v\vertex) at (\vertex) {};}
	\foreach \vertex in {E,...,H} {\node[vertex](v\vertex) at (\vertex) {};}
	\node[vertex,fill=gray](vD) at (D) {};       
	\draw[>=latex, ->] (vA)--(vE);
	\draw[>=latex, ->] (vA) -- (vB);
	\draw[>=latex, ->,gray] (vA) -- (vD);
	\draw[>=latex, ->] (vB) -- (vC); 
	\draw[>=latex, ->]	(vB) -- (vF);
	\draw[>=latex, ->]	(vC) -- (vG);
	\draw[>=latex, ->, gray]	(vD) -- (vC);
	\draw[>=latex, ->, gray]	(vD) -- (vH);
	\draw[>=latex, ->]	(vE) -- (vF);
	\draw[>=latex, ->]	(vE) -- (vH);
	\draw[>=latex, ->]	(vF) -- (vG);
	\draw[>=latex, ->]	(vH) -- (vG);

}

\newcommand\cubebdthree[8]{

	\path[shift={(0:#1)},shift={(90:#2)}, shift={(\angle:#3*\cosine)}]
		(0,0)  coordinate (A)  (1,0)     coordinate (B) %
		+(\angle:\cosine) coordinate (C)  (\angle:\cosine) coordinate (D)
		(0,1)  coordinate (E) +(\angle:\cosine) coordinate (H)   %
		(1,1)  coordinate (F) +(\angle:\cosine) coordinate (G)  ;
	\draw[dotted]
		(D) -- (H) (D) -- (C) (D) -- (A);
	\node[xyfv#8bd3](cback) at (barycentric cs:D=1,C=1,G=1,H=1){$#4$};
	\node[yzfv#8bd3](cleft) at (barycentric cs:A=1,D=1,H=1,E=1)   {$#4$};
	\node[xzfv#8bd3](cbottom) at (barycentric cs:A=1,B=1,C=1,D=1){$#4$};

	\draw[dotted]
		(A) -- (E) (B) -- (F) (C) -- (G)  %
		(A) -- (B) -- (C) %
		(E) -- (F) -- (G) -- (H) -- (E); %

		\def\gap{.99}
		\node[inner sep=0,minimum size=0](cfront) at (barycentric cs:A=1,B=1,F=1,E=1){};
		\path(cfront) -- coordinate[pos=\gap] (Af) (A.0);
		\path(cfront) -- coordinate[pos=\gap] (Bf) (B);
		\path(cfront) -- coordinate[pos=\gap] (Ff) (F);
		\path(cfront) -- coordinate[pos=\gap] (Ef) (E);
		\draw(F)--(Ff);
		\node[inner sep=0,minimum size=0](cright) at (barycentric cs:B=1,C=1,G=1,F=1){};
		\path(cright) -- coordinate[pos=\gap] (Br) (B);
		\path(cright) -- coordinate[pos=\gap] (Cr) (C);
		\path(cright) -- coordinate[pos=\gap] (Gr) (G);
		\path(cright) -- coordinate[pos=\gap] (Fr) (F);
		\node[inner sep=0,minimum size=0](ctop) at (barycentric cs:E=1,F=1,G=1,H=1){};
		\path(ctop) -- coordinate[pos=\gap] (Et) (E);
		\path(ctop) -- coordinate[pos=\gap] (Gt) (G);
		\path(ctop) -- coordinate[pos=\gap] (Ft) (F);
		\path(ctop) -- coordinate[pos=\gap] (Ht) (H);

	\draw (cfront)--(cback) (ctop)--(cbottom) (cleft)--(cright);
	\node(cubecenter)[fvertex] at ($(ctop)!.5!(cbottom)$){$#5$};

	\path[fill=white,fill opacity=0.4]
		(Af)--(Bf) -- (Ff)-- (Ef)
		(Et)--(Ft) -- (Gt)-- (Ht)
		(Br)--(Cr) -- (Gr)-- (Fr)
		;

	\node[xyfv#8bd3](vcfront) at (cfront){$#4$};
	\node[yzfv#8bd3](vcright) at (cright){$#4$};
	\node[xzfv#8bd3](vctop) at (ctop)  {$#4$};

		\ifthenelse{\equal{#5}{+}}
		{
			\node at (cubecenter)[fcil, gciliatednode={0 }{3}]{\phantom{$+$}};
			\node at (cubecenter)[fcil, gciliatednode={90}{3}]{\phantom{$+$}};
			\node at (cubecenter)[fcil, gciliatednode={0+\angle}{3}]{\phantom{$+$}};
		}
		{
			\node at (vcfront)[fcil, ciliatednode={\angle }{3}]{\phantom{$+$}};
			\node at (cleft)[fcil, ciliatednode={0}{3}]{\phantom{$+$}};
			\node at (cbottom)[fcil, ciliatednode={90}{3}]{\phantom{$+$}};
		}

}

\newcommand\cubebdthreex[8]{
	\phantom{ \cubebdthree{#1}{#2}{#3}{#4}{#5}{#6}{#7}{#8}}
	\draw(cubecenter) -- ++(0: 	 -.9);
	\draw(cubecenter) -- ++(0: 	  .9);
	\draw(cubecenter) -- ++(90: 	  .9);
	\draw(cubecenter) -- ++(90:	 -.9);
	\draw(cubecenter) -- ++(\angle:-.9*\cosine);
	\draw(cubecenter) -- ++(\angle: .9*\cosine);
	\cubebdthree{#1}{#2}{#3}{#4}{#5}{#6}{#7}{#8}
	\draw(vcright.east) -- ++(0: 	  .3);
	\draw(vcfront.{180+\angle}) -- ++(180+\angle: 	  .3);
	\draw(vctop.north) -- ++(90: 	  .3);
		\ifthenelse{\equal{#5}{+}}
		{
			\node at (cubecenter)[fcil, gciliatednode={0 }{3}]{\phantom{$+$}};
			\node at (cubecenter)[fcil, gciliatednode={90}{3}]{\phantom{$+$}};
			\node at (cubecenter)[fcil, gciliatednode={0+\angle}{3}]{\phantom{$+$}};
		}
		{
		}
}

\newcommand\torusbdthree[8]{

	\path[shift={(0:#1)},shift={(90:#2)}, shift={(\angle:#3*\cosine)}]
		(0,0)  coordinate (A)  (1,0)     coordinate (B) %
		+(\angle:\cosine) coordinate (C)  (\angle:\cosine) coordinate (D)
		(0,1)  coordinate (E) +(\angle:\cosine) coordinate (H)   %
		(1,1)  coordinate (F) +(\angle:\cosine) coordinate (G)  ;
	\draw[dotted]
		(D) -- (H) (D) -- (C) (D) -- (A);
	\node[fcil](cback) at (barycentric cs:D=1,C=1,G=1,H=1){};
	\node[fcil](cleft) at (barycentric cs:A=1,D=1,H=1,E=1){};
	\node[fcil](cbottom) at (barycentric cs:A=1,B=1,C=1,D=1){};

	\draw[dotted]
		(A) -- (E) (B) -- (F) (C) -- (G)  %
		(A) -- (B) -- (C) %
		(E) -- (F) -- (G) -- (H) -- (E); %

		\def\gap{.99}
		\node[inner sep=0,minimum size=0](cfront) at (barycentric cs:A=1,B=1,F=1,E=1){};
		\path(cfront) -- coordinate[pos=\gap] (Af) (A.0);
		\path(cfront) -- coordinate[pos=\gap] (Bf) (B);
		\path(cfront) -- coordinate[pos=\gap] (Ff) (F);
		\path(cfront) -- coordinate[pos=\gap] (Ef) (E);
		\draw(F)--(Ff);
		\node[inner sep=0,minimum size=0](cright) at (barycentric cs:B=1,C=1,G=1,F=1){};
		\path(cright) -- coordinate[pos=\gap] (Br) (B);
		\path(cright) -- coordinate[pos=\gap] (Cr) (C);
		\path(cright) -- coordinate[pos=\gap] (Gr) (G);
		\path(cright) -- coordinate[pos=\gap] (Fr) (F);
		\node[inner sep=0,minimum size=0](ctop) at (barycentric cs:E=1,F=1,G=1,H=1){};
		\path(ctop) -- coordinate[pos=\gap] (Et) (E);
		\path(ctop) -- coordinate[pos=\gap] (Gt) (G);
		\path(ctop) -- coordinate[pos=\gap] (Ft) (F);
		\path(ctop) -- coordinate[pos=\gap] (Ht) (H);

	\draw (cfront)--(cback) (ctop)--(cbottom) (cleft)--(cright);
	\node(cubecenter)[fvertex] at ($(ctop)!.5!(cbottom)$){$#5$};

	\path[fill=white,fill opacity=0.4]
		(Af)--(Bf) -- (Ff)-- (Ef)
		(Et)--(Ft) -- (Gt)-- (Ht)
		(Br)--(Cr) -- (Gr)-- (Fr)
		;

	\node[torusxyfv#8bd3](vcfront) at (cfront){$#4$};
	\node[torusyzfv#8bd3](vcright) at (cright){$#4$};
	\node[torusxzfv#8bd3](vctop) at (ctop)  {$#4$};

		\ifthenelse{\equal{#5}{+}}
		{
			\node at (cubecenter)[fcil, gciliatednode={0 }{3}]{\phantom{$+$}};
			\node at (cubecenter)[fcil, gciliatednode={90}{3}]{\phantom{$+$}};
			\node at (cubecenter)[fcil, gciliatednode={0+\angle}{3}]{\phantom{$+$}};
		}
		{
			\node at (vcfront)[fcil, ciliatednode={\angle }{3}]{\phantom{$+$}};
			\node at (cright)[fcil, ciliatednode={0}{3}]{\phantom{$+$}};
			\node at (ctop)[fcil, ciliatednode={90}{3}]{\phantom{$+$}};
		}
}

\def\distdang{2mm}

\tikzset {
    xyfvnonebd3/.style={
    draw, minimum size=3mm, fill=white, inner sep = 0pt,
    append after command={
      \pgfextra{
        \draw
		   ([shift={(0:0)}]\tikzlastnode.center) -- ([shift={(0:-\distdang)}]\tikzlastnode.center)
			 node[halfedge,rotate=\angle]{}
		 ;
      }
    },
  }
}
\tikzset {
    xzfvnonebd3/.style={
    draw, minimum size=3mm, fill=white, inner sep = 0pt,
    append after command={
      \pgfextra{
        \draw
		   ([shift={(0:0)}]\tikzlastnode.center) -- ([shift={(0:\distdang)}]\tikzlastnode.center)
			 node[halfedge,rotate=90]{}
		 ;
      }
    },
  }
}
\tikzset {
    yzfvnonebd3/.style={
    draw, minimum size=3mm, fill=white, inner sep = 0pt,
    append after command={
      \pgfextra{
        \draw
		   ([shift={(0:0)}]\tikzlastnode.center) -- ([shift={(1*\angle:\distdang)}]\tikzlastnode.center)
			 node[halfedge,rotate=0]{}
		 ;
      }
    },
  }
}

\tikzset {
    xyfvhalfbd3/.style={
    draw, minimum size=3mm, fill=white, inner sep = 0pt,
    append after command={
      \pgfextra{
        \draw
		   ([shift={(0:0)}]\tikzlastnode.center) -- ([shift={(0:-\distdang)}]\tikzlastnode.center)
			 node[halfedge,rotate=90]{}
		 ;
      }
    },
  }
}
\tikzset {
    xzfvhalfbd3/.style={
    draw, minimum size=3mm, fill=white, inner sep = 0pt,
    append after command={
      \pgfextra{
        \draw
		   ([shift={(0:0)}]\tikzlastnode.center) -- ([shift={(0:\distdang)}]\tikzlastnode.center)
			 node[halfedge,rotate=1*\angle]{}
		 ;
      }
    },
  }
}
\tikzset {
    yzfvhalfbd3/.style={
    draw, minimum size=3mm, fill=white, inner sep = 0pt,
    append after command={
      \pgfextra{
        \draw
		   ([shift={(0:0)}]\tikzlastnode.center) -- ([shift={(1*\angle:\distdang)}]\tikzlastnode.center)
			 node[halfedge,rotate=90+0*\angle]{}
		 ;
      }
    },
  }
}
\tikzset {
    torusxyfvhalfbd3/.style={
    draw, minimum size=3mm, fill=white, inner sep = 0pt,
    append after command={
      \pgfextra{
        \draw
		   ([shift={(0:0)}]\tikzlastnode.center) -- ([shift={(0:-\distdang)}]\tikzlastnode.center)
			 node[halfedge,rotate=90]{}
		 ;
        \draw
		   ([shift={(0:0)}]\tikzlastnode.center) -- ([shift={(180+\angle:1.5*\distdang)}]\tikzlastnode.center)
			;
      }
    },
  }
}
\tikzset {
    torusxzfvhalfbd3/.style={
    draw, minimum size=3mm, fill=white, inner sep = 0pt,
    append after command={
      \pgfextra{
        \draw
		   ([shift={(0:0)}]\tikzlastnode.center) -- ([shift={(0:\distdang)}]\tikzlastnode.center)
			 node[halfedge,rotate=1*\angle]{}
		 ;
        \draw
		   ([shift={(0:0)}]\tikzlastnode.center) -- ([shift={(90:1.5*\distdang)}]\tikzlastnode.center)
			;
      }
    },
  }
}
\tikzset {
    torusyzfvhalfbd3/.style={
    draw, minimum size=3mm, fill=white, inner sep = 0pt,
    append after command={
      \pgfextra{
        \draw
		   ([shift={(0:0)}]\tikzlastnode.center) -- ([shift={(1*\angle:\distdang)}]\tikzlastnode.center)
			 node[halfedge,rotate=90+0*\angle]{}
		 ;
        \draw
		   ([shift={(0:0)}]\tikzlastnode.center) -- ([shift={(0:1.5*\distdang)}]\tikzlastnode.center)
			;
      }
    },
  }
}

\newcommand\cubebdtwo[8]{

	\path[shift={(0:#1)},shift={(90:#2)}, shift={(\angle:#3*\cosine)}]
		(0,0)  coordinate (A)  (1,0)     coordinate (B) %
		+(\angle:\cosine) coordinate (C)  (\angle:\cosine) coordinate (D)
		(0,1)  coordinate (E) +(\angle:\cosine) coordinate (H)   %
		(1,1)  coordinate (F) +(\angle:\cosine) coordinate (G)  ;
	\draw[dotted]
		(D) -- (H) (D) -- (C) (D) -- (A);
	\node[xyfv#8](cback) at (barycentric cs:D=1,C=1,G=1,H=1){$#4$};
	\node[yzfv#8](cleft) at (barycentric cs:A=1,D=1,H=1,E=1)   {$#4$};
	\node[xzfv#8](cbottom) at (barycentric cs:A=1,B=1,C=1,D=1){$#4$};

	\node[xev#7bd2] at ($(D)!.5!(C)$){$#5$};
	\node[yev#7bd2] at ($(D)!.5!(H)$){$#5$};
	\node[zev#7bd2] at ($(D)!.5!(A)$){$#5$} (A);
	
	\draw[dotted]
		(A) -- (E) (B) -- (F) (C) -- (G)  %
		(A) -- (B) -- (C) %
		(E) -- (F) -- (G) -- (H) -- (E); %

		\def\gap{.99}
		\node[inner sep=0,minimum size=0](cfront) at (barycentric cs:A=1,B=1,F=1,E=1){};
		\path(cfront) -- coordinate[pos=\gap] (Af) (A.0);
		\path(cfront) -- coordinate[pos=\gap] (Bf) (B);
		\path(cfront) -- coordinate[pos=\gap] (Ff) (F);
		\path(cfront) -- coordinate[pos=\gap] (Ef) (E);
		\draw(F)--(Ff);
		\node[inner sep=0,minimum size=0](cright) at (barycentric cs:B=1,C=1,G=1,F=1){};
		\path(cright) -- coordinate[pos=\gap] (Br) (B);
		\path(cright) -- coordinate[pos=\gap] (Cr) (C);
		\path(cright) -- coordinate[pos=\gap] (Gr) (G);
		\path(cright) -- coordinate[pos=\gap] (Fr) (F);
		\node[inner sep=0,minimum size=0](ctop) at (barycentric cs:E=1,F=1,G=1,H=1){};
		\path(ctop) -- coordinate[pos=\gap] (Et) (E);
		\path(ctop) -- coordinate[pos=\gap] (Gt) (G);
		\path(ctop) -- coordinate[pos=\gap] (Ft) (F);
		\path(ctop) -- coordinate[pos=\gap] (Ht) (H);
	\path[fill=white,fill opacity=.5]
		(Af)--(Bf) -- (Ff)-- (Ef)
		(Et)--(Ft) -- (Gt)-- (Ht)
		(Br)--(Cr) -- (Gr)-- (Fr);

	\node[xyfv#8] at (cfront){$#4$};
	\node[yzfv#8] at (cright){$#4$};
	\node[xzfv#8] at (ctop)  {$#4$};

	\node[xev#7bd2] at ($(E)!.5!(F)$){$#5$};
	\node[yev#7bd2] at ($(E)!.5!(A)$){$#5$};
	\node[zev#7bd2] at ($(E)!.5!(H)$){$#5$};
	\node[xev#7bd2] at ($(G)!.5!(H)$){$#5$};
	\node[yev#7bd2] at ($(G)!.5!(C)$){$#5$};
	\node[zev#7bd2] at ($(G)!.5!(F)$){$#5$};
	\node[xev#7bd2] at ($(B)!.5!(A)$){$#5$};
	\node[yev#7bd2] at ($(B)!.5!(F)$){$#5$};
	\node[zev#7bd2] at ($(B)!.5!(C)$){$#5$};
			\node at (cfront)[fcil, ciliatednode={180}{3}]{\phantom{$+$}};
			\node at (cfront)[fcil, ciliatednode={90}{3}]{\phantom{$+$}};
			\node at (cback)[fcil, gciliatednode={180}{3}]{\phantom{$+$}};
			\node at (cback)[fcil, gciliatednode={90}{3}]{\phantom{$+$}};
			\node at (cright)[fcil, ciliatednode={180+\angle}{3}]{\phantom{$+$}};
			\node at (cright)[fcil, ciliatednode={90}{3}]{\phantom{$+$}};
			\node at (cleft)[fcil, gciliatednode={180+\angle}{3}]{\phantom{$+$}};
			\node at (cleft)[fcil, gciliatednode={90}{3}]{\phantom{$+$}};
			\node at (ctop)[fcil, ciliatednode={180}{3}]{\phantom{$+$}};
			\node at (ctop)[fcil, ciliatednode={\angle}{3}]{\phantom{$+$}};
			\node at (cbottom)[fcil, gciliatednode={180}{3}]{\phantom{$+$}};
			\node at (cbottom)[fcil, gciliatednode={\angle}{3}]{\phantom{$+$}};
}

\newcommand\torusbdtwo[8]{

	\path[shift={(0:#1)},shift={(90:#2)}, shift={(\angle:#3*\cosine)}]
		(0,0)  coordinate (A)  (1,0)     coordinate (B) %
		+(\angle:\cosine) coordinate (C)  (\angle:\cosine) coordinate (D)
		(0,1)  coordinate (E) +(\angle:\cosine) coordinate (H)   %
		(1,1)  coordinate (F) +(\angle:\cosine) coordinate (G)  ;
	\draw[dotted]
		(D) -- (H) (D) -- (C) (D) -- (A);

	\draw[dotted]
		(A) -- (E) (B) -- (F) (C) -- (G)  %
		(A) -- (B) -- (C) %
		(E) -- (F) -- (G) -- (H) -- (E); %

		\def\gap{.99}
		\node[inner sep=0,minimum size=0](cfront) at (barycentric cs:A=1,B=1,F=1,E=1){};
		\path(cfront) -- coordinate[pos=\gap] (Af) (A.0);
		\path(cfront) -- coordinate[pos=\gap] (Bf) (B);
		\path(cfront) -- coordinate[pos=\gap] (Ff) (F);
		\path(cfront) -- coordinate[pos=\gap] (Ef) (E);
		\draw(F)--(Ff);
		\node[inner sep=0,minimum size=0](cright) at (barycentric cs:B=1,C=1,G=1,F=1){};
		\path(cright) -- coordinate[pos=\gap] (Br) (B);
		\path(cright) -- coordinate[pos=\gap] (Cr) (C);
		\path(cright) -- coordinate[pos=\gap] (Gr) (G);
		\path(cright) -- coordinate[pos=\gap] (Fr) (F);
		\node[inner sep=0,minimum size=0](ctop) at (barycentric cs:E=1,F=1,G=1,H=1){};
		\path(ctop) -- coordinate[pos=\gap] (Et) (E);
		\path(ctop) -- coordinate[pos=\gap] (Gt) (G);
		\path(ctop) -- coordinate[pos=\gap] (Ft) (F);
		\path(ctop) -- coordinate[pos=\gap] (Ht) (H);
	\path[fill=white,fill opacity=.5]
		(Af)--(Bf) -- (Ff)-- (Ef)
		(Et)--(Ft) -- (Gt)-- (Ht)
		(Br)--(Cr) -- (Gr)-- (Fr);

	\node[xyfv#8] at (cfront){$#4$};
	\node[yzfv#8] at (cright){$#4$};
	\node[xzfv#8] at (ctop)  {$#4$};

	\node[torusxev#7bd2] at ($(E)!.5!(F)$){$#5$};
	\node[toruszev#7bd2] at ($(G)!.5!(F)$){$#5$};
	\node[torusyev#7bd2] at ($(B)!.5!(F)$){$#5$};
			\node at (cfront)[fcil, ciliatednode={180}{3}]{\phantom{$+$}};
			\node at (cfront)[fcil, ciliatednode={90}{3}]{\phantom{$+$}};
			\node at (cright)[fcil, ciliatednode={180+\angle}{3}]{\phantom{$+$}};
			\node at (cright)[fcil, ciliatednode={90}{3}]{\phantom{$+$}};
			\node at (ctop)[fcil, ciliatednode={180}{3}]{\phantom{$+$}};
			\node at (ctop)[fcil, ciliatednode={\angle}{3}]{\phantom{$+$}};
}

\def\distdang{2mm}

\tikzset {
    zevnonebd2/.style={ draw, minimum size=3mm, fill=white, inner sep = 0pt, }
}
\tikzset {
    yevnonebd2/.style={ draw, minimum size=3mm, fill=white, inner sep = 0pt, }
}
\tikzset {
    xevnonebd2/.style={ draw, minimum size=3mm, fill=white, inner sep = 0pt, }
}

\tikzset {
    xevhalfbd2/.style={
    draw, minimum size=3mm, fill=white, inner sep = 0pt,
    append after command={
      \pgfextra{
		 \draw ([shift={(0,0)}]\tikzlastnode.center) -- ([shift={(0:\distdang)}]\tikzlastnode.center)
		 node[halfedge,rotate=90]{}
		 ;
      }
    },
  }
}
\tikzset {
    yevhalfbd2/.style={
    draw, minimum size=3mm, fill=white, inner sep = 0pt,
    append after command={
      \pgfextra{
        \draw ([shift={(0mm, 0mm)}]\tikzlastnode.center) -- ([shift={(90:\distdang)}]\tikzlastnode.center)
		 node[halfedge,rotate=0]{}
		 ;
      }
    },
  }
}
\tikzset {
    zevhalfbd2/.style={
    draw, minimum size=3mm, fill=white, inner sep = 0pt,
    append after command={
      \pgfextra{
        \draw ([shift={(0mm, 0mm)}]\tikzlastnode.center) -- ([shift={(\angle:\distdang)}]\tikzlastnode.center)
		 node[halfedge,rotate=0]{}
		 ;
      }
    },
  }
}
\tikzset {
    torusxevnonebd2/.style={
    draw, minimum size=3mm, fill=white, inner sep = 0pt,
    append after command={
      \pgfextra{
		 \draw
		 ([shift={(0,0)}]\tikzlastnode.center) -- ([shift={(90:\distdang)}]\tikzlastnode.center)
		 ([shift={(0,0)}]\tikzlastnode.center) -- ([shift={(180+\angle:\distdang)}]\tikzlastnode.center)
		 ;
      }
    },
  }
}
\tikzset {
    torusyevnonebd2/.style={
    draw, minimum size=3mm, fill=white, inner sep = 0pt,
    append after command={
      \pgfextra{
		 \draw
		 ([shift={(0,0)}]\tikzlastnode.center) -- ([shift={(0:\distdang)}]\tikzlastnode.center)
		 ([shift={(0,0)}]\tikzlastnode.center) -- ([shift={(180+\angle:\distdang)}]\tikzlastnode.center)
		 ;
      }
    },
  }
}
\tikzset {
    toruszevnonebd2/.style={
    draw, minimum size=3mm, fill=white, inner sep = 0pt,
    append after command={
      \pgfextra{
		 \draw
		 ([shift={(0,0)}]\tikzlastnode.center) -- ([shift={(90:\distdang)}]\tikzlastnode.center)
		 ([shift={(0,0)}]\tikzlastnode.center) -- ([shift={(0:\distdang)}]\tikzlastnode.center)
		 ;
      }
    },
  }
}
\tikzset {
    torusxevhalfbd2/.style={
    draw, minimum size=3mm, fill=white, inner sep = 0pt,
    append after command={
      \pgfextra{
		 \draw
		 ([shift={(0,0)}]\tikzlastnode.center) -- ([shift={(0:\distdang)}]\tikzlastnode.center)
		 node[halfedge,rotate=90]{}
		 ;
		 \draw
		 ([shift={(0,0)}]\tikzlastnode.center) -- ([shift={(90:\distdang)}]\tikzlastnode.center)
		 ([shift={(0,0)}]\tikzlastnode.center) -- ([shift={(180+\angle:\distdang)}]\tikzlastnode.center)
		 ;
      }
    },
  }
}
\tikzset {
    torusyevhalfbd2/.style={
    draw, minimum size=3mm, fill=white, inner sep = 0pt,
    append after command={
      \pgfextra{
        \draw ([shift={(0mm, 0mm)}]\tikzlastnode.center) -- ([shift={(90:\distdang)}]\tikzlastnode.center)
		 node[halfedge,rotate=0]{}
		 ;
		 \draw
		 ([shift={(0,0)}]\tikzlastnode.center) -- ([shift={(0:\distdang)}]\tikzlastnode.center)
		 ([shift={(0,0)}]\tikzlastnode.center) -- ([shift={(180+\angle:\distdang)}]\tikzlastnode.center)
		 ;
      }
    },
  }
}
\tikzset {
    toruszevhalfbd2/.style={
    draw, minimum size=3mm, fill=white, inner sep = 0pt,
    append after command={
      \pgfextra{
        \draw ([shift={(0mm, 0mm)}]\tikzlastnode.center) -- ([shift={(\angle:\distdang)}]\tikzlastnode.center)
		 node[halfedge,rotate=0]{}
		 ;
		 \draw
		 ([shift={(0,0)}]\tikzlastnode.center) -- ([shift={(90:\distdang)}]\tikzlastnode.center)
		 ([shift={(0,0)}]\tikzlastnode.center) -- ([shift={(0:\distdang)}]\tikzlastnode.center)
		 ;
      }
    },
  }
}

\tikzset {
    xevkappabd2/.style={
    draw, minimum size=3mm, fill=white, inner sep = 0pt,
    append after command={
      \pgfextra{
		 \draw ([shift={(0,0)}]\tikzlastnode.center) -- ([shift={(0:\distdang)}]\tikzlastnode.center)
		 node[kappa]{$\kappa$}
		 ;
      }
    },
  }
}
\tikzset {
    yevkappabd2/.style={
    draw, minimum size=3mm, fill=white, inner sep = 0pt,
    append after command={
      \pgfextra{
        \draw ([shift={(0mm, 0mm)}]\tikzlastnode.center) -- ([shift={(90:\distdang)}]\tikzlastnode.center)
		 node[kappa]{$\kappa$}
		 ;
      }
    },
  }
}
\tikzset {
    zevkappabd2/.style={
    draw, minimum size=3mm, fill=white, inner sep = 0pt,
    append after command={
      \pgfextra{
        \draw ([shift={(0mm, 0mm)}]\tikzlastnode.center) -- ([shift={(\angle:\distdang)}]\tikzlastnode.center)
		 node[kappa]{$\kappa$}
		 ;
      }
    },
  }
}

\tikzset {
    xyfvnone/.style={
    draw, minimum size=3mm, fill=white, inner sep = 0pt,
    append after command={
      \pgfextra{
        \draw
		   ([shift={(0:0)}]\tikzlastnode.center) -- ([shift={(180:.5)}]\tikzlastnode.east)
		   ([shift={(0:0)}]\tikzlastnode.center) -- ([shift={(0:.5)}]\tikzlastnode.west)
		   ([shift={(0:0)}]\tikzlastnode.center) -- ([shift={(90:.5)}]\tikzlastnode.south)
		   ([shift={(0:0)}]\tikzlastnode.center) -- ([shift={(-90:.5)}]\tikzlastnode.north)
		 ;
      }
    },
  }
}
\tikzset {
    xzfvnone/.style={
    draw, minimum size=3mm, fill=white, inner sep = 0pt,
    append after command={
      \pgfextra{
        \draw
		   ([shift={(0:0)}]\tikzlastnode.center) -- ([shift={(180:.5)}]\tikzlastnode.east)
		   ([shift={(0:0)}]\tikzlastnode.center) -- ([shift={(0:.5)}]\tikzlastnode.west)
			([shift={(0:0)}]\tikzlastnode.center) -- ([shift={(\angle:.5*\cosine)}]\tikzlastnode.{180+\angle})
		   ([shift={(0:0)}]\tikzlastnode.center) -- ([shift={(180+\angle:.5*\cosine)}]\tikzlastnode.\angle)
		 ;
      }
    },
  }
}
\tikzset {
    yzfvnone/.style={
    draw, minimum size=3mm, fill=white, inner sep = 0pt,
    append after command={
      \pgfextra{
        \draw
			([shift={(0:0)}]\tikzlastnode.center) -- ([shift={(\angle:.5*\cosine)}]\tikzlastnode.{180+\angle})
		   ([shift={(0:0)}]\tikzlastnode.center) -- ([shift={(180+\angle:.5*\cosine)}]\tikzlastnode.\angle)
		   ([shift={(0:0)}]\tikzlastnode.center) -- ([shift={(90:.5)}]\tikzlastnode.south)
		   ([shift={(0:0)}]\tikzlastnode.center) -- ([shift={(-90:.5)}]\tikzlastnode.north)
		 ;
      }
    },
  }
}

\tikzset {
    xyfvhalf/.style={
    draw, minimum size=3mm, fill=white, inner sep = 0pt,
    append after command={
      \pgfextra{
        \draw
		   ([shift={(0:0)}]\tikzlastnode.center) -- ([shift={(180:.5)}]\tikzlastnode.east)
		   ([shift={(0:0)}]\tikzlastnode.center) -- ([shift={(0:.5)}]\tikzlastnode.west)
		   ([shift={(0:0)}]\tikzlastnode.center) -- ([shift={(90:.5)}]\tikzlastnode.south)
		   ([shift={(0:0)}]\tikzlastnode.center) -- ([shift={(-90:.5)}]\tikzlastnode.north)
		   ([shift={(0:0)}]\tikzlastnode.center) -- ([shift={(45:\distdang)}]\tikzlastnode.center)
			 node[halfedge,rotate=-45]{}
		 ;
      }
    },
  }
}
\tikzset {
    xzfvhalf/.style={
    draw, minimum size=3mm, fill=white, inner sep = 0pt,
    append after command={
      \pgfextra{
        \draw
		   ([shift={(0:0)}]\tikzlastnode.center) -- ([shift={(180:.5)}]\tikzlastnode.east)
		   ([shift={(0:0)}]\tikzlastnode.center) -- ([shift={(0:.5)}]\tikzlastnode.west)
			([shift={(0:0)}]\tikzlastnode.center) -- ([shift={(\angle:.5)}]\tikzlastnode.{180+\angle})
		   ([shift={(0:0)}]\tikzlastnode.center) -- ([shift={(180+\angle:.5)}]\tikzlastnode.\angle)
		   ([shift={(0:0)}]\tikzlastnode.center) -- ([shift={(\angle+90:\distdang)}]\tikzlastnode.center)
			 node[halfedge,rotate=0*\angle]{}
		 ;
      }
    },
  }
}
\tikzset {
    yzfvhalf/.style={
    draw, minimum size=3mm, fill=white, inner sep = 0pt,
    append after command={
      \pgfextra{
        \draw
		  ([shift={(0:0)}]\tikzlastnode.center) -- ([shift={(\angle:.5)}]\tikzlastnode.{180+\angle})
		   ([shift={(0:0)}]\tikzlastnode.center) -- ([shift={(180+\angle:.5)}]\tikzlastnode.\angle)
		   ([shift={(0:0)}]\tikzlastnode.center) -- ([shift={(90:.5)}]\tikzlastnode.south)
		   ([shift={(0:0)}]\tikzlastnode.center) -- ([shift={(-90:.5)}]\tikzlastnode.north)
		   ([shift={(0:0)}]\tikzlastnode.center) -- ([shift={(2*\angle:\distdang)}]\tikzlastnode.center)
			 node[halfedge,rotate=90+1*\angle]{}
		 ;
      }
    },
  }
}

\newcommand\cubebdone[8]{

	\path[shift={(0:#1)},shift={(90:#2)}, shift={(\angle:#3*\cosine)}]
		(0,0)  coordinate (A)  (1,0)     coordinate (B) %
		+(\angle:\cosine) coordinate (C)  (\angle:\cosine) coordinate (D)
		(0,1)  coordinate (E) +(\angle:\cosine) coordinate (H)   %
		(1,1)  coordinate (F) +(\angle:\cosine) coordinate (G)  ;
	\draw[draw=black]
		(D) -- (H) (D) -- (C) (D) -- (A);

	\node[xev#7](evDx) at ($(D)!.5!(C)$){$#5$};
	\node[yev#7](evDy) at ($(D)!.5!(H)$){$#5$};
	\node[zev#7](evDz) at ($(D)!.5!(A)$){$#5$} (A);
	
	\draw[draw=black]
		(A) -- (E) (B) -- (F) (C) -- (G)  %
		(A) -- (B) -- (C) %
		(E) -- (F) -- (G) -- (H) -- (E); %

		\def\gap{.99}
		\node[inner sep=0,minimum size=0](cf) at (barycentric cs:A=1,B=1,F=1,E=1){};
		\path(cf) -- coordinate[pos=\gap] (Af) (A.0);
		\path(cf) -- coordinate[pos=\gap] (Bf) (B);
		\path(cf) -- coordinate[pos=\gap] (Ff) (F);
		\path(cf) -- coordinate[pos=\gap] (Ef) (E);
		\draw(F)--(Ff);
		\node[inner sep=0,minimum size=0](cr) at (barycentric cs:B=1,C=1,G=1,F=1){};
		\path(cr) -- coordinate[pos=\gap] (Br) (B);
		\path(cr) -- coordinate[pos=\gap] (Cr) (C);
		\path(cr) -- coordinate[pos=\gap] (Gr) (G);
		\path(cr) -- coordinate[pos=\gap] (Fr) (F);
		\node[inner sep=0,minimum size=0](ct) at (barycentric cs:E=1,F=1,G=1,H=1){};
		\path(ct) -- coordinate[pos=\gap] (Et) (E);
		\path(ct) -- coordinate[pos=\gap] (Gt) (G);
		\path(ct) -- coordinate[pos=\gap] (Ft) (F);
		\path(ct) -- coordinate[pos=\gap] (Ht) (H);
	\path[fill=white,fill opacity=.6]
		(Af)--(Bf) -- (Ff)-- (Ef)
		(Et)--(Ft) -- (Gt)-- (Ht)
		(Br)--(Cr) -- (Gr)-- (Fr);
	\foreach \vertex in {A,...,C} {\node(v\vertex)[sitev#6] at (\vertex) {$#4$};}
	\foreach \vertex in {E,...,H} {\node(v\vertex)[sitev#6] at (\vertex) {$#4$};}
	\node(vD)[sitev#6,draw=gray,text=gray] at (D) {$#4$};       

	\node[xev#7](evEx) at ($(E)!.5!(F)$){$#5$};
	\node[yev#7](evEy) at ($(E)!.5!(A)$){$#5$};
	\node[zev#7](evEz) at ($(E)!.5!(H)$){$#5$};
	\node[xev#7](evGx) at ($(G)!.5!(H)$){$#5$};
	\node[yev#7](evGy) at ($(G)!.5!(C)$){$#5$};
	\node[zev#7](evGz) at ($(G)!.5!(F)$){$#5$};
	\node[xev#7](evBx) at ($(B)!.5!(A)$){$#5$};
	\node[yev#7](evBy) at ($(B)!.5!(F)$){$#5$};
	\node[zev#7](evBz) at ($(B)!.5!(C)$){$#5$};
		\ifthenelse{\equal{#4}{+}}
		{
			\node at (vA)[fcil, ciliatednode={0}{3}]{\phantom{$+$}};
			\node at (vA)[fcil, ciliatednode={90}{3}]{\phantom{$+$}};
			\node at (vA)[fcil, gciliatednode={\angle}{3}]{\phantom{$+$}};
			\node at (vB)[fcil, ciliatednode={\angle}{3}]{\phantom{$+$}};
			\node at (vB)[fcil, ciliatednode={90}{3}]{\phantom{$+$}};
			\node at (vC)[fcil, ciliatednode={90}{3}]{\phantom{$+$}};
			\node at (vD)[fcil, gciliatednode={0}{3}]{\phantom{$+$}};
			\node at (vD)[fcil, gciliatednode={90}{3}]{\phantom{$+$}};
			\node at (vE)[fcil, ciliatednode={0}{3}]{\phantom{$+$}};
			\node at (vE)[fcil, ciliatednode={\angle}{3}]{\phantom{$+$}};
			\node at (vF)[fcil, ciliatednode={\angle}{3}]{\phantom{$+$}};
			\node at (vH)[fcil, ciliatednode={0}{3}]{\phantom{$+$}};
		}
		{
			\node at (evDx)[fcil, gciliatednode={0}{3}]{\phantom{$+$}};
			\node at (evDy)[fcil, gciliatednode={90}{3}]{\phantom{$+$}};
			\node at (evDz)[fcil, gciliatednode={\angle}{3}]{\phantom{$+$}};
			\node at (evEx)[fcil,  ciliatednode={0}{3}]{\phantom{$+$}};
			\node at (evEy)[fcil,  ciliatednode={90}{3}]{\phantom{$+$}};
			\node at (evEz)[fcil,  ciliatednode={\angle}{3}]{\phantom{$+$}};
			\node at (evGx)[fcil,  ciliatednode={0}{3}]{\phantom{$+$}};
			\node at (evGy)[fcil,  ciliatednode={90}{3}]{\phantom{$+$}};
			\node at (evGz)[fcil,  ciliatednode={\angle}{3}]{\phantom{$+$}};
			\node at (evBx)[fcil,  ciliatednode={0}{3}]{\phantom{$+$}};
			\node at (evBy)[fcil,  ciliatednode={90}{3}]{\phantom{$+$}};
			\node at (evBz)[fcil,  ciliatednode={\angle}{3}]{\phantom{$+$}};
		}
}

\newcommand\torusbdone[8]{

	\path[shift={(0:#1)},shift={(90:#2)}, shift={(\angle:#3*\cosine)}]
		(0,0)  coordinate (A)  (1,0)     coordinate (B) %
		+(\angle:\cosine) coordinate (C)  (\angle:\cosine) coordinate (D)
		(0,1)  coordinate (E) +(\angle:\cosine) coordinate (H)   %
		(1,1)  coordinate (F) +(\angle:\cosine) coordinate (G)  ;
	\draw[dotted]
		(D) -- (H) (D) -- (C) (D) -- (A) 	%
		(A) -- (E) (B) -- (F) (C) -- (G)  	%
		(A) -- (B) -- (C) 						%
		(E) -- (F) -- (G) -- (H) -- (E) 		%
		;
	\draw[draw=black]
		(F)edge++(180	 	  :.9)
		(F)edge++(-90		  :.9)
		(F)edge++(    \angle:.9*\cosine)
		(F)edge++(  0	 	  :.4)
		(F)edge++( 90		  :.4)
		(F)edge++(180+\angle:.4*\cosine)
		; 
	\node(vF)[sitev#6] at (F) {$#4$};

	\node[yev#7](evFy) at ($(F)!.5!(B)$){$#5$};
	\node[xev#7](evFx) at ($(F)!.5!(E)$){$#5$};
	\node[zev#7](evFz) at ($(F)!.5!(G)$){$#5$};

		\ifthenelse{\equal{#4}{+}}
		{
			\node at (vF)[fcil, ciliatednode={0}{3}]{\phantom{$+$}};
			\node at (vF)[fcil, ciliatednode={90}{3}]{\phantom{$+$}};
			\node at (vF)[fcil, ciliatednode={\angle}{3}]{\phantom{$+$}};
		}
		{
			\node at (evFx)[fcil, ciliatednode={0}{3}]{\phantom{$+$}};
			\node at (evFy)[fcil, ciliatednode={90}{3}]{\phantom{$+$}};
			\node at (evFz)[fcil, ciliatednode={\angle}{3}]{\phantom{$+$}};
		}
}

\tikzset{sitevnone/.style={
	draw, minimum size=3mm, fill=white, inner sep = 0pt},
}
\tikzset {
	sitevhalf/.style={
	draw, minimum size=3mm, fill=white, inner sep = 0pt,
    append after command={
      \pgfextra{
		 \draw ([shift={(0,0)}]\tikzlastnode.center) -- ([shift={(135:\distdang)}]\tikzlastnode.center)
		 node[halfedge,rotate=0]{}
		 ;
      }
    },
  }
}

\tikzset {
    xevnone/.style={ draw, minimum size=3mm, fill=white, inner sep = 0pt, }
}
\tikzset {
    yevnone/.style={ draw, minimum size=3mm, fill=white, inner sep = 0pt, }
}
\tikzset {
    zevnone/.style={ draw, minimum size=3mm, fill=white, inner sep = 0pt, }
}

\tikzset {
    xevhalf/.style={
    draw, minimum size=3mm, fill=white, inner sep = 0pt,
    append after command={
      \pgfextra{
		 \draw ([shift={(0,0)}]\tikzlastnode.center) -- ([shift={(1*\angle:\distdang)}]\tikzlastnode.center)
		 node[halfedge,rotate=0]{}
		 ;
      }
    },
  }
}
\tikzset {
    yevhalf/.style={
    draw, minimum size=3mm, fill=white, inner sep = 0pt,
    append after command={
      \pgfextra{
        \draw ([shift={(0mm, 0mm)}]\tikzlastnode.center) -- ([shift={(0:\distdang)}]\tikzlastnode.center)
		 node[halfedge,rotate=90]{}
		 ;
      }
    },
  }
}
\tikzset {
    zevhalf/.style={
    draw, minimum size=3mm, fill=white, inner sep = 0pt,
    append after command={
      \pgfextra{
        \draw ([shift={(0mm, 0mm)}]\tikzlastnode.center) -- ([shift={(180:\distdang)}]\tikzlastnode.center)
		 node[halfedge,rotate=\angle]{}
		 ;
      }
    },
  }
}

\tikzset {
    xevkappa/.style={
    draw, minimum size=3mm, fill=white, inner sep = 0pt,
    append after command={
      \pgfextra{
		 \draw ([shift={(0,0)}]\tikzlastnode.center) -- ([shift={(1*\angle:\distdang)}]\tikzlastnode.center)
		 node[kappa]{$\kappa$}
		 ;
      }
    },
  }
}
\tikzset {
    yevkappa/.style={
    draw, minimum size=3mm, fill=white, inner sep = 0pt,
    append after command={
      \pgfextra{
        \draw ([shift={(0mm, 0mm)}]\tikzlastnode.center) -- ([shift={(0:\distdang)}]\tikzlastnode.center)
		 node[kappa]{$\kappa$}
		 ;
      }
    },
  }
}
\tikzset {
    zevkappa/.style={
    draw, minimum size=3mm, fill=white, inner sep = 0pt,
    append after command={
      \pgfextra{
        \draw ([shift={(0mm, 0mm)}]\tikzlastnode.center) -- ([shift={(180:\distdang)}]\tikzlastnode.center)
		 node[kappa]{$\kappa$}
		 ;
      }
    },
  }
}

\tikzset {
    xevkappahat/.style={
    draw, minimum size=3mm, fill=white, inner sep = 0pt,
    append after command={
      \pgfextra{
		 \draw ([shift={(0,0)}]\tikzlastnode.center) -- ([shift={(1*\angle:\distdang)}]\tikzlastnode.center)
		 node[kappa]{$\widehat{\kappa}$}
		 ;
      }
    },
  }
}
\tikzset {
    yevkappahat/.style={
    draw, minimum size=3mm, fill=white, inner sep = 0pt,
    append after command={
      \pgfextra{
        \draw ([shift={(0mm, 0mm)}]\tikzlastnode.center) -- ([shift={(0:\distdang)}]\tikzlastnode.center)
		  node[kappa]{$\widehat{\kappa}$}
		 ;
      }
    },
  }
}
\tikzset {
    zevkappahat/.style={
    draw, minimum size=3mm, fill=white, inner sep = 0pt,
    append after command={
      \pgfextra{
        \draw ([shift={(0mm, 0mm)}]\tikzlastnode.center) -- ([shift={(180:\distdang)}]\tikzlastnode.center)
		  node[kappa]{$\widehat{\kappa}$}
		 ;
      }
    },
  }
}

\newcommand\cubecoord[3]{
	\path[shift={(0:#1)},shift={(90:#2)}, shift={(\angle:#3*\cosine)}]
		(0,0)  coordinate (A)  (1,0)     coordinate (B) %
		+(\angle:\cosine) coordinate (C)  (\angle:\cosine) coordinate (D)
		(0,1)  coordinate (E) +(\angle:\cosine) coordinate (H)   %
		(1,1)  coordinate (F) +(\angle:\cosine) coordinate (G)  ;
	\foreach \vertex in {A,...,H}
		\node(d\vertex)[dvertex] at (\vertex){};
}

\newcommand\topdang[3]{
	\cubecoord{#1}{#2}{#3}
	\foreach \vertex in {E,...,H}{
		\draw (d\vertex)	--++(90		:.3) ;
	}
}
\newcommand\topextend[4]{
	\cubecoord{#1}{#2}{#3}
	\foreach \vertex in {E,...,H}{
		\draw[\extendtype,\extendclr] (d\vertex)	--++(90		:\extend) ;
		\ifthenelse{\equal{#4}{+}} { \node at (d\vertex)[fcil, extendciliatednode={90}{3}]{\phantom{$+$}}; }
	}
}

\newcommand\topfaceextend[4]{
	\cubecoord{#1}{#2}{#3}
	\node[dvertex](ctop) at (barycentric cs:E=1,F=1,G=1,H=1){};
	\draw[\extendtype,\extendclr] (ctop)--++(90:\extend) ;
	\ifthenelse{\equal{#4}{+}} {\node at (ctop)[fcil, extendciliatednode={90}{3}]{\phantom{$+$}}; }
}

\newcommand\bottomextend[4]{
	\cubecoord{#1}{#2}{#3}
	\foreach \vertex in {A,...,D}{
		\draw[\extendtype,\extendclr] (d\vertex)	--++(-90:\extend) ;
		\ifthenelse{\equal{#4}{+}} { \node at (d\vertex)[fcil, extendciliatednode={-90}{3}]{\phantom{$+$}}; }
	}
}
\newcommand\longbottomdang[5]{
	\cubecoord{#1}{#2}{#3}
	\foreach \vertex in {A,...,D}{
		\draw
			(d\vertex)--node(temp)[yev#4]{$#5$}++(-90:.7)
		;
		\ifthenelse{\equal{#5}{+}}
		{
			\node at (temp)[fcil, ciliatednode={90}{3}]{\phantom{$+$}};
		}
	}
}
\newcommand\bottomfacedang[3]{
	\cubecoord{#1}{#2}{#3}
	\node[dvertex](cbottom) at (barycentric cs:A=1,B=1,C=1,D=1){};
	\draw (cbottom)--++(-90:.5) ;
}
\newcommand\bottomfaceextend[4]{
	\cubecoord{#1}{#2}{#3}
	\node[dvertex](cbottom) at (barycentric cs:A=1,B=1,C=1,D=1){};
	\draw[\extendtype,\extendclr] (cbottom)--++(-90:\extend) ;
	\ifthenelse{\equal{#4}{+}} {\node at (cbottom)[fcil, extendciliatednode={90}{3}]{\phantom{$+$}}; }
}

\newcommand\leftdang[3]{
	\cubecoord{#1}{#2}{#3}
	\foreach \vertex in {A,D,E,H}{
		\draw(d\vertex)--++(180:.3);
	}
}	
\newcommand\leftextend[4]{
	\cubecoord{#1}{#2}{#3}
	\foreach \vertex in {A,D,E,H}{
		\draw[\extendtype,\extendclr](d\vertex)--++(180:\extend) ;
		\ifthenelse{\equal{#4}{+}} { \node at (d\vertex)[fcil, extendciliatednode={180}{3}]{\phantom{$+$}}; }
	}
}
\newcommand\longleftdang[5]{
	\cubecoord{#1}{#2}{#3}
	\foreach \vertex in {A,D,E,H}{
		\draw(d\vertex)--node(temp)[xev#4]{$#5$}++(180:.7);
		\ifthenelse{\equal{#5}{+}}
		{
			\node at (temp)[fcil, ciliatednode={0}{3}]{\phantom{$+$}};
		}
	}
}
\newcommand\leftfacedang[3]{
	\cubecoord{#1}{#2}{#3}
	\node[dvertex](cleft) at (barycentric cs:A=1,D=1,E=1,H=1){};
	\draw (cleft)--++(180:.5) ;
}
\newcommand\leftfaceextend[4]{
	\cubecoord{#1}{#2}{#3}
	\node[dvertex](cleft) at (barycentric cs:A=1,D=1,E=1,H=1){};
	\draw[\extendtype,\extendclr] (cleft)--++(180:\extend) ;
	\ifthenelse{\equal{#4}{+}} {\node at (cleft)[fcil, extendciliatednode={180}{3}]{\phantom{$+$}}; }
}

\newcommand\rightdang[3]{
	\cubecoord{#1}{#2}{#3}
	\foreach \vertex in {B,C,F,G}{
		\draw(d\vertex)--++(0:.3);
	}
}
\newcommand\rightextend[4]{
	\cubecoord{#1}{#2}{#3}
	\foreach \vertex in {B,C,F,G}{
		\draw[\extendtype,\extendclr](d\vertex)--++(0:\extend) ;
		\ifthenelse{\equal{#4}{+}} { \node at (d\vertex)[fcil, extendciliatednode={0}{3}]{\phantom{$+$}}; }
	}
}
\newcommand\longrightdang[5]{
	\cubecoord{#1}{#2}{#3}
	\foreach \vertex in {B,C,F,G}{
		\draw(d\vertex)--node[xev#4]{$#5$}++(0:.7);
	}
}

\newcommand\rightfaceextend[4]{
	\cubecoord{#1}{#2}{#3}
	\node[dvertex](cright) at (barycentric cs:B=1,C=1,F=1,G=1){};
	\draw[\extendtype,\extendclr] (cright)--++(00:.5) ;
	\ifthenelse{\equal{#4}{+}} {\node at (cright)[fcil, extendciliatednode={0}{3}]{\phantom{$+$}}; }
}

\newcommand\frontdang[3]{
	\cubecoord{#1}{#2}{#3}
	\foreach \vertex in {A,B,E,F}{
		\draw(d\vertex)--++(180+\angle:.3);
	}
}
\newcommand\frontextend[4]{
	\cubecoord{#1}{#2}{#3}
	\foreach \vertex in {A,B,E,F}{
		\draw[\extendtype,\extendclr](d\vertex)--++(180+\angle:.3);
		\ifthenelse{\equal{#4}{+}} {\node at (d\vertex)[fcil, extendciliatednode={180+\angle}{3}]{\phantom{$+$}}; }
	}
}

\newcommand\frontfaceextend[4]{
	\cubecoord{#1}{#2}{#3}
	\node[dvertex](cfront) at (barycentric cs:A=1,B=1,E=1,F=1){};
	\draw[\extendtype,\extendclr] (cfront)--++(180+\angle:\extend) ;
	\ifthenelse{\equal{#4}{+}} {\node at (cfront)[fcil, extendciliatednode={180+\angle}{3}]{\phantom{$+$}}; }
}

\newcommand\backextend[4]{
	\cubecoord{#1}{#2}{#3}
	\foreach \vertex in {C,D,G,H}{
		\draw[\extendtype,\extendclr](d\vertex)--++(\angle:\extend) ;
		\ifthenelse{\equal{#4}{+}} { \node at (d\vertex)[fcil, extendciliatednode={\angle}{3}]{\phantom{$+$}}; }
	}
}
\newcommand\longbackdang[5]{
	\cubecoord{#1}{#2}{#3}
	\foreach \vertex in {C,D,G,H}{
		\draw(d\vertex)--node(temp)[zev#4]{$#5$}++(\angle:.7);
		\ifthenelse{\equal{#5}{+}}
		{
			\node at (temp)[fcil, ciliatednode={\angle}{3}]{\phantom{$+$}};
		}
	}
}
\newcommand\backfacedang[3]{
	\cubecoord{#1}{#2}{#3}
	\node[dvertex](center) at (barycentric cs:C=1,D=1,G=1,H=1){};
	\draw (center)--++(\angle:.5) ;
}
\newcommand\backfaceextend[4]{
	\cubecoord{#1}{#2}{#3}
	\node[dvertex](cback) at (barycentric cs:C=1,D=1,G=1,H=1){};
	\draw[\extendtype,\extendclr] (cback)--++(\angle:\extend) ;
	\ifthenelse{\equal{#4}{+}} { \node at (cback)[fcil, extendciliatednode={\angle}{3}]{\phantom{$+$}}; }
}

\newcommand\cubecoordmid[3]{
	\path[shift={(0:#1)},shift={(90:#2)}, shift={(\angle:#3*\cosine)}]
		(0,0)  coordinate (A)  (1,0)     coordinate (B) %
		+(\angle:\cosine) coordinate (C)  (\angle:\cosine) coordinate (D)
		(0,1)  coordinate (E) +(\angle:\cosine) coordinate (H)   %
		(1,1)  coordinate (F) +(\angle:\cosine) coordinate (G)  ;
	\path
		(A) -- node(ab)[dvertex]{}(B) (A)--node(ae)[dvertex]{}(E) (A)--node(ad)[dvertex]{}(D)
		(F) -- node(bf)[dvertex]{}(B) (F)--node(ef)[dvertex]{}(E) (F)--node(fg)[dvertex]{}(G)
		(C) -- node(bc)[dvertex]{}(B) (C)--node(cd)[dvertex]{}(D) (C)--node(cg)[dvertex]{}(G)
		(H) -- node(dh)[dvertex]{}(D) (H)--node(eh)[dvertex]{}(E) (H)--node(gh)[dvertex]{}(G)
		;
}

\newcommand\topextendmid[4]{
	\cubecoordmid{#1}{#2}{#3}
	\foreach \vertex in {ef,fg,gh,eh}{
		\draw[\extendtype,\extendclr] (\vertex)	--++(90		:\extend) ;
	}
}
\newcommand\bottomextendmid[4]{
	\cubecoordmid{#1}{#2}{#3}
	\foreach \vertex in {ab,ad,bc,cd}{
		\draw[\extendtype,\extendclr] (\vertex)	--++(-90		:\extend) ;
	}
}
\newcommand\leftextendmid[4]{
	\cubecoordmid{#1}{#2}{#3}
	\foreach \vertex in {ae,ad,dh,eh}{
		\draw[\extendtype,\extendclr] (\vertex)	--++(180		:\extend) ;
	}
}
\newcommand\rightextendmid[4]{
	\cubecoordmid{#1}{#2}{#3}
	\foreach \vertex in {bc,bf,fg,cg}{
		\draw[\extendtype,\extendclr] (\vertex)	--++(0		:\extend) ;
	}
}
\newcommand\frontextendmid[4]{
	\cubecoordmid{#1}{#2}{#3}
	\foreach \vertex in {ab,ae,bf,ef}{
		\draw[\extendtype,\extendclr] (\vertex)	--++(180+\angle	:\extend) ;
	}
}
\newcommand\backextendmid[4]{
	\cubecoordmid{#1}{#2}{#3}
	\foreach \vertex in {cd,cg,dh,gh}{
		\draw[\extendtype,\extendclr] (\vertex)	--++(\angle	:\extend) ;
	}
}

\tikzset{dvertex/.style={
			draw=none, 
			minimum size=3mm,
			fill=none,
			inner sep = 0pt},
}

\tikzset{fvertex/.style={
				draw, minimum size=3mm, fill=white, inner sep = 0pt},
}
\tikzset{fcil/.style={
				draw=none, minimum size=3mm, fill=none, inner sep = 0pt},
}

\tikzset{halfedge/.style={
				draw, minimum height=0mm, minimum width=1.5mm, inner sep = 0pt},
}

\tikzset{kappa/.style={
				draw, fill=white, minimum size=2mm, inner sep = .3mm, font=\small,
			},
}

\def\clr{}
\def\linetype{}
\def\subscale{1}	
\def\dist{1}	
\tikzset{vcaption/.style={node distance=.0cm, font=\scriptsize} }

\makeatletter
\newlength\@SizeOfCirc%
\newcommand{\CricArrowRight}[1]{%
	\setlength{\@SizeOfCirc}{\maxof{\widthof{#1}}{\heightof{#1}}}%
	\tikz [>=angle 90, x=1.0ex,y=1.0ex,line width=.1ex, draw=black]%
	\draw [->,anchor=center]%
		node[fill=none, draw=none] (0,0) {#1}%
		(0,1.0\@SizeOfCirc) arc (85:-240:1.0\@SizeOfCirc);%
}%
\newcommand{\clockwise}[1]{%
	\setlength{\@SizeOfCirc}{\maxof{\widthof{#1}}{\heightof{#1}}}%
	\tikz [>=angle 90, x=1.0ex,y=1.0ex,line width=.1ex, draw=black]%
	\draw [->,anchor=center]%
		node[fill=none, draw=none] (0,0) {#1}%
		(0,1.0\@SizeOfCirc) arc (85:-240:1.0\@SizeOfCirc);%
}%
\newcommand{\counterclockwise}[1]{%
	\setlength{\@SizeOfCirc}{\maxof{\widthof{#1}}{\heightof{#1}}}%
	\begin{tikzpicture}[>=angle 90, x=1.0ex,y=1.0ex,line width=.1ex, draw=black]%
	\draw [->,anchor=center]%
		node[fill=none, draw=none] (0,0) {#1}%
		(0,1.0\@SizeOfCirc) arc (-275:40:1.0\@SizeOfCirc);%
	\end{tikzpicture}
}%
\makeatother

\def\clocksize{4}

\def\angend{45}
\newcommand{\clockxy}[1]{%
	\tikz[>=angle 90, x=1.0ex,y=1.0ex,line width=.1ex, draw=black]%
	\draw [->,anchor=center]%
		node[fill=none, draw=none] (0,0) {#1}%
		(0,{1*\clocksize}) arc (0:\angend-360:{.75*\clocksize});%
}%
\newcommand{\clockxz}[1]{%
	\tikz[>=angle 90, x=1.0ex,y=1.0ex,line width=.1ex, draw=black]%
	\draw [->,anchor=center]%
		node[fill=none, draw=none] (0,0) {#1}%
		(0,{.5*\clocksize}) arc (0:\angend-360:{\clocksize*\cosine} and {\clocksize*\sine});%
}%
\newcommand{\clockyz}[1]{%
	\tikz[>=angle 90, x=1.0ex,y=1.0ex,line width=.1ex, draw=black]%
	\draw [->,anchor=center]%
		node[fill=none, draw=none] (0,0) {#1}%
		(0,{.5*\clocksize}) arc (0:\angend-360:{\clocksize*\cosine} and {\clocksize*\sine});%
}%

\tikzstyle ciliation=[circle,draw,fill=white,inner sep=0pt, minimum size=#1]
\tikzset{ ciliatednode/.style 2 args
	={ pin={[pin distance=0mm, ciliation=#2]#1:{}} }
}
\tikzstyle gciliation=[gray, circle,draw,fill=white,inner sep=0pt, minimum size=#1]
\tikzset{ gciliatednode/.style 2 args
	={ pin={[pin distance=0mm, gciliation=#2]#1:{}} }
}
\tikzstyle extendciliation=[\extendclr, circle,draw,fill=white,inner sep=0pt, minimum size=#1]
\tikzset{ extendciliatednode/.style 2 args
	={ pin={[pin distance=0mm, extendciliation=#2]#1:{}} }
}

\tikzset{ --|/.style={
  decoration={ markings,
  mark=at position 1 with {\pgftransformscale{1.5}\arrowreversed[line width=0mm]{|}}},
  postaction={decorate}
  }
}

\usepackage{soul}

\usepackage{cite}

\newcommand\C{\mathbb{C}}
\newcommand\Z{\mathbb{Z}}

\newcommand{\x}{x}
\newcommand{\y}{y}
\newcommand{\z}{z}

\excludecomment{commentold}
\excludecomment{comment}

\begin{document}

\title{A Factor-Graph Approach to Algebraic Topology,
       With Applications to Kramers--Wannier Duality}

\author{Ali Al-Bashabsheh 
        and 
        Pascal O.\ Vontobel,~\IEEEmembership{Senior Member,~IEEE}%
	\thanks{
		Manuscript received June 30, 2016; revised June 22, 2017 and December 27, 2017; accepted April 5, 2018.
		This paper was presented in part at the 
                2015 IEEE International Symposium on Information 
                Theory~\cite{av:kramers--wannier}.}
	\thanks{A.~Al-Bashabsheh was with the Institute of Network Coding,
	        The Chinese University of Hong Kong, Shatin, Hong Kong.
                He is now with
					 the Beijing Advanced Innovation Center for Big Data
					 and Brain Computing, Beihang University, Beijing, 100191
                (email: entropyali@gmail.com).}
	\thanks{P.~O.~Vontobel is with the 
                Department of Information Engineering
                and the Institute of Theoretical Computer Science 
                and Communications,
                The Chinese University of Hong Kong, Shatin, Hong Kong
                (email: pascal.von\-tobel@ieee.org).}
	\thanks{This work was supported in part by a grant from
	        the University Grants Committee of the Hong Kong SAR, 
                China (Project No.~{AoE/E-02/08)}.}
}

\maketitle

\begin{abstract}
  Algebraic topology studies topological spaces with the help of tools from
  abstract algebra. The main focus of this paper is to show that many concepts
  from algebraic topology can be conveniently expressed in terms of (normal)
  factor graphs. As an application, we give an alternative proof of a
  classical duality result of Kramers and Wannier, which expresses the
  partition function of the two-dimensional Ising model at a low temperature in
  terms of the partition function of the two-dimensional Ising model at a high
  temperature.  Moreover, we discuss analogous results for the
  three-dimensional Ising model and the Potts model.
\end{abstract}


\section{Introduction}
\label{sec:into}

General topology (also known as point-set topology) studies
transformation-invariant properties of topological spaces. For example, a
mug-shaped object can be smoothly transformed into a doughnut-shaped object,
and so general topology does not study the exact details of a given mug-shaped
object, but only the properties that are maintained after the mug-shaped
object has been smoothly transformed into a doughnut-shaped object.

A particular approach to topology is based on abstract algebra and results in
what is known as algebraic topology (see, e.g., \cite{bamberg:course,
  hatcher:topology, munkres:algebraic-topology}). The power of algebraic
topology comes from the vast amount of results available in abstract algebra,
in particular about vector spaces.\footnote{More generally, algebraic topology
  can be formulated in terms of modules, however, for our purposes vector
  spaces are general enough. (Recall that a module is, roughly speaking, a
  generalization of a vector space where the scalar multiplication by an
  element in some field is replaced by scalar multiplication by an element in
  some ring.)} As we will see in this paper, central objects of algebraic
topology are vector spaces associated with topological spaces, along with
boundary operators and coboundary operators defined on these vector
spaces. Homology is then the study of certain quotient spaces based on images
and kernels of the boundary operators. On the other hand, cohomology is the
study of certain quotient spaces based on images and kernels of the coboundary
operators. The importance of homology and cohomology comes from the fact that
the dimensions of the above-mentioned quotient spaces are invariants of a
topological space.

In this paper we show how these objects can be conveniently represented with
the help of normal factor graphs~\cite{frank:factor, loeliger:intro,
  forney:normal}. Besides this representation being of inherent interest, the
power of this approach comes from the fact that one can apply various known
results from the factor-graph literature to study the resulting factor graphs,
in particular, one can apply various duality results (see, e.g.,
\cite{forney:normal, ay:nfg-hol}).

Of particular interest in this paper are topological spaces where the
above-mentioned quotient spaces are trivial or low dimensional.
\begin{itemize}

\item Consider first the case where one of these quotient spaces is trivial,
  which implies that certain two vector spaces are equal. This equality of two
  vector spaces is interesting because these two vector spaces have typically
  rather different looking representations in terms of factor graphs. Now,
  assume to have a factor graph representing one of the two vector spaces. The
  above observation allows one to replace this factor graph by a factor graph
  representing the other vector space. As mentioned in the next paragraph,
  such a replacement can be beneficially used in the study of certain types of
  factor graphs.

\item Consider now the case where one of the above-mentioned quotient spaces
  is low dimensional, but not trivial. Similar observations as above can be
  made, however, one factor graph is now replaced by a small number of factor
  graphs. As we will see, also this replacement can be beneficially used in
  the study of certain types of factor graphs.

\end{itemize}

As an application of our factor-graph approach to algebraic topology, we show
how the Kramers--Wannier duality~\cite{kramers-wannier} (see also~\cite{savit:duality,
druhl:duality}), which expresses the partition function of
the two-dimensional Ising model \cite{ising1925} at a low temperature in terms of the partition
function of the two-dimensional Ising model at a high temperature, can be
re-proven with the tools introduced in this paper. As a quick preview,
Fig.~\ref{fig:kw-duality:simplified:1} summarizes our approach to proving the
Kramers--Wannier duality: first, we will apply a ``Fourier transform'' step
(which amounts to using Fourier duality results for factor graphs) and then we
will apply a ``change of support NFG'' step (which amounts to using the
observations made in the previous paragraph). In our opinion, the resulting
proof is easier than existing proofs and gives additional insights.

The fact that one can express the partition function of a statistical model at a
low temperature in terms of the partition function of the same or another
statistical model at a high temperature is very valuable because frequently it
is easier to simulate systems at higher temperatures than at lower
temperatures. Overall, note that the Kramers--Wannier duality fits into a more
general theme in physics, where some properties of some model are expressed as
some other properties of some other model. A very famous, somewhat recent,
example is the anti-de Sitter~/ conformal field theory correspondence, also
known as Maldacena duality or gauge~/ gravity duality~\cite{maldacena:98:1}.

On the side, let us point out that the factor-graph approach to algebraic
topology discussed in this paper has recently been used beneficially to study
stabilizer quantum codes~\cite{Li:Vontobel:16:1}. Moreover, a recently
submitted paper by Forney~\cite{Forney:algebraic:topology:notes} discusses
additional results on how to express objects from algebraic topology in terms
of normal factor graphs, with a particular emphasis on coding theory; that
paper was motivated by~\cite{av:kramers--wannier} and an earlier version of
the present paper.

This paper is structured as follows. In Section~\ref{sec:nfg}, we review the
basics of normal factor graphs (NFGs), in particular also how to obtain the
Fourier transform of an NFG. A central part of the present paper are
Sections~\ref{sec:one:complex}--\ref{sec:two:complex}, where we review notions
from algebraic topology and then show how they can be expressed in terms of
NFGs. In Section~\ref{sec:statistical}, we review statistical models, in
particular the Boltzmann distribution. We then combine the results of the
previous sections toward re-proving the Kramers--Wannier duality for the
two-dimensional Ising model in Section~\ref{sec:kw}. Finally, we consider
extensions of Kramer--Wannier type duality results to the three-dimensional
Ising model and to the two-dimensional Potts model \cite{potts1952} (see also, e.g.,
\cite{wu:Potts})
in Section~\ref{sec:ext},
and conclude the paper in Section~\ref{sec:conclusion:1}.

Throughout this paper, we use calligraphic letters to primarily denote
sets. Moreover, $\mathbb{F}$ is some fixed finite field, $\R$ is the field of
real numbers, and $\mathbb{C}$ is the field of complex numbers.

\section{Normal Factor Graphs}
\label{sec:nfg}

In this section we review, as far as needed for this paper, the basics of
normal factor graphs (NFGs) and their Fourier transform. For further
background on these topics we refer the interested reader
to~\cite{frank:factor, loeliger:intro, forney:normal, ay:nfg-hol,
  forney-pascal:partition-functions}.

\subsection{NFG definition}
\label{sec:nfg:definition:1}

We use the following notation. Let $\E$ and $\X_e$, $e \in \E$, be some finite
sets. Based on these sets, we define $\X^{\E} \defeq \prod_{e \in \E}
\X_e$. Frequently, an element of $\X^{\E}$ is denoted by $x_{\E}$ and referred
to as a \emph{configuration}. For any configuration $x_{\E}$ and any nonempty
$\I \subseteq \E$, we use $x_{\I}$ to denote the components of $x_{\E}$ that
are indexed by $\I$, i.e., $x_{\I} \defeq (x_i : i \in \I)$.

\begin{definition}[Normal factor graph]

  Consider an undirected (finite) graph with vertex set $\vnfg$ and edge set
  $\E$. Based on this graph, we construct an NFG $\N \defeq (\vnfg, \E)$ as
  follows:
  \begin{itemize}

  \item with each edge $e \in \E$ we associate a finite alphabet $\X_e$ and a variable $x_e$ taking
	  values from $\X_e$;

  \item and with each vertex $f \in \vnfg$ we associate a (local) function
    $f(x_{\E(f)})$, where $\E(f)$ is the set of edges incident on $f$. The
    degree of a local function $f$ is defined as the degree of $f$, i.e.,
    the number of edges incident on $f$. (Note that we use $f$ as the
    label of the function node and as the name of the local function.)

  \end{itemize}

\item We have two types of edges in $\E$: half-edges, each of which is
  incident on one vertex, and full-edges, each of which is incident on two
  vertices. In the following, we will use $\D$ to denote the set of half-edges
  and $\E \setminus \D$ to denote the set of full-edges. We define
  \begin{itemize}

  \item the global function of the NFG $\N$ to be
    \begin{align*}
      f_{\N}(x_{\E})
        &\defeq\prod_{f \in \vnfg} f(x_{\E(f)}) \eqpunc ,
    \end{align*}

  \item and the exterior function of $\N$ to be
    \begin{align*}
      Z_{\N}(x_{\D})
        &\defeq 
           \sum_{x'_{\E} : \, x'_{\D} = x_{\D}} 
             f_{\N}(x'_{\E}),
               \quad x_{\D} \in \X_{\D} \eqpunc .
    \end{align*}
    If $\D$ is empty, then $Z_{\N}$ is a constant that is called the partition
    function of the NFG $\N$. (Clearly, such a constant depends on the local
    functions. We often allow some of the local functions to depend on some
    parameter(s), and so $Z_{\N}$ is a function of such parameter(s).)

  \end{itemize}
  Throughout the paper, local functions will take on values in $\mathbb{R}$ or
  $\mathbb{C}$. With that, the global function and the exterior function
  will also take on values in $\mathbb{R}$ or $\mathbb{C}$.
  \qedd
\end{definition}

Subsequently, we make the following assumption about the variable
alphabets. (Many of the upcoming results can be suitably generalized to other
scenarios.)

\begin{assumption}
  \label{assumption:alphabet:1}

  We assume $\X_e \defeq \X$, $e\in \E$, where $\X$ is some finite field
  $\F$.
  \qeda
\end{assumption}

Of common use in the context of NFGs are indicator functions, which are
zero/one-valued functions. The following two indicator functions are of
particular interest:
\begin{itemize}

\item the equality indicator function $\delta_{=}$, which evaluates to
  one if all its arguments are equal, and to zero otherwise;

\item the parity indicator function $\delta_{+}$, which evaluates to
  one if the sum of its arguments is zero, and to zero otherwise. 

\end{itemize}
The two indicator functions above are marked in the NFG by drawing an ``$=$''
or a ``$+$'' symbol, respectively, inside the corresponding vertices.  Note
that an indicator function may be viewed as a constraint on its arguments, and
so, for a configuration for which the indicator function evaluates to one, we
may say that the indicator function is satisfied.

We say that an NFG $\N$ represents a function $f$ if $Z_{\N}$ equals $f$ up to
a (positive) scaling factor, where, whenever possible, we keep track of such a
scaling factor in the explicit expression relating $f$ and $Z_{\N}$. Let $\Y$
be some subset of $\X^{\D}$ and let $\delta_{\Y}$ be its indicator
function. If $\N$ represents $\delta_{\Y}$, then, for brevity, we say that
$\N$ represents the set $\Y$.

Finally, a small circle marking an edge $e$ incident on a vertex $f$ is
adopted to mean that $x_{e}$ appears negated in the local function $f$, i.e.,
such a circle is a shorthand notation of the parity indicator function of
degree two.%
\footnote{If the characteristic of $\X$ is $2$, such markings on the NFG are
  not necessary since $x = -x$ for all $x \in \X$. However, even for setups
  where the characteristic of $\X$ is $2$ (see, e.g., Section~\ref{sec:kw}
  about the Ising model), we choose this convention in order to make the
  discussion general.}

\begin{example}
  \label{example:nfg:introduction:1}

  Fig.~\ref{fig:eg-nfg}(a) shows an NFG $\N$, whose global func\-tion is given
  by
  \begin{align*}
    f_{\N}(x_{1}, \ldots, x_{45}) 
      &= f_1(x_1,x_{12},x_{15}) \cdot f_{3}(x_3,x_{23},x_{34}) \\
      &\quad\quad
         \cdot \,
         f_4(-x_{24},x_{34},x_{45})
         \cdot
         \delta_{+}(x_{12},x_{23},x_{24}) \\
      &\quad\quad
         \cdot \,
         \delta_{=}(x_{5},x_{15},x_{45})
         \eqpunc .
  \end{align*}
  From the definition of the exterior function it follows that 
  \begin{align*}
    Z_{\N}(x_{1}, \, x_{3}, \, x_{5}) 
      &= \sum_{x_{12},x_{23},x_{34}}
           f_1(x_1,x_{12},x_{5})
           \cdot
           f_{3}(x_3,x_{23},x_{34}) \\
      &\quad\quad\quad\quad\quad\quad\quad
           \cdot \,
           f_4(x_{12}\!+\!x_{23},x_{34},x_5) \eqpunc .
  \end{align*}
  \qede
\end{example}

\begin{figure}[t]
	\def\dist{1.3cm}
	\centering
	\def\xshift{4.5}
	\def\yshift{-3.5}
	\begin{tikzpicture}
          \node(a) at (0,0){\input{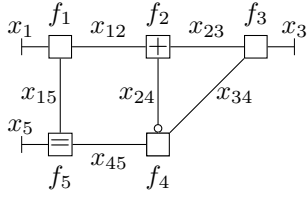}};
	\end{tikzpicture}
	\caption[aaa]{NFG $\N$ discussed in Example~\ref{example:nfg:introduction:1}.}
	\label{fig:eg-nfg}
\end{figure}

\begin{example}
  \label{example:nfg:introduction:2}
  As another example,
  Fig.~\ref{fig:nfg:fourier:transform:1}(a) shows an NFG $\N$
  whose global function is given by
  \begin{align*}
	  f_{\N}(&x_1, \dots, x''_{\mathrm{d}}) =
	  f_1(x_1) \cdot
	  f_2(x_2) \cdot
	  f_3(x_3) \cdot
	  f_4(x_4) \cdot
	  f_5(x_5) \cdot
	  \\ &
	  \delta_{=}(x_{\rm{a}},x'_{\rm{a}}) \cdot
	  \delta_{=}(x_{\rm{b}},x'_{\rm{b}},x''_{\rm{b}}) \cdot
	  \delta_{=}(x_{\rm{c}},x'_{\rm{c}}) \cdot
	  \delta_{=}(x_{\rm{d}},x'_{\rm{d}},x''_{\rm{d}}) \cdot
	  \\ &
	  \delta_{+}(x_1,x'_{\rm{a}},-x_{\rm{d}}) \cdot
	  \delta_{+}(x_2,x_{\rm{a}},-x_{\rm{b}}) \cdot
	  \delta_{+}(x_3,x'_{\rm{b}},-x'_{\rm{d}}) \cdot
	  \\ &
	  \delta_{+}(x_4,x''_{\rm{b}},-x'_{\rm{c}}) \cdot
	  \delta_{+}(x_5,x_{\rm{c}},-x''_{\rm{d}}) \eqpunc{.}
  \end{align*}
  One can verify that the NFG's exterior function (which here is also the
  NFG's partition function) is given by
	\begin{align*}
		Z_{\N}= \sum_{x_{\rm{a}}, x_{\rm{b}}, x_{\rm{c}}, x_{\rm{d}}}
		& f_{1}(x_{\rm{d}}-x_{\rm{a}})\cdot
		f_{2}(x_{\rm{b}}-x_{\rm{a}})\cdot
		f_{3}(x_{\rm{d}}-x_{\rm{b}})\cdot
		\\ &
		f_{4}(x_{\rm{c}}-x_{\rm{b}})\cdot
		f_{5}(x_{\rm{d}}-x_{\rm{c}}) \eqpunc .
	\end{align*}
  \qede
\end{example}

In this work, except for the NFG in Example~\ref{example:nfg:introduction:1}, the NFGs we consider are of a particular form. 

\begin{assumption}
  \label{assumption:nfg:1}

  The class of NFGs we are interested in is such that every local function is a parity indicator
  function or an equality indicator function. In addition to such functions and the edges
  connecting them, the NFG also contains degree-one (real or complex valued) functions attached to some
  of the indicator functions. 
  At some places we may replace every degree-one function with a half edge as detailed below.
    \qeda

\end{assumption}

For ease of reference, we define an \emph{interaction function} to be a
degree-one local function that is not an indicator function. The motivation
for this terminology will become evident as we progress.

\begin{example}
  \label{example:special:type:nfg:example:1}
	Figs.~\ref{fig:nfg:fourier:transform:1}(a) and (b) show two examples of the NFGs we will see in
	this work. Both NFGs satisfy Assumption~\ref{assumption:nfg:1}.
  \qede
\end{example}

Given an NFG $\N = (\vnfg,\E)$ as in Assumption~\ref{assumption:nfg:1},
we make the following definitions:
\begin{itemize}

\item The support NFG of $\N$, sometimes denoted $\N^{\circ} = (\vnfg^{\circ},
  \E^{\circ})$, is the NFG obtained from $\N$ by cutting out all the
  interaction functions, i.e., by replacing each interaction function (and its
  incident edge) with a half-edge. Note that $\E^{\circ} = \E$, where we
  assume that full-edges that became half-edges kept their label. The
  half-edges of $\N^{\circ}$ will be denoted by $\D^{\circ}$.

\item The set of (global) \emph{valid} configurations, denoted by $\B_{\N}$,
  is the set of configurations $x_{\E}$ that satisfy all the indicator
  functions in $\N$.
  (Or, equivalently, all the indicator functions in $\N^{\circ}$.)

\item The set of projected valid configurations, denoted by $\setC_{\N}$, is
  the projection of the elements of $\B_{\N}$ onto the half-edges of the
  support NFG of $\N$, i.e.,
  \begin{align}
    \setC_{\N}
      &\defeq
         \left\{ 
           x_{\D^{\circ}}: x \in \B_{\N}
         \right\} \eqpunc ,
           \label{eq:set:proj:valid:configuration:1}
  \end{align}
\end{itemize}
(If the NFG $\N$ is clear from the context, then $\B_{\N}$ and $\setC_{\N}$
might be simplified to $\B$ and $\setC$, respectively.)

\begin{example}
  \label{example:special:type:nfg:example:2}

  We continue Example~\ref{example:special:type:nfg:example:1}. The support
  NFGs of the NFGs in Figs.~\ref{fig:nfg:fourier:transform:1}(a) and~(b) are
  shown in Figs.~\ref{fig:nfg:fourier:transform:1}(c) and~(d), respectively.
  \qede
\end{example}

The importance of the support NFG and the set of projected valid
configurations will become clear from subsequent discussions. Note that for
any valid configuration $x_{\E}$, the value of the global function, i.e.,
$f_{\N}(x_{\E})$, depends only on the part of $x_{\E}$ indexed by the edges
incident on the interaction functions.

\begin{figure}[t]
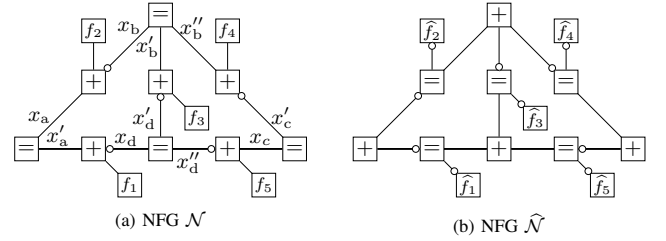

	\centering
  \def\xshift{4.5cm}
  \def\yshift{-3.5cm}
  \begin{tikzpicture}
    \def\xone{} \def\xtwo{} \def\xthree{} \def\xfour{} \def\xfive{}
    \def\fone{} \def\ftwo{} \def\fthree{} \def\ffour{} \def\ffive{}
    \def\dist{2.8cm}
    \def\distin{\dist*.5}
    \def\distout{\dist*.27}
    \def\s{.9}
    \small
    \def\angletri{45}
    \def\dist{2.8cm}
    \def\fone{$f_1$} \def\ftwo{$f_2$} \def\fthree{$f_3$} 
    \def\ffour{$f_4$} \def\ffive{$f_5$}
    \node(c)[scale=\s] at (0,0)					
      {\input{figtex/support/support1-labeled.tex}};
    \def\fone{$\widehat{f}_1$} \def\ftwo{$\widehat{f}_2$} 
    \def\fthree{$\widehat{f}_3$}
    \def\ffour{$\widehat{f}_4$} \def\ffive{$\widehat{f}_5$}
    \node(d)[scale=\s] at (\xshift,0)		
      {\input{figtex/support/support1ft.tex}};
    \def\dist{2.8cm}
    \def\fone{} \def\ftwo{} \def\fthree{} \def\ffour{} \def\ffive{}
    \node(g)[scale=\s] at (0*\xshift,1*\yshift)
      {\input{figtex/support/support1s.tex}};
	 \def\fone{} \def\ftwo{} \def\fthree{} \def\ffour{} \def\ffive{}
		\node(h)[scale=\s] at (1*\xshift,1*\yshift) 
                {\input{figtex/support/support1sft.tex}};
    \node(c)[vcaption, below = of c]{(a) NFG $\N$};
    \node(d)[vcaption, below = of d]{(b) NFG $\widehat{\N}$};
    \node(g)[vcaption, below = of g]{(c) The support NFG of $\N$};
    \node(h)[vcaption, below = of h]{(d) The support NFG of $\widehat{\N}$};
  \end{tikzpicture}
  \caption{NFGs discussed in Examples~\ref{example:special:type:nfg:example:1},
    \ref{example:special:type:nfg:example:2},
    and~\ref{example:special:type:nfg:example:3}.}
  \label{fig:nfg:fourier:transform:1}
\end{figure}

\subsection{Pairwise interaction NFG}
\label{sec:interaction}

As we mentioned in the previous subsection, any NFG in this work will be
according to Assumption~\ref{assumption:nfg:1}.
We also saw in Fig.~\ref{fig:nfg:fourier:transform:1} some NFGs representative
of the ones we will see in this work. As discussed in the next subsection, the
NFGs in Figs.~\ref{fig:nfg:fourier:transform:1}(a) and~(c) are intimately
related to the ones in Figs.~\ref{fig:nfg:fourier:transform:1}(b) and~(d),
respectively.
Namely, given one of the NFGs, the other NFG is uniquely determined, where
we not only have an explicit relation between the local functions of the NFGs,
but also an explicit relation between the exterior functions of the NFGs.  This
motivates naming one of the NFGs and treating it as a primary object,
while viewing the other one as a secondary NFG.
For reasons that are briefly motivated below, and will become more apparent in
subsequent sections, we choose NFGs similar to the NFG in
Fig.~\ref{fig:nfg:fourier:transform:1}(a) to be the primary NFGs.

The NFG in Fig.~\ref{fig:nfg:fourier:transform:1}(a) has a simple
interpretation as a physical system. For instance,
\begin{itemize}

\item each equality indicator function can be seen (through the equality
  constraint it imposes on the variables of its incident edges) as a variable
  representing some physical property of, say, a particle, e.g., spin
  orientation;

\item the degree-one functions can be seen as representing the interaction
  between neighbouring particles;

\item and the parity indicator functions can be seen as restricting attention
  to systems in which the interaction between neighbouring particles depends
  only on the difference between the value of the particles' variables.%
  \footnote{One can consider more general interactions where each parity
    indicator function and its degree-one function are replaced by a bivariate
    real or complex function. However, this is beyond the scope of this work.}

\end{itemize}
In another instance,
\begin{itemize}

\item each equality indicator function can be seen as the voltage potential of
  a site in an electrical network;

\item the degree-one functions can be seen as representing the interaction
  between two sites through an electrical component, e.g., a resistor;

\item and the parity indicator functions can be seen as asserting that
  interactions between adjacent sites depend on the voltage differences
  between the two sites.

\end{itemize}
(For more details on electrical networks in this context, we refer the
interested reader to \cite{vontobel:electrical}, where an NFG that resembles
Fig.~\ref{fig:nfg:fourier:transform:1}(a) was referred to as a ``voltage
version'' factor graph.)

Such a physical system of pairwise interactions can be depicted as a
(function) weighted graph $\G = (\V, \A)$ with vertex set $\V$ and directed
edge set $\A$ as shown in Fig.~\ref{fig:support-interaction} for the example
NFG $\N$ in Fig.~\ref{fig:nfg:fourier:transform:1}(a). For such a graph, we
associate with every vertex $v \in \V$ the variable $x_v$, and with every
directed edge $a_i = (u,v) \in \A$ the weight $f_i(x_v - x_u)$, where $f_i$ is the
corresponding interaction function. (The directedness of the edge reflects the
fact that an interaction function might not necessarily be symmetric, i.e.,
$f_i(-x) \neq f_i(x)$, in general.)

With this, the weighted graph $\G$ can be constructed from the NFG $\N$ as
follows.
\begin{itemize}

\item Replace each equality indicator function in $\N$ with a vertex $v$ in
  $\G$ and associate with such a vertex a variable $x_v$.

\item Replace each parity indicator function, its incident edges, and interaction
  function in $\N$ with an edge $a \defeq (u,v)$ in $\G$, where $u$ and $v$
  are the vertices in $\G$ that correspond to the (equality) neighbours of the
  parity indicator function in $\N$. The edge is directed towards the variable
  that appears negated in the parity indicator function. Finally, with $f$
  denoting the interaction function in hand, the edge is assigned the weight
  $f(x_v-x_u)$.

\end{itemize}
It is clear that the above procedure is reversible, i.e., given $\G$, we can
recover $\N$ in the obvious way. Moreover, note that we can now express the
partition function of $\N$ in terms of $\G = (\V,\A)$ as
\[
  Z_{\N} = \sum_{x_{\V}} \prod_{a_i=(u,v)\in \A} f_{i}(x_{v}-x_{u}) \eqpunc{,}
\]
which is the expression given in Example~\ref{example:nfg:introduction:2}
since $\V = \{ v_{\mathrm{a}}, \ldots, v_{\mathrm{d}} \}$ and
$\A=\{a_1,\dots,a_5\}$.

\begin{figure}[t]
	\centering
	\ifonecol
		\def\xshift{6cm}
		\centering
		\else
		\def\xshift{4.5cm}
	\fi
	\def\yshift{-3.5cm}
	\begin{tikzpicture}
	\def\xone{} \def\xtwo{} \def\xthree{} \def\xfour{} \def\xfive{}
	\def\fone{} \def\ftwo{} \def\fthree{} \def\ffour{} \def\ffive{}
	\def\dist{2.8cm}
	\def\distin{\dist*.5}
	\def\distout{\dist*.27}
	\def\s{.9}
	\small
	\def\angletri{45}
		\node(a)[scale=\s] at (0,-\yshift)		{\input{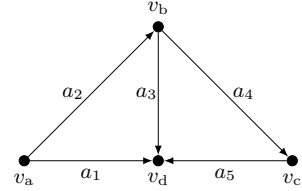}};
	\end{tikzpicture}
	\caption{
		The corresponding weighted directed graph of the NFG in Fig.~\ref{fig:nfg:fourier:transform:1}(a).  
	}
	\label{fig:support-interaction}
\end{figure}

The above discussion motivates the following definition.
\begin{definition}
  \label{assumption:nfg:2}
  An NFG according to Assumption~\ref{assumption:nfg:1} that resembles
  Fig.~\ref{fig:nfg:fourier:transform:1}(a), i.e., can be associated with a
  weighted graph according to the above procedure, is referred to as a
  (pairwise) interaction NFG.%
  \footnote{Similar to the interaction functions attached to parity indicator
    functions, one can also allow degree-one functions to be attached to the
    equality indicator functions. In the system of particles' spins, such
    degree-one functions represent the existence of an external field.  }
  \qeda
\end{definition}
While the main focus of this work is on pairwise interactions, the definition
of an interaction NFG can be easily extended to higher-order interactions.
In this case,
an interaction function is attached to a parity indicator function of a degree that can be larger than three.
We will see an example of such NFGs in
Section~\ref{sec:3d:ising:model:1}.

In contrast to the above, the NFG obtained from an interaction NFG as
described in the next subsection (e.g., the NFG in
Fig.~\ref{fig:nfg:fourier:transform:1}(b)) does not seem to have a physical
interpretation in general. (Note that in terms of electrical networks, such an
NFG can be seen as a ``current version'' factor graph
\cite{vontobel:electrical}, but it can no longer be regarded as a system of
pairwise interactions. In the example of a system of particles' spins, it is
not clear what interpretation can be given to such an NFG.) Nevertheless, as
will be detailed in place, such an NFG will be very useful as an alternative
NFG for approximating the partition function at low temperature (as in
sampling-based methods, which perform poorly on the interaction NFG at low
temperature), or as an intermediate step in mapping an interaction NFG at low
temperature to another interaction NFG at high temperature (as in
Kramers--Wannier duality).

Throughout this paper, we make the following assumption.
\begin{assumption}
  \label{assumption:connectedness:1}

  All graphs and NFGs in this work are connected graphs, unless specified
  otherwise.
\end{assumption}

\color{black}

\subsection{Fourier transform of an NFG}
\label{sec:fourier}

In this section we give a brief review of the Fourier transform of an
NFG~\cite{mao:convolution-fg, ay:nfg-hol, forney-pascal:partition-functions}.

Note that, because of Assumption~\ref{assumption:alphabet:1},
$\X$ is a finite Abelian group w.r.t. addition. Let $\widehat{x}:\X \rightarrow
\C\backslash\{0\}$ be a group homomorphism, i.e., $\widehat{x}(x+y) =
\widehat{x}(x) \cdot \widehat{x}(y)$ for all $x,y\in \X$.  Such a homomorphism
is called a character (on $\X$). Moreover, the set of all characters,
denoted by $\widehat{\X}$, is a group that is isomorphic to $\X$, where for all
$\widehat{x},\widehat{y}\in \widehat{\X}$, the addition in $\widehat{\X}$ is
defined as $(\widehat{x}+\widehat{y})(x) \defeq \widehat{x}(x) \cdot
\widehat{y}(x)$, $x \in \X$. (See, e.g., \cite{lidl:finite-fields}.)

\begin{definition}
  Let $f: \X \rightarrow \C$ be an arbitrary function on $\X$. The Fourier
  transform $\widehat{f}: \widehat{\X} \rightarrow \C$ of $f$ is then defined
  as
  \begin{align}
    \widehat{f}(\widehat{x})
      &\defeq 
         \sum_{x\in \X} 
           f(x) \cdot \widehat{x}(x) \eqpunc .
  	\label{eq:fourier} \\[-1cm]
    \nonumber
  \end{align}
  \qedd
\end{definition}

One can show that $f$ is recoverable from $\widehat{f}$ via\footnote{This
  follows from the well-known fact that for two characters $\widehat{x}_1,
  \widehat{x}_2 \in \widehat{\X}$ the expression $\sum_{x} \widehat{x}_1(x)
  \cdot \widehat{x}_2(-x)$ evaluates to $|\X|$ if $\widehat{x}_1 =
  \widehat{x}_2$ and to $0$ otherwise.}
\begin{align}
  f(x) 
    &= \frac{1}{|\X|}
         \cdot
         \sum_{\widehat{x} \in \widehat{\X}} 
           \widehat{f}(\widehat{x}) \cdot \widehat{x}(-x) \eqpunc .
\end{align}

Let $\Y$ be a subgroup of $\X$ and define
\begin{align*}
  \Y^{\perp}
    &\defeq 
       \Bigl\{ 
         \widehat{x} \in \widehat{\X}
       \Bigm|
         \widehat{x}(x) = 1 \ \text{for all $x \in \Y$}
       \Bigr\} \eqpunc .
\end{align*}
One can verify that ${\Y}^{\perp}$ is a subgroup of $\widehat{\X}$; it is
called the orthogonal subgroup to $\Y$.
Let $\delta_{\Y}$ be the indicator function of $\Y$, i.e.,
for all $x\in \X$, the indicator function defined as 
$\delta_{\Y}(x) = 1$ iff $x\in \Y$.
The following fact regarding
the Fourier transform of $\delta_{\Y}$
will be useful.

\begin{lemma}
  \label{lemma:fourier-group-indicator}

  Let $\Y$ be a subgroup of $\X$. Then it holds that
  \begin{align}
    \widehat{\delta_{\Y}}= |\Y| \cdot \delta_{ {\cal
        Y}^{\perp}} \eqpunc .
	\label{eq:dual-group}
  \end{align}
\end{lemma}

\begin{proof}
  See, e.g., \cite[Ch.~5]{lidl:finite-fields}.
\end{proof}

\begin{example}
  \label{example:important:fourier:transforms:1}

  Consider an equality indicator function node and a parity indicator function
  node of degree $d$. It is straightforward to verify that there are subgroups
  $\Y_{=}$ and $\Y_{+}$ of $\X^d$ such that $\delta_{=} = \delta_{\Y_{=}}$ and
  $\delta_{+} = \delta_{\Y_{+}}$, respectively. Moreover, because of
  $(\Y_{=})^{\perp} = \Y_{+}$, $|\Y_{=}| = |\X|$, $|\Y_{+}| = |\X|^{d-1}$, and
  Lemma~\ref{lemma:fourier-group-indicator}, we have
	\begin{equation}
	\begin{aligned}
    \label{eq:equality-parity}
    \widehat{\delta_{=}} 
      &= |\X|\cdot \delta_{+} \eqpunc ,
	  \\
    \widehat{\delta_{+}} 
      &= |\X|^{d-1}\cdot \delta_{=} \eqpunc .
	\end{aligned}
	\end{equation}
  \qede
\end{example}

\begin{remark}[Fourier-transformed NFG]

  Given an NFG $\N$, the Fourier-transformed NFG
  $\widehat{\N}$ can be obtained as follows:
  \begin{enumerate}
  
  \item Insert a degree-two parity indicator function on each internal edge.
  
  \item Replace each local function with its Fourier
    transform.\footnote{\label{footnote:fourier:transform:rescaling:1}
      Frequently, we will replace a local (indicator) function not by its Fourier
      transform but by a scaled version of its Fourier transform, and
      separately keep track of the product of all the omitted scaling
      factors.
		For example, $\delta_{=}$ is replaced by $\delta_{+}$ (not by $|\X| \cdot \delta_{+}$) and the omitted scaling factor $|\X|$ is dealt with separately.
		Similarly, $\delta_{+}$ is replaced by $\delta_{=}$ (not by $|\X|^{d-1} \cdot \delta_{=}$) and the omitted scaling factor $|\X|^{d-1}$ is dealt with separately.
	}

  \item As far as needed in order to be explicit, associate suitable variables
    with the edges in $\widehat{\N}$.
  
  \end{enumerate}
  The exterior functions of $\N$ and $\widehat{\N}$ are related as
  follows~\cite{ay:nfg-hol}
  \begin{align}
    \widehat{Z_{\N}}(\widehat{x}_{\D})
      &= \frac{1}
              {\prod_{e \in \E\backslash \D}
                |\X_e|
              }
         \cdot
         Z_{\widehat{\N}}(\widehat{x}_{\D})
           \eqpunc .
             \label{eq:nfg-dual-general}
  \end{align}
  In particular, when there are no half edges, $Z_{\N}$ and $Z_{\widehat \N}$
  are equal, up to a scaling factor, i.e., $ Z_{\N} = \frac{1}{\prod_{e \in
      \E} |\X_e|} \cdot Z_{\widehat{\N}}$. (The expression in
  \eqref{eq:nfg-dual-general} assumes that local functions are replaced by
  their Fourier transform in Step 2. If the procedure outlined in
  Footnote~\ref{footnote:fourier:transform:rescaling:1} is applied, then the
  scaling factor in~\eqref{eq:nfg-dual-general} needs to be suitably
  adjusted.)
  \qede

\end{remark}

\begin{example}
  \label{example:special:type:nfg:example:3}

  Fig.~\ref{fig:nfg:fourier:transform:1}(b) shows the Fourier transform of the
  NFG in Fig.~\ref{fig:nfg:fourier:transform:1}(a). Note that we have taken
  advantage of the procedure outlined in
  Footnote~\ref{footnote:fourier:transform:rescaling:1}.
  \qede
\end{example}

The next theorem is a specialization of \eqref{eq:nfg-dual-general} customized
to the NFGs in this work.
\begin{theorem}
	\label{thm:nfg-dual}
	Let $\N$ be an interaction NFG as in Definition~\ref{assumption:nfg:2} and $\widehat{\N}$ be its
	Fourier-transformed NFG, then
	\begin{align}
		\code_{\N} = (\code_{\widehat{\N}})^{\perp}\eqpunc.
		\label{thm:ft-code}
	\end{align}
	Moreover,  with the scaling factors of the indicator functions omitted in $\widehat{\N}$ (see
	Footnote~\ref{footnote:fourier:transform:rescaling:1}), we have 
	\begin{align}
		Z_{\N} = \frac{1}{|\X|^{|\A|-|\V|}} \cdot Z_{\widehat{\N}}\eqpunc,
		\label{thm:ft-z}
	\end{align}
	or more explicitly,
	\ifonecol
	\begin{align}
		|\X| \cdot \!\!\!
		\sum_{x_{\A}\in \code_{\N}} \prod_{e\in \A}  f_{e}(x_{e})
		=
		\frac{1}{|\X|^{|\A| - |\V|}}  \cdot  \!\!
		\sum_{x_{\A}\in (\code_{\N}\!)^{\perp}} \prod_{e\in \A} \! \widehat{f}_{e}(x_{e})\eqpunc.
		\label{thm:ft-explicit}
	\end{align}
	\else
	\begin{align}
		|\X| \cdot \!\!\!
		\sum_{x_{\A}\in \code_{\N}} \prod_{e\in \A} \! f_{e}(x_{e})
		=
		\frac{1}{|\X|^{|\A| - |\V|}} \cdot \!\!\! \!\!\! \!
		\sum_{x_{\A}\in (\code_{\N}\!)^{\perp}} \prod_{e\in \A} \widehat{f}_{e}(x_{e})\eqpunc,
		\label{thm:ft-explicit}
	\end{align}
	\fi
	where $(\V,\A)$ is the weighted graph obtained from $\N$ using the procedure above Definition~\ref{assumption:nfg:2}.
	\qedb
\end{theorem}
\begin{proof}
	We prove (\ref{thm:ft-code})
	by observing that
	\begin{align*}
		{\delta_{\code_{\widehat{\N}}}}
		&= 
		c\cdot Z_{(\widehat{\N})^{\circ}}
		=
		c\cdot Z_{\widehat{(\N^{\circ})}}
		\\ &
		\stackrel{(\ref{eq:nfg-dual-general})}{=}
		c'\cdot \widehat{Z_{\N^{\circ}}}
		=
		c''\cdot \widehat{\delta_{\code_{\N}}}
		\stackrel{(\ref{eq:dual-group})}{=}
		c''' \cdot \delta_{(C_{\N})^{\perp}}\eqpunc,
	\end{align*}
	where $c, \ldots, c'''$ are scaling factors, and where the unlabelled equalities
	follow directly from the definition of the support NFG.
	Equation (\ref{thm:ft-z}) follows from
	(\ref{eq:nfg-dual-general}) and (\ref{eq:equality-parity}) by noting that, in $\N$, 
	there are $3\times |\A|$ full edges
	and no half edges,
	and that there are $|\V|$ equality-indicator functions and $|\A|$  parity-indicator functions (each of degree three).
	Finally, $Z_{\N}$ is equal to the l.h.s. of (\ref{thm:ft-explicit}),
        which is clear by noting that only a valid configuration contributes
        to $Z_{\N}$ and such a contribution depends only on the corresponding
        projection in $\code_{\N}$.  The scaling factor $|\X|$
        appears %
	because there are $|\X|$ valid configurations that have the same
        projection $x$ for all $x\in \code_{\N}$. (Here we used
        Assumption~\ref{assumption:connectedness:1}.)  The r.h.s. of
        (\ref{thm:ft-explicit}) follows from ({\ref{thm:ft-code}}) and
        (\ref{thm:ft-z}) by noting that $Z_{\widehat{\N}} = \sum_{x_{\A} \in
          \code_{\widehat{\N}}} \prod_{e\in \A}\widehat{f}_{e}(x_e)$. (Observe
        that, in comparison to the argument leading to the l.h.s, the scaling
        factor here is one since a configuration on the half edges of the
        support NFG of $\widehat{\N}$ uniquely determines
		a global configuration on all the edges of such an NFG.)
\end{proof}
\color{black}

\section{$1$-Complexes}
\label{sec:one:complex}

In this section we introduce $1$-complexes. More precisely,
Section~\ref{sec:one:complex:primal:domain:1} defines chains, boundary
operators, and homology spaces, whereas
Section~\ref{sec:one:complex:dual:domain:1} defines the duals of these
objects, i.e., cochains, coboundary operators, and cohomology spaces,
respectively. For further background on these topics, we refer the interested
reader to~\cite{bamberg:course, hatcher:topology, munkres:algebraic-topology}.

Sections~\ref{sec:one:complex:primal:domain:1} and~\ref{sec:one:complex:dual:domain:1} are complemented by
Sections~\ref{sec:one:complex:nfg:representation:1}
and~\ref{sec:one:complex:nfg:representation:2}, where we introduce some NFGs
that we associate with boundary and coboundary operators, respectively.

Note that $2$-complexes will be introduced in Sections~\ref{sec:two:complex}.
(Higher-order complexes are briefly discussed in
Appendix~\ref{sec:cw-complex}.)

\subsection{Chains, boundary operators, and homology spaces}
\label{sec:one:complex:primal:domain:1}

\begin{definition}
  \label{def:one:complex:1}

  Consider a graph $\G \defeq (\V,\A)$ with vertex set $\V$ and edge set
  $\A$. We define the following objects:
  \begin{itemize}
 
  \item $C_0 \defeq \F^{\V}$, whose elements are called $0$-chains;

  \item $C_1 \defeq \F^{\A}$, whose elements are called $1$-chains.

  \end{itemize}
  Note that $C_0$ is the set of functions from $\V$ to $\F$, and $C_1$ is the
  set of functions from $\A$ to $\F$. We can also think of $C_0$ to be a
  $|\V|$-dimensional vector space over $\F$ and of $C_1$ to be an
  $|\A|$-dimensional vector space over $\F$.

  With a $0$-chain $\x \in \F^{\V}$ we associate the formal sum
  \begin{align}
	  \label{eq:0-chain-basis}
    \x
      &\defeq
         \sum_{v \in \V}
           \x(v) \cdot v \eqpunc .
  \end{align}
  Similarly, with a $1$-chain $\y \in \F^{\A}$ we associate the formal sum
  \begin{align}
	  \label{eq:1-chain-basis}
    \y
      &\defeq
         \sum_{a \in \A}
           \y(a) \cdot a \eqpunc .
  \end{align}
  With the above associations, we can view $\V$ and $\A$ as (standard) bases
  of $\F^{\V}$ and $\F^{\A}$, respectively. (Here, for every $v \in \V$, the
  vertex $v$ is associated with the $0$-chain $1\cdot v \in \F^{\V}$, with a
  similar statement for the edges $a \in \A$.)
  \qedd
\end{definition}

In the following, for every $a \in \A$ we fix a direction.\footnote{Although
  some of the resulting calculations will be different for different choices
  of directions, the quantities that are ultimately of interest, like the
  dimension of certain quotient spaces, will be independent of this choice of
  directions.} We will write $a = (v,v')$ if the direction of the edge $a$ has
been chosen such that the edge starts at vertex $v$ and ends at vertex
$v'$. For such an edge $a$, the boundary vertices are $v$ and $v'$.  We
formalize this as follows.

\begin{definition}
  The boundary operator $\partial_1: C_1 \to C_0$, which maps $1$-chains to
  $0$-chains, is defined to be the (unique) linear map which satisfies
  \begin{align*}
    \text{for every $a = (v,v') \in \A$:} \quad \partial_1(a) \defeq v' - v.
  \end{align*}
  \qedd
\end{definition}

We see that $\partial_1(a)$ not only gives the boundary vertices of a directed
edge $a$, but also tells us which boundary vertex is the starting vertex and
which boundary vertex is the ending vertex, namely, the vertices with
coefficients $-1$ and $+1$, respectively. Note that if $a = (v,v')$, then
\begin{align*}
  \partial_1(-a) 
    &= - \partial_1(a) 
     = - (v' \! - \! v) 
     = v - v' \eqpunc ,
\end{align*}
i.e., $-a$ represents the directed edge from $v'$ to $v$.

\begin{figure}[t]
	\def\s{.80}	
	\def\L{2}
	\def\dist{1.9cm}	
	\def\linetype{}
	\def\dongle{.1}
	\def\c{.5}
	\ifonecol
		\def\xshift{6}
		\else
		\def\xshift{4.5}
	\fi
	\def\yshift{-4.9}
	\def\xoffset{.3*\s}
	\def\yoffset{.35*\s}
	\centering
	\begin{tikzpicture}
          \node(a)[scale=\s] at
          (.5*\xshift,0)
          {\input{figtex/square/square-1complex/square-graph.tex}};
	\end{tikzpicture}
	\caption{A graph $\G = (\V,\A)$.}
	\label{fig:square}
\end{figure}

\begin{example}
  \label{example:one:complex:1}

  Consider the directed graph in Fig.~\ref{fig:square} with $|\V| = 9$
  vertices and $|\A| = 12$ edges.
  In the following,
  we use  $a_{ij}$ to denote the directed edge $(v_i,v_j)$.
  \begin{itemize}

  \item Boundary of $a_{01} + a_{14} + a_{47}$:
    \begin{align*}
       \partial_1(a_{01} + a_{14} + a_{47})
         &= \partial_1(a_{01}) + \partial_1(a_{14}) + \partial_1(a_{47}) \\
         &= (v_1 \! - \! v_0) + (v_4 \! - \! v_1) + (v_7 \! - \! v_4) \\
         &= v_7 - v_0 \eqpunc .
    \end{align*}
    This calculation makes sense because $a_{01} + a_{14} + a_{47}$ represents
    the concatenation of the directed edges $a_{01}$, $a_{14}$, and $a_{47}$,
    which is the directed path from starting vertex $v_0$ to ending vertex
    $v_7$.
  
  \item Boundary of $a_{01} + a_{14} + (-a_{34}) + (-a_{03})$:
    \begin{align*}
      &
      \partial_1(a_{01} + a_{14} + (-a_{34}) + (-a_{03})) \\
        &\quad
         = \partial_1(a_{01} + a_{14} - a_{34} - a_{03}) \\
        &\quad
         = \partial_1(a_{01}) 
           + 
           \partial_1(a_{14}) 
           - 
           \partial_1(a_{34})
           - 
           \partial_1(a_{03}) \\
        &\quad
         = (v_1 - v_0)
           +
           (v_4 - v_1)
           -
           (v_4 - v_3)
           -
           (v_3 - v_0) \\
        &\quad
         = 0 \eqpunc .
    \end{align*}
    Again, this calculation makes sense because $a_{01} + a_{14} + (-a_{34}) +
    (-a_{03})$ represents the concatenation of the directed edges $a_{01}$,
    $a_{14}$, $-a_{34}$, and $-a_{03}$, which is a directed closed cycle from
    $v_0$ to $v_0$, and because directed closed cycles have no boundary.
  
  \item Boundary of $2a_{67} + 3a_{78}$:
    \begin{align*}
       \partial_1(2a_{67} + 3a_{78})
         &= 2 \cdot \partial_1(a_{67}) + 3 \cdot \partial_1(a_{78}) \\
         &= 2 \cdot (v_7 - v_6) + 3 \cdot (v_8 - v_7) \\
         &= 3 v_8 - v_7 - 2v_6 \eqpunc .
    \end{align*}
  
  \end{itemize}
  \qede
\end{example}

Recall that the kernel (denoted by ``$\ker$'') of a linear map is the set of
all elements in the domain that map to zero. %

\begin{figure}[t]
  \begin{align}
    \label{eq:one:complex:chain:1}
	&\textcolor{gray}{C_{2}} \ 
         \textcolor{gray}{\xrightarrow{\partial_{2}}} \
         C_{1} 
         \xrightarrow{\partial_{1}} 
         C_{0} \ 
         \textcolor{gray}{\xrightarrow{\partial_{0}}} \ 
         \textcolor{gray}{C_{-1}} \\
    \label{eq:one:complex:cochain:1}
        &\textcolor{gray}{\widehat{C}_{2}} \ 
         \textcolor{gray}{\xleftarrow{d_{2}}} \
         \widehat{C}_{1} 
         \xleftarrow{d_{1}} 
         \widehat{C}_{0} \ 
         \textcolor{gray}{\xleftarrow{d_{0}}} \ 
         \textcolor{gray}{\widehat{C}_{-1}}
  \end{align}
  \caption{Spaces and mappings associated with a $1$-complex.}
  \label{fig:one:complex:chain:and:cochain:1}
\end{figure}

\begin{definition}
  The elements of $\ker \partial_1$, i.e., the elements of $C_1$ with zero
  boundaries, will be called $1$-cycles.
  \qedd
\end{definition}

\begin{example}
  \label{example:one:complex:2}

  The second calculation in Example~\ref{example:one:complex:1} shows that
  $a_{01} + a_{14} - a_{34} - a_{03}$ is a $1$-cycle. Note that any multiple
  of $a_{01} + a_{14} - a_{34} - a_{03}$ is also a $1$-cycle.
  \qede
\end{example}

The following definition turns out to be useful for subsequent considerations.

\begin{definition}
  \label{def:one:complex:1:extension:1}

  We define the sets $C_2 \defeq \{ 0 \}$ and $C_{-1} \defeq \{ 0 \}$, along
  with the trivial mappings $\partial_2: C_2 \to C_1$ and $\partial_0: C_0 \to
  C_{-1}$.
  \qedd
\end{definition}

The objects that we have introduced so far are summarized in
\eqref{eq:one:complex:chain:1}; the collection of these objects is known as a
$1$-dimensional chain complex, or simply $1$-complex.

In the following, we will use ``$\img$'' to denote the image of a map.
Because $\img \partial_2 = \{ 0 \}$ and $\ker \partial_0 = C_0$, we clearly
have
\begin{align}
  \img \partial_2 
    &\subseteq 
       \ker \partial_1 \eqpunc ,
         \label{eq:one:complex:containment:1} \\
  \img \partial_1 
    &\subseteq 
       \ker \partial_0 \eqpunc . 
         \label{eq:one:complex:containment:0}
\end{align}
The fact that there might be a gap between $\img \partial_2$ and
$\ker \partial_1$ is captured by the $1$-st homology space. Similarly, the
fact that there might be a gap between $\img \partial_1$ and $\ker \partial_0$
is captured by the $0$-th homology space.

\begin{definition}
  \label{def:homology:spaces:1:1}
 
  The $1$-st and the $0$-th homology spaces are defined to be the quotient
  spaces
  \begin{align}
    H_1
      &\defeq
         \ker\partial_{1} \ / \ \img\partial_{2} \eqpunc ,
      \label{eq:homology:1:1} \\
    H_0
      &\defeq
         \ker\partial_{0} \ / \ \img\partial_{1} \eqpunc ,
      \label{eq:homology:1:0}
  \end{align}
  respectively.
  \qedd
\end{definition}

We now discuss the dimensions of the spaces we have seen so far.
We start with $H_{0}$ and let $\G = (\V,\A)$ be a connected graph.%
\footnote{Here we benefited from \cite{forney:kw-private} in first discussing $\dim H_0$
and then $\dim H_{1}$, instead of the other way around. Such a choice facilitates the arguments
and  appears more natural.
}
	For an arbitrary $v_0 \in \V$ and any other $v\in \V$, let $\y \in C_{1}$ be a 
	$1$-chain corresponding to an oriented path in $\G$ from $v_0$ to $v$.
	(An oriented path is a directed path that contains elements from $\{a,-a \mid a \in \A\}$, where
	$-a \defeq (v,u)$ for all $(u,v)\in \A$.
	The corresponding $1$-chain $y$ is such that for all
	$a\in \A$, $\y(a)=1$, if $a$ is in the oriented path;
	$\y(a)=-1$, if $-a$ is in the oriented path;
	and $\y(a)=0$, otherwise.)
	Then,
	\[
		\partial_{1}\y = v-v_0 \in \img \partial_{1}\eqpunc,
	\]
	and so, every element $v$ of the basis $\V$ is equivalent to $v_0$, modulo $\img \partial_1$.
	Hence,
	\begin{align}
		\label{eq:dim-h0}
		\dim H_{0} = 1\eqpunc.
	\end{align}
	From this and (\ref{eq:homology:1:0}), it is immediate that
	\begin{align}
		\label{eq:dim-img-bd1}
		\dim(\img\partial_{1}) = \dim(\ker\partial_{0}) - \dim H_{0} = |\V|-1\eqpunc.
	\end{align}
	(Recall that $\partial_{0}$ is the zero map.)
	On the other hand, by the rank-nullity theorem, we have
	\begin{align}
		\label{eq:dim-ker-bd1}
		\dim(\ker\partial_{1}) = |\A| - \dim(\img\partial_{1}) = |\A|-|\V|+1\eqpunc,
	\end{align}
	and so,
	\begin{align}
		\label{eq:dim-h1}
		\dim H_{1} = |\A|-|\V|+1\eqpunc.
	\end{align}
	(Recall that, for a $1$-complex, $\partial_2$ is the zero map.) 

	To construct a basis for $\img \partial_{1}$, let $\T \defeq (\V,\A_{\T})$ be a spanning tree of
	$\G$. The $0$-chains in the set
	$\left\{ \partial_{1}a: a\in \A_{\T} \right\}$
	are independent since $\T$ contains no cycles, and so, by \eqref{eq:dim-img-bd1} they form a
	basis of $\img\partial_1$.

	To construct a basis for $\ker \partial_{1}$, note that each edge $a\in \A \backslash \A_{\T}$ induces an
	oriented cycle in $\G$, and so the corresponding $1$-chain is a
	$1$-cycle.
	These $1$-cycles are independent since each one involves a different
	edge from $\A\backslash \A_{\T}$, and so, by \eqref{eq:dim-ker-bd1} they form a basis of $\ker\partial_{1}$.
	\begin{example}
	  \label{example:one:complex:basis}
	  Continuing with the example graph $\G = (\V,\A)$ in Fig.~\ref{fig:square},
	  let $\cal T\defeq(\V,\A_{\cal T})$ be a spanning tree of $\G$, say,
	  \begin{align*}
		 \A_{\cal T}
			&\defeq
				\left\{ 
				  a_{01}, a_{12}, a_{34}, a_{45}, a_{67}, a_{78}, a_{14}, a_{47}
		 \right\} \eqpunc .
	  \end{align*}
	  Then the set
	  \begin{align*}
		 \bigl\{ 
			\partial_1(a)
		 \bigm|
			a \in \A_{\cal T}
		 \bigr\} 
			&= \left\{
				  v_{1}\!-\!v_{0}, \ 
				  v_{2}\!-\!v_{1}, \ 
				  \ldots, v_{7} \! - \! v_{4}
				\right\}
	  \end{align*}
	  forms a basis of $\img\partial_{1}$.

	  The edges $\A\backslash \A_{\T} = \{a_{03},a_{25},a_{36},a_{58}\}$ induce the cycles
	  \begin{align*}
		 c_{0} 
			&\defeq 
				a_{01} + a_{14} - a_{34} - a_{03}\eqpunc, \\
		 c_{1} 
			&\defeq 
				a_{12} + a_{25} - a_{45} - a_{14}\eqpunc, \\
		 c_{2}
			&\defeq 
				a_{34} + a_{47} - a_{67} - a_{36}\eqpunc, \\
		 c_{3}
			&\defeq 
				a_{45} + a_{58} - a_{78} - a_{47}\eqpunc,
	  \end{align*}
	  which form a basis of $\ker\partial_{1}$.
	  \qede
	\end{example}

	The above arguments can be extended to graphs with more than one
        component. Namely, for a graph $\G$ with $\operatorname{comp}(\G)$
        components, one can show (see, e.g., \cite[p. 425]{bamberg:course})
        that
        \begin{align*}
          \dim H_0 
            &= \operatorname{comp}(\G) \eqpunc, \\
          \dim H_1 
            &= |\A| - |\V| + \operatorname{comp}(\G) \eqpunc.
        \end{align*}
        These two quantities are known as the zeroth and first Betti numbers of the
        graph $\G$, respectively. (More generally, $\dim H_i$ is called the
        $i$-th Betti number.) Moreover, note that $\dim H_1$ is also known at
        the cyclomatic number of the graph $\G$, i.e., the minimum number of
        edges that must be removed in order to make the graph cycle free.

\subsection{NFG representation of the boundary operator}
\label{sec:one:complex:nfg:representation:1}

In this subsection we introduce the following NFGs:
\begin{itemize}

\item $\N_{\partial_1}$, an NFG whose exterior function is proportional to the
  indicator function of $\bigl\{ (\y, \partial_1 \y) \bigm| \y \in C_1 \bigr\}$.

\item $\N_{\ker \partial_1}$, an NFG whose exterior function is proportional
  to the indicator function of $\ker \partial_1$.

\item $\N_{\img \partial_1}$, an NFG whose exterior function is proportional
  to the indicator function of $\img \partial_1$.

\end{itemize}
We start by defining the set
\begin{align*}
  \B_{\partial_1}
    &\defeq
       \bigl\{
         (\y, \x)
       \bigm|
         \y \in C_1, \ \x = \partial_1 \y
       \bigr\} \eqpunc ,
\end{align*}
which contains all possible pairs of a $1$-chain and its image under
$\partial_1$.
We can rewrite $\B_{\partial_1}$ in terms of the bases in Definition~\ref{def:one:complex:1} as
\begin{align}
  \!\!\! \B_{\partial_1}
    = \Bigl\{ 
         \Bigl( 
           \bigl( \y(a) \bigr)_{a \in \A}, 
           \bigl( \x(v) \bigr)_{v \in \V} \kern-.08em
         \Bigr) 
       \Bigm| 
			 \y \in C_1, \x = \partial_{1}\y
       \Bigr\} \eqpunc .
		 \label{eq:bd:partial:1} 
\end{align}

Using this notation, we make the following definition, whose terminology was
motivated by analogous objects in~\cite{Forney:algebraic:topology:notes}.

\begin{definition}
  \label{def:input:output:nfg:1}

  Given a graph $\G = (\V,\A)$, we define the input/output NFG
  $\N_{\partial_1}$ associated with $\partial_1$ to be the NFG that has the
  following properties:
  \begin{itemize}

  \item For every $a \in \A$, there is an equality indicator function with two
    full edges and an input (ingoing) half-edge whose associated variable is
    $\y(a)$.

  \item For every $v \in \V$, there is a parity indicator function with a full
    edge (possibly with a sign inverter) for every edge incident on $v$ and an
    output (outgoing) half-edge whose associated variable is $\x(v)$.

  \item The full edges connect the indicator functions in the obvious way,
    i.e., the parity indicator function of $v\in \V$ is connected to the
    equality indicator function of $a\in \A$ iff $a$ is incident on $v$.
    \qedd

  \end{itemize}
\end{definition}

\begin{example}
  \label{example:input:output:nfg:1}

  Consider again the graph $\G = (\V,\A)$ in Fig.~\ref{fig:square}. The
  input/output NFG $\N_{\partial_1}$ associated with $\partial_1$ is shown in
  Fig.~\ref{fig:square1-continue}(a).  \qede
\end{example}

\begin{figure}[t]
	\def\s{.71}	
	\def\L{2}
	\def\dist{1.9cm}	
	\def\linetype{}
	\def\dongle{.1}
	\ifonecol
          \def\xshift{6}
          \else
          \def\xshift{4.5}
	\fi
	\def\yshift{-4.0}
	\def\xoffset{.3*\s}
	\def\yoffset{.35*\s}
	\centering
	\begin{tikzpicture}

                \def\c{.5}
		\node(a)[scale=\s] at
		(0*\xshift,0*\yshift)
                {\input{figtex/square/square-1complex/square-nfga.tex}};
		\node[vcaption, below = of a]
                     {(a) NFG $\N_{\partial_1}$};

		\def\c{.6}
		\node(c)[scale=\s] at
		(0*\xshift,1*\yshift)
                {\input{figtex/square/square-1complex/square-nfga-delta1.tex}};
		\node[vcaption, below = of c]
                     {(c) Intermediate NFG $\N_{\ker \partial_1}$};

		\def\c{.5}
		\node(e)[scale=\s] at
		(0*\xshift,2*\yshift)
                {\input{figtex/square/square-1complex/square-nfga-delta2.tex}};
		\node[vcaption, below = of e]
                     {(e) NFG $\N_{\ker \partial_1}$};

		\def\c{.6}
		\node(g)[scale=\s] at
		(0*\xshift,3*\yshift)
                {\input{figtex/square/square-1complex/square-nfga-one1.tex}};
		\node[vcaption, below = of g]
                     {(g) Intermediate NFG $\N_{\img \partial_1}$};

		\def\c{.5}
		\node(i)[scale=\s] at
		(0*\xshift,4*\yshift)
                {\input{figtex/square/square-1complex/square-nfga-one2.tex}};
		\node[vcaption, below = of i]
                     {(i) NFG $\N_{\img \partial_1}$};
                
                \def\c{.5}
		\node(b)[scale=\s] at
		(1*\xshift,0*\yshift)
                {\input{figtex/square/square-1complex/square-nfgb.tex}};
		\node[vcaption, below = of b]
                     {(b) NFG $\N_{d_1}$};

		\def\c{.6}
		\node(d)[scale=\s] at
		(1*\xshift,1*\yshift)
                {\input{figtex/square/square-1complex/square-nfgb-delta1.tex}};
		\node[vcaption, below = of d]
                     {(d) Intermediate NFG $\N_{\ker d_1}$};

		\def\c{.5}
		\node(f)[scale=\s] at
		(1*\xshift,2*\yshift)
                {\input{figtex/square/square-1complex/square-nfgb-delta2.tex}};
		\node[vcaption, below = of f]
                     {(f) NFG $\N_{\ker d_1}$};

		\def\c{.6}
		\node(h)[scale=\s] at
		(1*\xshift,3*\yshift)
                {\input{figtex/square/square-1complex/square-nfgb-one1.tex}};
		\node[vcaption, below = of h]
                     {(h) Intermediate NFG $\N_{\img d_1}$};

		\def\c{.5}
		\node(j)[scale=\s] at
		(1*\xshift,4*\yshift)
                {\input{figtex/square/square-1complex/square-nfgb-one2.tex}};
		\node[vcaption, below = of j]
                     {(j) NFG $\N_{\img d_1}$};
	\end{tikzpicture}
	\caption{NFGs associated with the graph $\G = (\V,\A)$ in
          Fig.~\ref{fig:square}.}
	\label{fig:square1-continue}
\end{figure}

\begin{lemma}
  \label{lemma:input:output:nfg:1}

  Let $\G = (\V,\A)$ be some graph and let $\N_{\partial_1}$ be the
  input/output NFG associated with $\partial_1$, then
	\begin{align}
		\label{eq:z-indicator}
	  Z_{\N_{\partial_1}}(x_{\D})
		 &\ \propto \ 
			 \delta_{\B_{\partial_1}}\!(x_{\D})
				\quad \text{for all $x_{\D} \in \X_{\D}$} \eqpunc .
	\end{align}
  \qedl
\end{lemma}

\begin{proof}
  Consider an arbitrary $1$-chain $\y$ and an arbitrary $0$-chain $\x$.  Then,
  we have $(\y,\x) \in \B_{\partial_{1}}$ iff 
	\begin{align}
		\x
		& = \partial_{1}\y
		= \sum_{a\in \A} \y(a) \cdot \partial_{1}(a)
		= \hspace{-5mm}\sum_{a \defeq (u,v)\in \A} \hspace{-5mm}\y(a) \cdot (v-u) \nonumber \\
		 &= \sum_{v\in \V} \Big( \sum_{a\in {\rm In}(v)} \y(a) - \sum_{a\in {\rm Out}(v)} \y(a) \Big) \cdot v,
                  \label{eq:input:output:nfg:partial:1}
	\end{align}
	where the second equality follows from~(\ref{eq:1-chain-basis}) and
        the linearity of $\partial_1$, and the last equality is obtained by
        rearranging the terms in the summation preceding it. Here we used the
        notation ${\rm In}(v)$ and ${\rm Out}(v)$ to denote the sets of edges
        entering and leaving a vertex $v$, respectively. By the independence
        of the vertices (in $C_0$), we have $\x = \partial_1 \y$ iff $\x(v) =
        \sum_{a\in {\rm In}(v)} \y(a) - \sum_{a\in {\rm Out}(v)}\y(a)$ for all
        $v\in \V$.  This constraint (with the output and input half-edges of
        $\N_{\partial_1}$ associated with $\x(v), v \in \V$, and $\y(a)$, $a
        \in \A$, respectively) corresponds to the local parity indicator
        functions in $\N_{\partial_1}$. The equality indicator functions
        simply account for the fact that each edge is incident on two
        vertices. Finally, one can verify that the proportionality constant
        in~\eqref{eq:z-indicator} is given by $\sum_{x_{\E} \in
          \X_{\E}: \, x_{\D} = 0_{\D}} f_{\N_{\partial_1}}(x_{\E})$, i.e., the
        number of valid configurations of $\N_{\partial_1}$ that are all-zero
        on the half-edges.
\end{proof}

We now proceed to define $\N_{\ker \partial_1}$ and
$\N_{\img \partial_1}$. Note that coordinate-based representations of
$\ker \partial_1$ and $\img \partial_1$ are, respectively,
\begin{align}
  \ker \partial_1
    &= \Bigl\{
         \bigl( \y(a) \bigr)_{a \in \A}
       \Bigm|
         \y \in C_1, \ 
         \partial_1 \y = 0
       \Bigr\} \eqpunc ,
         \label{eq:ker:partial:1} \\
  \img \partial_1
    &= \Bigl\{
          \x \defeq \bigl( (\partial_1 \y)(v) \bigr)_{v \in \V}
       \Bigm|
         \y \in C_1
       \Bigr\} \eqpunc ,
         \label{eq:img:partial:1}
\end{align}

\begin{definition}
  \label{def:img-kernel:boundary:nfg:1}
  Given an input/output NFG $\N_{\partial_1}$, we define the NFG
  $\N_{\ker \partial_1}$ to be be the NFG obtained from $\N_{\partial_1}$ by
  deleting all the output half-edges. We also define $\N_{\img \partial_1}$ as
  the NFG obtained from $\N_{\partial_1}$ by deleting all the input-half
  edges.  \qedd
\end{definition}

A lemma similar to Lemma~\ref{lemma:input:output:nfg:1} leads to
\begin{align*}
  Z_{\N_{\ker \partial_1}}(x_{\D})
    &\ \propto \ 
       \delta_{\ker \partial_1}\!(x_{\D}),
         &&\quad \text{$x_{\D} \in \X_{\D}$, 
                       $\D \defeq \D_{\N_{\ker \partial_1}}$} \eqpunc , \\
  Z_{\N_{\img \partial_1}}(x_{\D})
    &\ \propto \ 
       \delta_{\img \partial_1}\!(x_{\D}),
         &&\quad \text{$x_{\D} \in \X_{\D}$,
                       $\D \defeq \D_{\N_{\img \partial_1}}$} \eqpunc .
\end{align*}
This follows by noting that the NFG $\N_{\ker \partial_1}$ can be obtained
from $\N_{\partial_1}$ by attaching function nodes representing
Kronecker-delta functions to the output half-edges (see
(\ref{eq:bd:partial:1}) and (\ref{eq:ker:partial:1})), thereby, in effect,
deleting these half-edges. (More precisely, the NFG $\N_{\ker \partial_1}$ can
be obtained from $\N_{\partial_1}$ by constructing an intermediate NFG where
function nodes representing Kronecker-delta functions are attached to the
output half-edges (see (\ref{eq:bd:partial:1}) and (\ref{eq:ker:partial:1})),
and then simplifying the intermediate NFG to an NFG having the same exterior
function by deleting these Kronecker-delta functions and their incident
edges.)
Similarly, the NFG $\N_{\img \partial_1}$ can be obtained
from $\N_{\partial_1}$ by attaching function nodes representing all-one
functions to the input half-edges (see (\ref{eq:bd:partial:1}) and
(\ref{eq:img:partial:1})), thereby, in effect, deleting these half-edges.

\begin{example}
  \label{example:input:output:nfg:2}

  We continue
  Example~\ref{example:input:output:nfg:1}. Figs.~\ref{fig:square1-continue}(e)
  and~(c) show the NFG $\N_{\ker \partial_1}$ and the intermediate NFG used to
  obtain $\N_{\ker \partial_1}$ from $\N_{\partial_1}$, respectively.
  Similarly, Figs.~\ref{fig:square1-continue}(i) and~(g) show the NFG
  $\N_{\img \partial_1}$ and the intermediate NFG used to obtain
  $\N_{\img \partial_1}$ from $\N_{\partial_1}$, respectively.  \qede
\end{example}

  We remark that the same subspace may be represented by many NFGs and so,
  besides the NFGs presented in this section,
  one can give several NFGs representing $\B_{\partial_1}$, $\img \partial_1$, and $\ker
  \partial_1$.

Another example where an NFG appeared whose exterior function is proportional
to the indicator function of $\ker \partial_1$ is the support NFG
of~\cite[Fig.~12]{vontobel:electrical}.
In that paper, factor-graph representations of
electrical networks are considered, and $\ker \partial_1$ appears naturally in
that context because $\ker \partial_1$ encodes exactly Kirchhoff's current law
when $1$-chains are used to express currents along branches of an electrical
network.\footnote{Let us point out an important difference in the NFG drawing
  conventions used here and in~\cite{vontobel:electrical}. Namely, whereas
  here arrows are used to label input and output half-edges, respectively, in
  the paper~\cite{vontobel:electrical} arrows are used to express which
  arguments are taken positively and which arguments are taken negatively in a
  parity indicator function node.}

In order to highlight the importance of NFGs whose exterior function is
proportional to the indicator function of the kernel of some mapping, or
proportional to the indicator function of the image of some mapping, let us
connect the above findings to coding theory. Namely, in coding theory there
are two popular approaches to define linear codes:
\begin{itemize}

\item In a \emph{kernel representation}, a linear code of length $n$ and
  dimension $k$ over $\F$ is described as the kernel of some linear map
  $\phi$, i.e.,
  \begin{align*}
    \setC
      &\defeq
         \bigl\{
           x \in \F^n
         \bigm|
           \phi(x) = 0
         \bigr\}.
  \end{align*}
  Typically, the map $\phi$ is defined via a rank-$(n \! - \! k)$ parity-check
  matrix of size $m \times n$, where $m \geq n-k$. There are well-known
  approaches, in particular in the context of low-density parity-check codes,
  on how to specify an NFG whose exterior function is proportional to the
  indicator function of the code
  $\setC$~\cite{Tanner81,wiberg:95,frank:factor,loeliger:intro,forney:normal}.

\item In an \emph{image representation}, a linear code of length $n$ and
  dimension $k$ over $\F$ is described as the image of some linear map
  $\theta$, i.e.,
  \begin{align*}
    \setC
      &\defeq
         \bigl\{
           \theta(x)
         \bigm|
           x \in \F^k
         \bigr\}.
  \end{align*}
  Typically, the map $\theta$ is defined via a generator matrix of size $k
  \times n$. Again, there are well-known approaches, in particular in the
  context of turbo and low-density generator-matrix codes on how to specify an
  NFG whose exterior function is proportional to the indicator function of the
  code $\setC$~\cite{frank:factor,loeliger:intro,forney:normal}.

\end{itemize}

\begin{definition}
  \label{def:nfg:representation:forms:1}

  In the following,
  an NFG whose support NFG equals $\N_{\img \theta}$ for some
  linear mapping $\theta$ will be said to be in image-representation form.
  Similarly,
  an NFG whose support NFG equals $\N_{\ker \phi}$ for some
  linear mapping $\phi$ will be said to be in kernel-representation
  form.
  \qedd
\end{definition}

We conclude this subsection by noting that besides introducing NFGs associated
with the boundary operator $\partial_1$, one can also introduce NFGs
associated with the boundary operators $\partial_2$ and $\partial_0$. The
generalization of the above definitions is straightforward and we omit the
details.

\subsection{Cochains, coboundary operators, and cohomology spaces}
\label{sec:one:complex:dual:domain:1}

For all the objects that we encountered in
Section~\ref{sec:one:complex:primal:domain:1}, there are dual objects that we
now introduce.

\begin{definition}
  \label{def:cochain:1}

  Given a $1$-complex as in Definitions~\ref{def:one:complex:1}
  and~\ref{def:one:complex:1:extension:1}, for all $i = -1, 0, 1, 2$, let
  $\widehat{C}_i$ be the dual space of $C_i$, i.e., the space of all linear
  maps from $C_i$ to $\F$. Elements of $\widehat{C}_i$ are called
  $i$-cochains. Note that $\widehat{C}_{-1} = \{ 0 \}$ and $\widehat{C}_{2} = \{
  0 \}$.
  \qedd
\end{definition}

Because $\widehat{C}_0$ is the dual space of $C_0$, when defining an element
$\varphi$ of $\widehat{C}_0$, we do not need to specify $\varphi(\x)$ for all
$\x \in C_0$. It is sufficient to specify $\varphi(\x)$ on any basis of $C_0$, in particular,
it is sufficient to specify $\varphi(\x)$ for all $\x \in \V$.

\begin{definition}
  \label{def:coboundary:1}
  
  For all $i = 0, 1, 2$, we define the coboundary operator $d_{i}:
  \widehat{C}_{i-1} \rightarrow \widehat{C}_{i}$ to be the linear map which
  satisfies
  \begin{align}
    (d_i\varphi)(\cdot)
      &\defeq 
         \varphi\bigl( \partial_i(\cdot) \bigr)
           \quad \text{for all $\varphi \in \widehat{C}_{i-1}$} \eqpunc .
             \label{eq:def:coboundary:1:condition:1}
  \end{align}
  \qedd
\end{definition}

For $i = 1$, the condition in~\eqref{eq:def:coboundary:1:condition:1} implies
the following for an arbitrary directed edge $a = (v,v')$ and an arbitrary
$\varphi \in \widehat{C}_{0}$:
\begin{align*}
  (d_1 \varphi)(a) 
    &= \varphi\bigl( \partial_1(a) \bigr) 
     = \varphi(v'-v) 
     = \varphi(v') - \varphi(v) \eqpunc .
\end{align*}
Therefore, $d_1 \varphi$ is the function that assigns to the directed edge $a$
the difference between the $\varphi$-value at the ending vertex of the edge
and the $\varphi$-value at the starting vertex of the edge.\footnote{For
  example, in some physics application, $\varphi$ might represent some
  potential function. Then $d_1 \varphi$ is the function that yields the
  potential differences along directed edges. We refer the interested reader
  to~\cite{bamberg:course} for a very accessible introduction to the use of
  algebraic topology in physics.}

The objects that we have introduced in this subsection are summarized in
\eqref{eq:one:complex:cochain:1}; the collection of these objects is known as
a $1$-dimensional cochain complex.

Because $d_0$ and $d_2$ are the trivial maps, it holds that $\img d_0 = \{ 0
\}$ and $\ker d_2 = \widehat{C}_1$, from which it follows that
\begin{align}
  \img d_0
    &\subseteq 
       \ker d_1 \eqpunc ,
         \label{eq:one:complex:d:containment:0} \\
  \img d_1 
    &\subseteq 
       \ker d_2 \eqpunc . 
         \label{eq:one:complex:d:containment:1}
\end{align}
The fact that there might be a gap between $\img d_0$ and $\ker d_1$ is
captured by the $0$-th cohomology space. Similarly, the fact that there might
be a gap between $\img d_1$ and $\ker d_2$ is captured by the $1$-st
cohomology space.

\begin{definition}
  \label{def:cohomology:spaces:1:1}

  The $0$-th and the $1$-st cohomology spaces are defined to be the quotient
  spaces
  \begin{align}
    \widehat{H}_0
      &\defeq
         \ker d_{1} \ / \ \img d_{0} \eqpunc ,
      \label{eq:cohomology:1:0} \\
    \widehat{H}_1
      &\defeq
         \ker d_{2} \ / \ \img d_{1} \eqpunc ,
      \label{eq:cohomology:1:1}
  \end{align}
  respectively.
  \qedd
\end{definition}

Again, although the relationships in~\eqref{eq:one:complex:d:containment:0}
and~\eqref{eq:one:complex:d:containment:1} are rather trivial here, they will
naturally generalize to $m$-complexes for $m \geq 2$. 

Since (see Appendix~\ref{sec:cw-complex} for details) 
\begin{align}
	\label{eq:diff-ker}
	\ker d_{i} = \left( \img \partial_{i} \right)^{\perp}\eqpunc, \\
	\img d_{i} = \left( \ker\partial_{i} \right)^{\perp}\eqpunc,
	\label{eq:diff-img}
\end{align}
we have $\dim \widehat{H}_i = \dim H_i$ for $i = 0, 1$. Therefore, if one is
interested in computing $\dim H_i$ for $i = 0, 1$, one can not only compute
them via Definition~\ref{def:homology:spaces:1:1}, but also via $\dim H_i
= \dim \widehat{H}_i$ and Definition~\ref{def:cohomology:spaces:1:1}.

\subsection{NFG representation of the coboundary operator}
\label{sec:one:complex:nfg:representation:2}

In the same way that we associated the NFGs $\N_{\partial_1}$,
$\N_{\ker \partial_1}$, and $\N_{\img \partial_1}$ with the linear map
$\partial_1$, we can associate the NFGs $\N_{d_1}$, $\N_{\ker d_1}$, and
$\N_{\img d_1}$ with the linear map $d_1$. 

We start by defining the set
\begin{align*}
  \B_{d_1}
    &\defeq
       \bigl\{
			 (\widehat{\x}, \widehat{\y})
       \bigm|
         \widehat{\x} \in \widehat{C}_0, \ 
			\widehat{\y} = d_1 \widehat{\x}
       \bigr\} \eqpunc ,
\end{align*}
which contains all possible pairs of a $0$-cochain and its image under
$d_1$. In order to go from a coordinate-free representation of the elements of
this set to a coordinate-based representation, we define bases for
$\widehat{C}_0$ and $\widehat{C}_1$ such that the matrix representation of
$d_1$ is the transpose of the matrix representation of $\partial_1$. We can
then write 
\begin{align*}
  \widehat{\x} 
    &= \sum_{v \in \V} 
         \widehat{\x}(v) \cdot v, \\
  d_1 \widehat{\x}
    &= \sum_{a \in \A} 
         (d_1 \widehat{\x})(a) \cdot a.
\end{align*}
With this, we obtain
\begin{align*}
  \B_{d_1}
    &\defeq
       \Bigl\{
         \Bigl( 
           \bigl( \widehat{\x}(v) \bigr)_{v \in \V}, 
           \bigl( \widehat{\y}(a) \bigr)_{a \in \A}
         \Bigr)
       \Bigm|
         \widehat{\x} \in \widehat{C}_0, \ 
			\widehat{\y} = d_1 \widehat{\x}
       \Bigr\} \eqpunc .
\end{align*}

The following definitions are analogous to
Section~\ref{sec:one:complex:nfg:representation:2}, and so we omit the
details.
\begin{definition}
  \label{def:input:output:nfg-df}

  Given a graph $\G = (\V,\A)$, we define the input/output NFG
  $\N_{d_1}$ associated with $d_1$ to be the NFG that has the
  following properties:
  \begin{itemize}

  \item For every $a \in \A$, there is a parity indicator function with two
    full edges and an output (outgoing) half-edge whose associated variable is
    $\widehat \y(a)$.

  \item For every $v \in \V$, there is an equality indicator function with a
    full edge (possibly with a sign inverter) for every edge incident on $v$
    and an input (ingoing) half-edge whose associated variable is $\widehat
    \x(v)$.

  \item The full edges connect the indicator functions in the obvious way,
    i.e., the equality indicator function of $v\in \V$ is connected to the
    parity indicator function of $a\in \A$ iff $a$ is incident on $v$.  \qedd
  \end{itemize}

\end{definition}

\begin{definition}
  \label{def:img-kernel:boundary:nfg-df}

  Given an input/output NFG $\N_{d_1}$, we define the NFG $\N_{\ker d_1}$ to
  be be the NFG obtained from $\N_{d_1}$ by deleting all the output
  half-edges. We also define $\N_{\img d_1}$ as the NFG obtained from
  $\N_{d_1}$ by deleting all the input-half edges.  \qedd
\end{definition}

\begin{example}
  \label{example:input:output:nfg:3}

  Consider again the graph $\G = (\V,\A)$ in
  Fig.~\ref{fig:square}. Fig.~\ref{fig:square1-continue}(b) shows the
  input/output NFG $\N_{d_1}$; Figs.~\ref{fig:square1-continue}(f) and~(d)
  show the NFG $\N_{\ker d_1}$ and its intermediate version,
  respectively; Figs.~\ref{fig:square1-continue}(j) and~(h) show the
  NFG $\N_{\img d_1}$ and its intermediate version, respectively.
  \qede
\end{example}

From~(\ref{eq:diff-ker}), (\ref{eq:diff-img}), and \eqref{thm:ft-code} of
Theorem~\ref{thm:nfg-dual}, it is evident that, up to a scaling factor, we have
\begin{align}
	\label{eq:img1-ker1-nfg-a}
   \N_{\img d_1}&= \widehat{\N}_{\ker \partial_1} \\
	\label{eq:ker1-img1-nfg-a}
	\N_{\ker d_1}&= \widehat{\N}_{\img \partial_1}.
\end{align}
(See Figs.~\ref{fig:square1-continue}(j) and (e), and Figs.~\ref{fig:square1-continue}(f) and (i),
respectively.)

We conclude this subsection by noting that besides introducing NFGs associated
with the coboundary operator $d_1$, one can also introduce NFGs associated
with the coboundary operators $d_2$ and $d_0$.

\section{$2$-Complexes}
\label{sec:two:complex}

The algebraic topology approach becomes more valuable when moving to
$2$-complexes, and, more generally, to higher-order complexes.

The present section is about $2$-complexes and has a similar structure as
Section~\ref{sec:one:complex}: in
Section~\ref{sec:two:complex:primal:domain:1} we introduce chains, boundary
operators, and homology spaces, whereas in
Section~\ref{sec:two:complex:dual:domain:1} we introduce cochains, coboundary
operators, and cohomology spaces; these sections are complemented by
Sections~\ref{sec:two:complex:nfg:representation:1}
and~\ref{sec:two:complex:nfg:representation:2}, where we introduce some NFGs
that we associate with boundary and coboundary operators,
respectively. Finally, because of the importance of $2$-torus lattice graphs for
later sections, Section~\ref{sec:cw-torus} discusses these objects in detail.

\subsection{Chains, boundary operators, and homology spaces}
\label{sec:two:complex:primal:domain:1}

\begin{definition}[$2$-complex]
  \label{def:two:complex:1}

  In addition to vertices and edges, a $2$-complex also includes
  two-dimensional objects (surfaces). Consider a graph $\G \defeq (\V,\A)$
  that can be drawn on some surface without edge crossings, then the graph
  divides the surface into regions called \emph{faces}, which we denote by the
  set $\setS$. In the following, we will use the notation $\G \defeq
  (\V,\A,\setS)$ for such graphs. We define the following objects:
  \begin{itemize}
 
  \item $C_0 \defeq \F^{\V}$, whose elements are called $0$-chains;

  \item $C_1 \defeq \F^{\A}$, whose elements are called $1$-chains.

  \item $C_2 \defeq \F^{\setS}$, whose elements are called $2$-chains.

  \end{itemize}
  With a $0$-chain $\x \in \F^{\V}$, a $1$-chain $\y \in \F^{\A}$, a $2$-chain
  $\z \in \F^{\setS}$ we associate the formal sums
  \begin{align*}
    \x
      &\defeq
         \sum_{v \in \V}
           \x(v) \cdot v \eqpunc ,
    \quad
    \y
       \defeq
         \sum_{a \in \A}
           \y(a) \cdot a \eqpunc ,
    \quad
    \z
       \defeq
         \sum_{s \in \setS}
           \z(s) \cdot s \eqpunc ,
  \end{align*}
  respectively.
  \qedd
\end{definition}

As in Section~\ref{sec:one:complex}, for every $a \in \A$ we fix a direction.
Similarly, for every $s \in \setS$, we fix an orientation. Namely,
letting $\A_s$ be the set of edges adjacent to the face $s$, we associate an
orientation of $s$ which refers to traversing the edges in $\A_s$ in some direction (clockwise or counter-clockwise). 
The boundary of $s$ is defined as the weighted sum of the edges in $\A_s$, where the weight of an
edge depends on the relative direction of the edge w.r.t. to the chosen orientation of $s$. 

\begin{definition}
  \label{def:two:complex:boundary:operator:1}

  The boundary operator $\partial_1: C_1 \to C_0$, which maps $1$-chains to
  $0$-chains, is defined to be the (unique) linear map which satisfies
  \begin{align*}
    \text{for every $a = (v,v') \in \A$:} \quad \partial_1(a) \defeq v' - v \eqpunc .
  \end{align*}
  The boundary operator $\partial_2: C_2 \to C_1$, which maps $2$-chains to
  $1$-chains, is defined to be the (unique) linear map which satisfies
  \begin{align*}
    \text{for every $s \in \setS$:} 
      \quad 
      \partial_2(s) \defeq \sum_{a \in \A_s} \alpha(a) \cdot a \eqpunc ,
  \end{align*}
  where $\A_s$ is the set of edges adjacent to $s$ and $\alpha(a) = +1$ if the direction of
  $a$ is the same as the orientation of $s$ (i.e., the direction in which $\A_s$ is traversed) and $\alpha(a)=-1$ otherwise. \\
  {\mbox{} \hfill \qedd}
\end{definition}

In the case of a planar graph, we take the orientation of an inner face to be
clockwise and we take the orientation of the outer face, if it is included in $\s$, to be
counter-clockwise.

\begin{example}%
  \label{eg:2-complex-square:prelim:1}

  Consider the planar graph in Fig.~\ref{fig:square2}(a) with $|\V| = 16$ vertices,
  $|\A| = 24$ edges, and $|\setS| = 9$ faces. (Note that in this example there
  is no outer face.)
  In the following, we use $a_{ij}$ to denote the directed edge $(v_i,v_j)$.
  We have
  \begin{align*}
    \partial_2 s_0
      &= a_{01} + a_{15} - a_{45} - a_{04} \eqpunc , \\
    \partial_2 s_1 
      &= a_{12} + a_{26} - a_{56} - a_{15} \eqpunc , \\
    \vdots \ \ 
      &= \ \ \vdots \quad \ \quad 
             \vdots \quad \ \quad 
             \vdots \quad \ \quad \vdots %
  \end{align*}
  Note that applying the boundary operator $\partial_2$ to $s_0 + s_1$ yields
  \begin{align*}
    \partial_2(s_0 + s_1)
      &= \partial_2(s_0)
         +
         \partial_2(s_1) \\
      &= (a_{01} \! + \! a_{15} \! - \! a_{45} \! - \! a_{04})
         +
         (a_{12} \! + \! a_{26} \! - \! a_{56} \! - \! a_{15}) \\
      &= a_{01} + a_{12} + a_{26} - a_{56} - a_{45} - a_{04} \eqpunc .
  \end{align*}
	\qede
\end{example}

\begin{example}%
  \label{eg:2-complex-square-hole:prelim:1}

  The graph in Fig.~\ref{fig:square2-hole}(a) is similar to the graph in
  Fig.~\ref{fig:square2}(a), but has only $|\setS| = 8$ faces. Namely, the face $s_4$ has been
  removed, thereby leaving a ``hole'' in the plane.
  \qede
\end{example}

The following definition turns out to be useful for subsequent considerations.

\begin{definition}
  \label{def:two:complex:1:extension:1}

  We define the sets $C_3 \defeq \{ 0 \}$ and $C_{-1} \defeq \{ 0 \}$, along
  with the trivial mappings $\partial_3: C_3 \to C_2$ and $\partial_0: C_0 \to
  C_{-1}$.
  \qedd
\end{definition}

\begin{figure}[t]
  \begin{align}
    \label{eq:two:complex:chain:1}
	&\textcolor{gray}{C_{3} } \ 
         \textcolor{gray}{\xrightarrow{\partial_{3}}} \ 
         C_{2} 
         \xrightarrow{\partial_{2}}
         C_{1} 
         \xrightarrow{\partial_{1}} 
         C_{0} \ 
         \textcolor{gray}{\xrightarrow{\partial_{0}}} \ 
         \textcolor{gray}{C_{-1}} \\
    \label{eq:two:complex:cochain:1}
        &\textcolor{gray}{\widehat{C}_{3}} \
         \textcolor{gray}{\xleftarrow{d_{3}}} \ 
         \widehat{C}_{2} 
         \xleftarrow{d_{2}} 
         \widehat{C}_{1} 
         \xleftarrow{d_{1}} 
         \widehat{C}_{0} \ 
         \textcolor{gray}{\xleftarrow{d_{0}}} \
         \textcolor{gray}{\widehat{C}_{-1}}
  \end{align}
  \caption{Spaces and mappings associated with a $2$-complex.}
  \label{fig:two:complex:chain:and:cochain:1}
\end{figure}

The objects that we have introduced so far are summarized in
\eqref{eq:two:complex:chain:1}; the collection of these objects is known as a
$2$-dimensional chain complex, or simply $2$-complex.

Note that the oriented boundary $\partial_2(s)$ of a face $s$ is a cycle,
i.e., it is in $\ker \partial_1$. Therefore, $\img \partial_2 \subseteq
\ker \partial_1$. This, together with the trivial observations
$\img \partial_3 = \{ 0 \}$ and $\ker \partial_0 = C_0$ implies that
\begin{align}
  \img \partial_3
    &\subseteq 
       \ker \partial_2 \eqpunc ,
         \label{eq:two:complex:containment:2} \\
  \img \partial_2 
    &\subseteq 
       \ker \partial_1 \eqpunc ,
         \label{eq:two:complex:containment:1} \\
  \img \partial_1 
    &\subseteq 
       \ker \partial_0 \eqpunc . 
         \label{eq:two:complex:containment:0}
\end{align}

\begin{definition}
  \label{def:homology:spaces:2:1}
 
  The $2$-nd, the $1$-st, and the $0$-th homology spaces are defined to be the
  quotient spaces
  \begin{align}
    H_2
      &\defeq
         \ker\partial_{2} \ / \ \img\partial_{3} \eqpunc ,
      \label{eq:homology:2:1} \\
    H_1
      &\defeq
         \ker\partial_{1} \ / \ \img\partial_{2} \eqpunc ,
      \label{eq:homology:2:1} \\
    H_0
      &\defeq
         \ker\partial_{0} \ / \ \img\partial_{1} \eqpunc ,
      \label{eq:homology:2:0}
  \end{align}
  respectively.
  \qedd
\end{definition}

Of particular interest is $\dim H_1$, which represents the number of ``holes''
of the $2$-complex. Intuitively, the graph $\G = (\V,\A,\setS)$ in
Fig.~\ref{fig:square2}(a) has no hole, whereas the graph $\G = (\V,\A,\setS)$
in Fig.~\ref{fig:square2-hole}(a) has a hole because the latter graph is
missing the surface element $s_4$. This can be made more rigorous by
considering closed curves that pass through the surface elements. In the case
of Fig.~\ref{fig:square2}(a), any such closed curve can be continuously
contracted to a single point, whereas in the case of
Fig.~\ref{fig:square2-hole}(a), this is not possible because of the
missing surface element $s_4$. (For more details, see, e.g.,
\cite[Chap.~2]{hatcher:topology}.)

\begin{example}%
  \label{eg:2-complex-square}

  We continue Example~\ref{eg:2-complex-square:prelim:1}. Consider again the
  graph in Fig.~\ref{fig:square2}(a).
  Similar to Example~\ref{example:one:complex:basis}, 
  the nine cycles
  $\partial_2 s_i$, $i = 0, 1, \ldots, 8$, form a basis of
  of $\ker\partial_1$, and so also form a basis of $\img\partial_{2}$
  since $\img\partial_2 \subseteq \ker\partial_1$.
  Therefore,
  \begin{align}
    \img \partial_{2} 
      &= \ker\partial_1 \eqpunc ,
           \label{eq:exact}
  \end{align}
  which may equivalently be stated as
  \begin{align*}
    \dim H_{1} 
      &= 0,
  \end{align*}
  reflecting the lack of holes in this $2$-complex.  

  We can compute $\dim H_2 = \dim(\ker\partial_{2}) - \dim(\img\partial_{3}) =
  0$ as follows. First, by the rank-nullity theorem, we have
  $\dim(\ker\partial_{2}) = 0$. Second, from the triviality of the map
  $\partial_3$ it follows that $\dim(\img\partial_{3}) = 0$.

  Finally, $\dim H_0$ equals the number of connected components of $\G$, which
  is one.

  These findings can be summarized as follows:\footnote{Here and in the
    following, we omit the trivial spaces $C_{3}$ and $C_{-1}$ when discussing
    $2$-complexes.}
  \begin{align}
	  \label{eq:cw-planar}
    \begin{array}{cccc cc}
      & C_{2} & \xrightarrow{\partial_{2}} 
      & C_{1} & \xrightarrow{\partial_{1}} 
      & C_{0} \\
      \dim C_i & |\setS|  &   & |\A| &   & |\V| \\
      \dim H_i & 0 &   & 0 &   & 1 \eqpunc , 
    \end{array}
  \end{align}
  \qede
\end{example}

\begin{remark}
	Equation~\eqref{eq:cw-planar} holds for any planar graph, provided that $\s$ is taken as the set
	of all inner faces. 
  \qede
\end{remark}

\begin{remark}
	\label{remark:outer-face}
	If we include the outer face, then we will have $\dim H_{2} = 1$. To see this, note that
	adding the outer face will increase the dimension of $C_2$ by one, but leave the dimension of $\img \partial_2$ unchanged since the boundary of
	the added face is contained in the original image space. (In our example, it is not hard to
	verify that the boundary of the outer face is equal to
        $-\partial_2(s_0 + s_1 + \cdots + s_8)$.)
	Hence, as promised, we have $\dim(\ker \partial_2) = 1 = \dim H_2$, where the first equality is by the
	rank-nullity theorem and the second equality is by the triviality of $\partial_3$.
	Finally, the dimensions of $H_0$ and $H_1$ remain unchanged.
	We summarize this as follows.
	\begin{align}
		\begin{array}{cccc cc}
		  \label{eq:cw-planar-outer}
				& C_{2} & \xrightarrow{\partial_{2}} & C_{1} & \xrightarrow{\partial_{1}} & C_{0} \cr
				\dim C_i	& |S| 	   & 		                 & |\A|    & 					    & |\V| \cr
		\dim H_i	& 1 	   & 		                 & 0    & 					    & 1\eqpunc .
		\end{array}
	\end{align}
  \qede
\end{remark}

\begin{figure}[t]
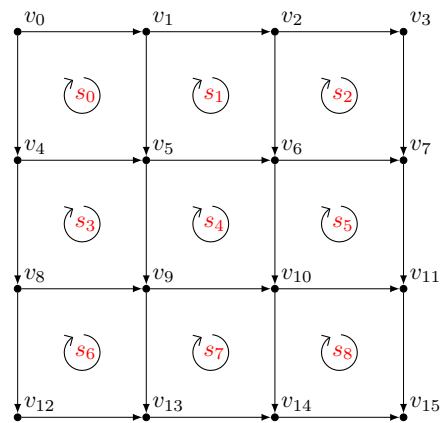
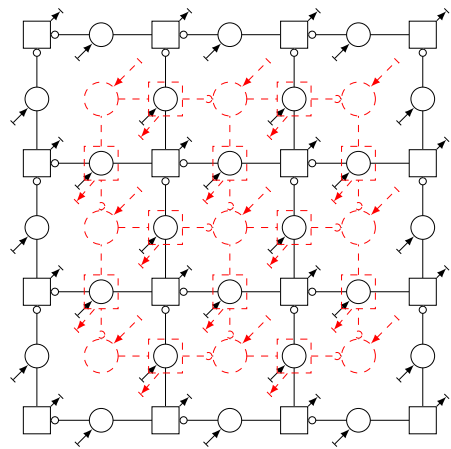
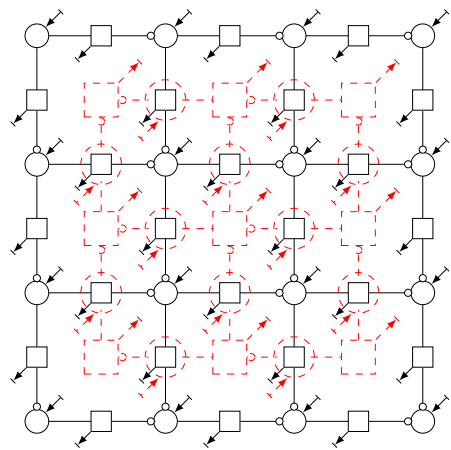

  \def\s{.90}	
  \def\L{3}
  \def\dist{1.9cm}	
  \def\linetype{}
  \def\dongle{.1}
  \def\c{.5}
  \ifonecol
  	\def\xshift{6}
  	\else
  	\def\xshift{4.5}
  \fi
  \def\yshift{-6.6}
  \def\xoffset{.3*\s}
  \def\yoffset{.35*\s}
  \centering
  \begin{tikzpicture}
  	\node(a)[scale=\s] at
  	(0*\xshift,0)
          {\input{figtex/square/square-2complex/square-graph.tex}};
  	\node[vcaption, below = of a]
          {(a) Graph $\G = (\V,\A,\setS)$};

   	\def\xshift{-.25}

   	\node(b)[scale=\s] at
   	(1*\xshift,\yshift)
           {\input{figtex/square/square-2complex/square-nfga.tex}};
   	\def\xoffset{-.362*\s} \def\yoffset{-0*\s}
   	\node[scale=\s] at
   	(0*\xshift+\xoffset,\yshift+\yoffset)
        {\def\clr{red} 
          \def\linetype{dashed}
          \input{figtex/square/square-2complex/square-nfgb-part.tex}};
  	\node[vcaption, below = of b]
          {(b) NFG $\N_{\partial_1}$ (black, solid) and
               NFG $\N_{\partial_2}$ (red, dashed)};

  	\node(c)[scale=\s] at
  	(1*\xshift,2*\yshift)
          {\input{figtex/square/square-2complex/square-nfgb.tex}};
  	\def\xoffset{-.085*\s} \def\yoffset{-0*\s}
  	\node[scale=\s] at
  	(1*\xshift+\xoffset,2*\yshift+\yoffset)
        {\def\clr{red} 
          \def\linetype{dashed}
          \input{figtex/square/square-2complex/square-nfga-part.tex}};
  	\node[vcaption, below = of c]
          {(c) NFG $\N_{d_1}$ (black, solid) and
               NFG $\N_{d_2}$ (red, dashed)};
  \end{tikzpicture}
  \caption{
	  A graph $\G = (\V,\A,\setS)$ and associated NFGs.  To allow a clearer overlay of figures, we
	  deviate from conventions here and use a circle node to denote an equality indicator function
	  and a square node to denote a parity indicator function.
  }
  \label{fig:square2}
\end{figure}

\begin{figure}[t]
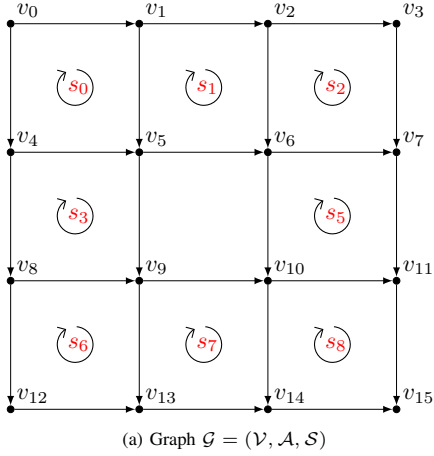
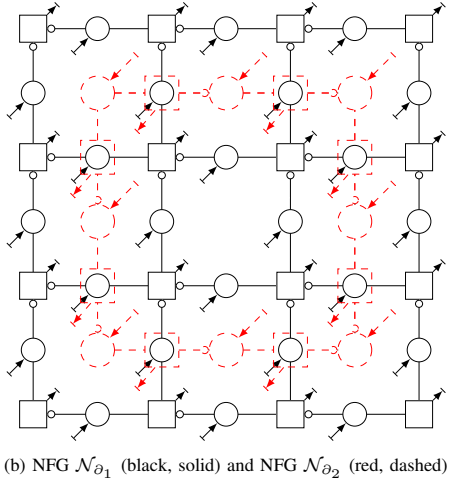
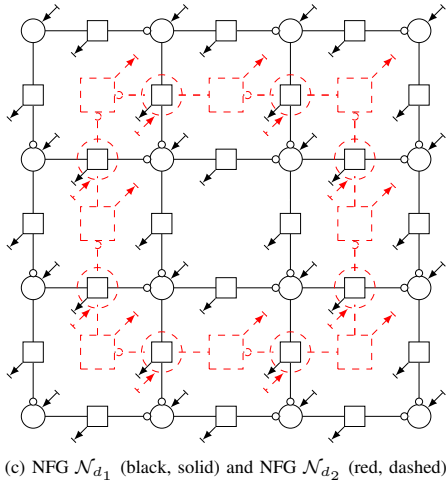

  \def\s{.90}	
  \def\L{3}
  \def\dist{1.9cm}	
  \def\linetype{}
  \def\dongle{.1}
  \def\c{.5}
  \ifonecol
  	\def\xshift{6}
  	\else
  	\def\xshift{4.5}
  \fi
  \def\yshift{-6.6}
  \def\xoffset{.3*\s}
  \def\yoffset{.35*\s}
  \centering
  \begin{tikzpicture}
  	\node(a)[scale=\s] at
  	(0*\xshift,0)
          {\input{figtex/square/square-2complex/square-graph-hole.tex}};
  	\node[vcaption, below = of a]
            {(a) Graph $\G = (\V,\A,\setS)$};

  	\def\xshift{-.25}
  	\node(b)[scale=\s] at
  	(1*\xshift,\yshift)
          {\input{figtex/square/square-2complex/square-nfga.tex}};
  	\def\xoffset{-.361*\s} \def\yoffset{-0*\s}
  	\node[scale=\s] at
  	(0*\xshift+\xoffset,\yshift+\yoffset)
          {\def\clr{red} \def\linetype{dashed}
          \input{figtex/square/square-2complex/square-nfgb-part-hole.tex}};
  	\node[vcaption, below = of b]
          {(b) NFG $\N_{\partial_1}$ (black, solid) and
               NFG $\N_{\partial_2}$ (red, dashed)};

  	\node(c)[scale=\s] at
  	(1*\xshift,2*\yshift)
          {\input{figtex/square/square-2complex/square-nfgb.tex}};
  	\def\xoffset{-.084*\s} \def\yoffset{-0*\s}
  	\node[scale=\s] at
  	(1*\xshift+\xoffset,2*\yshift+\yoffset){\def\clr{red}
        \def\linetype{dashed}
          \input{figtex/square/square-2complex/square-nfga-part-hole.tex}};
  	\node[vcaption, below = of c]
          {(c) NFG $\N_{d_1}$ (black, solid) and
               NFG $\N_{d_2}$ (red, dashed)};
  \end{tikzpicture}
  \caption{
	  A graph $\G = (\V,\A,\setS)$ and associated NFGs.  To allow a clearer overlay of figures, we
	  deviate from conventions here and use a circle node to denote an equality indicator function
	  and a square node to denote a parity indicator function.
  }
  \label{fig:square2-hole}
\end{figure}

\begin{example}%
  \label{eg:2-complex-square-hole}

  We continue Example~\ref{eg:2-complex-square-hole:prelim:1}. Consider again
  the graph in Fig.~\ref{fig:square2-hole}(a).
  Note that this $2$-complex has
  $1$-cycles that cannot be represented as the boundary of some
  $2$-chains. For example, the cycle $a_{56} + a_{6,10} - a_{9,10} -
  a_{59}$ is such a $1$-cycle. This points to the fact that, although the
  eight $1$-cycles $\partial_2 s_i$,
  $i \in \{0\dots,8\} \backslash \{4\}$,
  form a basis of
  $\img\partial_{2}$, they do not form a basis of $\ker\partial_1$.
  We would like to show that
  \begin{align*}
    \dim H_{1} 
      &= 1,
  \end{align*}
  which reflects the ``hole'' in the $2$-complex represented by
  Fig.~\ref{fig:square2-hole}(a).
  This is indeed the case since, starting with the faces as in Example~\ref{eg:2-complex-square},
  then removing $s_4$ does not affect $\ker\partial_{1}$, but it reduces
  $\dim(\img\partial_2)$ by one, since $\partial_{2}(s_i)$, $i=0,\ldots,8$, are independent as
  observed in Example~\ref{eg:2-complex-square}.

  For $\dim H_2$ and $\dim H_0$ we obtain $0$ and $1$, respectively. Overall,
  these findings can be summarized as follows:
  \begin{align}
    \begin{array}{cccc cc}
      & C_{2} & \xrightarrow{\partial_{2}} 
      & C_{1} & \xrightarrow{\partial_{1}} 
      & C_{0} \\
      \dim C_i & |\setS|  &   & |\A| &   & |\V| \\
      \dim H_i & 0 &   & 1 &   & 1 \eqpunc  
    \end{array}
  \end{align}
  \qede
\end{example}

Let us conclude this section by pointing out that there are quantum stabilizer
codes based on $2$-complexes. For these codes, the number of information
qubits equals $\dim H_1$. We refer the interested reader
to~\cite{kitaev2003fault} about the so-called toric (quantum stabilizer)
code. See also the discussion in~\cite{Li:Vontobel:16:1}.

\subsection{NFG representation of the boundary operator}
\label{sec:two:complex:nfg:representation:1}

It is straightforward to generalize the definitions in
Section~\ref{sec:one:complex:nfg:representation:1} to $2$-complexes. In this
subsection, we will therefore only discuss some examples.

\begin{example}
  \label{eg:2-complex-square:2}

  We continue Examples~\ref{eg:2-complex-square:prelim:1}
  and~\ref{eg:2-complex-square}. Fig.~\ref{fig:square2}(b) shows an
  input/output NFG $\N_{\partial_1}$ (black, solid lines) and an input/output
  NFG $\N_{\partial_2}$ (red, dashed lines).
  To allow a clearer overlay of NFGs in this figure, we omit the ``$=$'' and ``$+$'' symbols and draw an equality indicator function as a
  circle node and a parity indicator function as a square node.
  Note that the output half-edges
  of $\N_{\partial_2}$ are in parallel to the input half-edges of
  $\N_{\partial_1}$. If we remove such pairs of half-edges and replace them by
  full-edges connecting $\N_{\partial_2}$ and $\N_{\partial_1}$, we obtain the
  input/output NFG $\N_{\partial_1 \circ \partial_2}$ corresponding to the
  mapping $\partial_1 \circ \partial_2$, which is the concatenation of first
  applying the mapping $\partial_2$ and then the mapping $\partial_1$. Note
  that because $\partial_1 \circ \partial_2 = 0$, which is a consequence of
  $\img \partial_2 \subseteq \ker \partial_1$, this mapping is trivial.
  \qede
\end{example}

\begin{example}
  \label{eg:2-complex-square-hole:2}

  We continue Examples~\ref{eg:2-complex-square-hole:prelim:1}
  and~\ref{eg:2-complex-square-hole}. Fig.~\ref{fig:square2-hole}(b) shows an
  input/output NFG $\N_{\partial_1}$ (black, solid lines) and an input/output
  NFG $\N_{\partial_2}$ (red, dashed lines).
  \qede
\end{example}

It should be clear that once we have drawn $\N_{\mu}$ for some linear mapping
$\mu$, we can easily obtain $\N_{\ker \mu}$ and $\N_{\img \mu}$.

\subsection{Cochains, coboundary operators, and cohomology spaces}
\label{sec:two:complex:dual:domain:1}

For all the objects that we encountered in
Section~\ref{sec:two:complex:primal:domain:1}, there are dual objects that we
now introduce.

\begin{definition}
  \label{def:cochain:1}

  Given a $2$-complex as in Definitions~\ref{def:two:complex:1}
  and~\ref{def:two:complex:1:extension:1}, for all $i = -1, 0, 1, 2, 3$, let
  $\widehat{C}_i$ be the dual space of $C_i$, i.e., the space of all linear
  maps from $C_i$ to $\F$. Elements of $\widehat{C}_i$ are called
  $i$-cochains. Note that $\widehat{C}_{-1} = \{ 0 \}$ and $\widehat{C}_{3} = \{
  0 \}$.
  \qedd
\end{definition}

For $i = 0, 1, 2, 3$, the coboundary operator $d_i$ is defined analogously to
Definition~\ref{def:coboundary:1}. For $i = 2$, the condition
in~\eqref{eq:def:coboundary:1:condition:1} implies the following for an
arbitrary face $s \in \setS$ and an arbitrary $\varphi \in \widehat{C}_{1}$:
\begin{align*}
  (d_2 \varphi)(s) 
    &= \varphi \bigl( \partial_2(s) \bigr) 
     = \varphi \left( \sum_{a \in A_s} \!\!  \alpha(a) \cdot a \right) 
     = \!\! \sum_{a \in A_s} \!\!  \alpha(a) \cdot \varphi(a) \eqpunc ,
\end{align*}
where we have used the notation as in
Definition~\ref{def:two:complex:boundary:operator:1}. If $\varphi$ represents
some flow along the directed edges, then $d_2 \varphi$ is the function that
yields the curl around oriented faces.

Note that if a flow is obtained by applying the coboundary operator $d_1$ to a
$0$-chain representing some potential function, then the curl around any
oriented face will be zero. This implies that $\img d_1 \subseteq \ker
d_2$. (This generalizes a well-known result from vector calculus which states
that $\operatorname{curl} \circ \operatorname{grad} = 0$. See \cite[Section~15.4]{bamberg:course}
for a more general discussion.) Because $\img d_0 =
\{ 0 \}$ and $\ker d_3 = \widehat{C}_2$, we therefore have
\begin{align}
  \img d_0
    &\subseteq 
       \ker d_1 \eqpunc ,
         \label{eq:two:complex:d:containment:0} \\
  \img d_1 
    &\subseteq 
       \ker d_2 \eqpunc ,
         \label{eq:two:complex:d:containment:1} \\
  \img d_2
    &\subseteq 
       \ker d_3 \eqpunc . 
         \label{eq:two:complex:d:containment:2}
\end{align}
The fact that there might be a gap between $\img d_i$ and $\ker d_{i+1}$, $i =
0, 1, 2$, is captured by the $i$-th cohomology space.

\begin{definition}
  The $0$-th, the $1$-st, and the $2$-nd cohomology spaces are defined to be
  the quotient spaces
  \begin{align}
    \widehat{H}_0
      &\defeq
         \ker d_{1} \ / \ \img d_{0} \eqpunc ,
      \label{eq:cohomology:1:0} \\
    \widehat{H}_1
      &\defeq
         \ker d_{2} \ / \ \img d_{1} \eqpunc ,
      \label{eq:cohomology:1:1} \\
    \widehat{H}_2
      &\defeq
         \ker d_{3} \ / \ \img d_{2} \eqpunc ,
      \label{eq:cohomology:1:2}
  \end{align}
  respectively.
  \qedd
\end{definition}

A similar argument to the one in Section~\ref{sec:one:complex:dual:domain:1} gives $\dim
\widehat{H}_i = \dim H_i$ for $i = 0, 1, 2$,
and so, $\dim \widehat{H}_i$  has the same interpretation as $\dim H_i$.

The objects that we have introduced in this subsection are summarized in
\eqref{eq:two:complex:cochain:1}; the collection of these objects is known as
a $2$-dimensional cochain complex.

An important application of these homology and cohomology spaces is to
characterize continuous $2$-dimensional surfaces. This can be done by
triangulating these surfaces and studying the dimension of the resulting
homology and cohomology spaces. Most importantly, for a given $i$, $\dim H_i$
will be independent of the chosen triangulation.

\subsection{NFG representation of the coboundary operator}
\label{sec:two:complex:nfg:representation:2}

It is straightforward to generalize the definitions in
Section~\ref{sec:one:complex:nfg:representation:2} to $2$-complexes. In this
subsection, we will therefore only discuss some examples.

\begin{example}
  \label{eg:2-complex-square:3}

  We continue Examples~\ref{eg:2-complex-square:prelim:1},
  \ref{eg:2-complex-square}, and~\ref{eg:2-complex-square:2}.
  Fig.~\ref{fig:square2}(c) shows an input/output NFG $\N_{d_1}$ (black, solid
  lines) and an input/output NFG $\N_{d_2}$ (red, dashed lines). Note that the
  output half-edges of $\N_{d_1}$ are in parallel to the input half-edges of
  $\N_{d_2}$. If we remove such pairs of half-edges and replace them by
  full-edges connecting $\N_{d_1}$ and $\N_{d_2}$, we obtain the input/output
  NFG $\N_{d_2 \circ d_1}$ corresponding to the mapping $d_2 \circ d_1$. Note
  that because $d_2 \circ d_1 = 0$, which is a consequence of $\img d_1
  \subseteq \ker d_2$, this mapping is trivial.
  \qede
\end{example}

\begin{example}
  \label{eg:2-complex-square-hole:3}

  We continue Examples~\ref{eg:2-complex-square-hole:prelim:1},
  \ref{eg:2-complex-square-hole},
  and~\ref{eg:2-complex-square-hole:2}. Fig.~\ref{fig:square2-hole}(c) shows
  an input/output NFG $\N_{d_1}$ (black, solid lines) and an input/output NFG
  $\N_{d_2}$ (red, dashed lines).
  \qede
\end{example}

We conclude this subsection by noting that an example where an NFG appeared
whose exterior function is proportional to $\img d_1$ is the support NFG of
Fig.~11 in~\cite{vontobel:electrical}. In that paper, factor-graph
representations of electrical networks are considered, and $\ker d_2$ and
$\img d_1$ appear naturally in that context because $\ker d_2$ encodes exactly
Kirchhoff's voltage law and $\img d_1$ gives voltage differences along edges
based on voltage potentials at nodes. Because $\img d_1 \subseteq \ker d_2$,
these voltage differences automatically satisfy Kirchhoff's voltage law.

\subsection{$2$-torus lattice graph}
\label{sec:cw-torus}

Because of the importance of $2$-torus lattice graphs for later sections, this
subsection discusses this object in detail.

\begin{figure}[t]
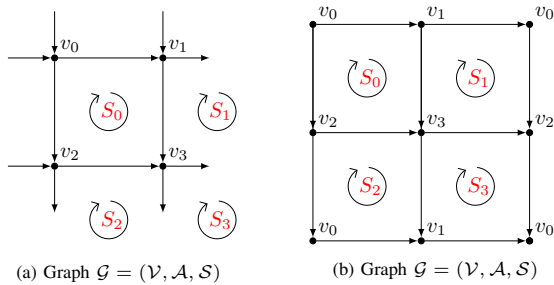
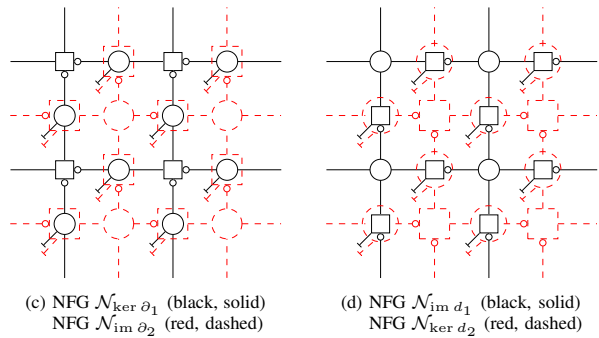

  \def\s{.83}	
  \def\L{1}
  \def\dist{1.8cm}	
  \def\linetype{}
  \def\dongle{.1}
  \def\c{.5}
  \def\minsize{2.5mm}
  \ifonecol
  	\def\xshift{6}
  \else
  	\def\xshift{4}
  \fi
  \def\yshift{-5}
  \centering
  \begin{tikzpicture}
    \def\s{.80}
    \ifonecol \def\xshift{6.2} \else \def\xshift{4.2} \fi
    \def\xoffset{-.246*\s} \def\yoffset{0*\s}

    \node(a)[scale=\s] at (0,0)
      {\input{figtex/torus2/torus-graph-variable.tex}};
    \node(al)[vcaption, below=of a]
             {(a) Graph $\G = (\V,\A,\setS)$};

    \node(b)[scale=\s] at
      (1.05*\xshift+\xoffset,\yoffset)
      {\input{figtex/torus2/torus-graph2-variable.tex}};
    \node(bl)[vcaption, below=of b]
             {(b) Graph $\G = (\V,\A,\setS)$};
     
    \node(c)[scale=\s] at (0,\yshift){\input{figtex/torus2/torus-r-dnfg.tex}};
    \node[scale=\s] at
      (\xoffset,\yshift+\yoffset)
    {\def\linetype{dashed} 
      \def\clr{red} 
      \input{figtex/torus2/torus-l-pnfg.tex}};
    \node[vcaption, below=of c]
         {(c) \parbox[t]{0.32\linewidth}{
                NFG $\N_{\ker\partial_1}$ (black, solid) \\
                NFG $\N_{\img\partial_2}$ (red, dashed)
              }};

  	\node(d)[scale=\s] at (\xshift,\yshift)
          {\input{figtex/torus2/torus-r-pnfg.tex}};
  	\node[scale=\s] at
  	(\xshift+\xoffset,\yshift+\yoffset)
        {\def\linetype{dashed}
          \def\clr{red} \input{figtex/torus2/torus-l-dnfg.tex}};
  	\node[vcaption, below=of d]
         {(d) \parbox[t]{0.32\linewidth}{
                NFG $\N_{\img d_1}$ (black, solid) \\
                NFG $\N_{\ker d_2}$ (red, dashed)
              }};
  \end{tikzpicture}
  \caption{
	  A $2$-torus lattice graph $\G = (\V,\A,\setS)$ and its associated NFGs.
	  To allow a clearer overlay of figures, we deviate from conventions here and use a circle node
	  to denote an equality indicator function and a square node to denote a parity indicator
	  function. In (c) and (d) the left and right most edges are identified and the top and
	  bottom most edges are identified.
  }
  \label{fig:torus}
\end{figure}

We define a $2$-torus lattice graph as a $2$-dimensional square lattice graph
drawn on a torus.  Namely, the graph $\G$ in Fig.~\ref{fig:torus}(a) is a
$2$-torus lattice graph, which is obtained from the graph in
Fig.~\ref{fig:square} by identifying the left and right most edges (resulting
in a cylinder) and the top and bottom most edges (resulting in a doughnut
shape).  More explicitly, the graph $\G \defeq (\V,\A,\setS)$ in
Fig.~\ref{fig:torus}(a) is obtained from the graph in Fig.~\ref{fig:square} as
follows:
\begin{itemize}

\item $v_0$ is obtained by identifying $v_0$, $v_2$, $v_6$, $v_8$;

\item $v_1$ is obtained by identifying $v_1$, $v_7$;

\item $v_2$ is obtained by identifying $v_3$, $v_5$;

\item $v_3$ is obtained from $v_4$;

\item edges in Fig.~\ref{fig:torus}(a) are obtained by identifying edges in
  Fig.~\ref{fig:square} in a manner consistent with the above vertex
  identifications.

\end{itemize}
In other words, the resulting graph in Fig.~\ref{fig:torus}(a) is such that
the edge in Fig.~\ref{fig:torus}(a) that leaves $v_2$ downwards continues as
the edge that enters $v_0$ from above, etc.
Note that $|\V| = 4$, $|\A| = 8$, and $|\setS| = 4$.

An alternative representation of the graph in Fig.~\ref{fig:torus}(a) is shown
in Fig.~\ref{fig:torus}(b), where identically labeled vertices and edges are
identified.

Figs.~\ref{fig:torus}(c) and~(d) show some of the NFGs that can be associated
with the graph in Fig.~\ref{fig:torus}(a).

\begin{lemma}
  \label{lemma:torroidal:lattice:dim:h1:1}

  For the above-defined $2$-torus lattice graph it holds that
  \begin{align*}
    \dim H_1
      &= 2.
  \end{align*}
  \qedl
\end{lemma}

\begin{proof}
  We first show that $\dim(\ker\partial_{1}) = 5$ and that
  $\dim(\img\partial_{2}) = 3$. This then yields
  \begin{align*}
    \dim H_1
      &= \dim( \ker\partial_{1} \ / \ \img\partial_{2} ) \\
      &= \dim( \ker\partial_{1} ) - \dim( \img\partial_{2} ) \\
      &= 5 - 3 \\
      &= 2 \eqpunc ,
  \end{align*}
  which is the promised result. The result $\dim(\ker\partial_{1}) = 5$
  follows immediately from \eqref{eq:dim-ker-bd1} as $\dim(\ker\partial_{1}) = |\A| - (|\V| - 1) =
  5$. On the other hand, the result $\dim(\img\partial_{2}) = 3$ follows from
  observing that the boundary of a $2$-chain is equal to zero if and only if
  it is (up to a scaling factor in $\F$) equal to the sum of all the faces,
  i.e., $\dim(\ker\partial_{2}) = 1$. Therefore, $\dim(\img\partial_{2}) =
  |\setS| - 1 = 3$.
\end{proof}

The fact that $\dim H_1 > 0$ is an indication that the graph $\G \defeq
(\V,\A,\setS)$ has holes, i.e., there are cycles in $C_1$ that cannot be
written as boundaries of $2$-chains. The space $H_1$ consists of the
equivalence classes of such cycles, where two cycles are equivalent if they
differ by a boundary of a $2$-chain. Because $\dim H_1 = 2$,
a basis of $H_1$ consists of two such equivalence classes.

Let $\{c_{\rm h}+\img\partial_{2}, \ c_{\rm v}+\img\partial_{2}\}$ be a basis
of $H_{1}$, where $c_{\rm h}$, $c_{\rm v}\in \ker\partial_{1} \setminus
\img\partial_{2}$ are two non-equivalent cycles that are not boundaries, say,
\begin{align*}
  c_{\rm h}
    &\defeq 
       a_{01} + a_{10}\eqpunc , \\
  c_{\rm v}
    &\defeq 
       a_{02} + a_{20}\eqpunc .
\end{align*}
(Here ``$\rm h$'' and ``$\rm v$'' stand for ``horizontal'' and ``vertical''
w.r.t. Figs.~\ref{fig:torus}(a) or (b).)  Then one can express the space of
$1$-cycles as the union of the cosets
\begin{align*}
  \ker\partial_{1}
    &= \bigcup_{\alpha_{\rm h},\alpha_{\rm v} \in \F} \ 
         (
           \alpha_{\rm h} \cdot c_{\rm h} 
           + 
           \alpha_{\rm v} \cdot c_{\rm v}
           +
           \img\partial_{2}
         ) \eqpunc .
\end{align*}
Fig.~\ref{fig:torus}(c) shows the NFGs representing $\img\partial_{2}$
(dashed) and $\ker\partial_{1}$ (solid).

The above findings can easily be generalized to the $2$-torus lattice graph with
$n\defeq L_1\times L_2$ vertices, where $L_1$ and $L_2$ are arbitrary 
integers larger than one.
(This $2$-torus lattice graph can be obtained from the $(L_1+1)\times (L_2+1)$ square lattice using
a similar identification of edges and vertices as described at the beginning of the
subsection.)
In particular, we find that $\dim(\img\partial_{2})
= n - 1$ and $\dim(\ker\partial_{1}) = n + 1$, which imply $\dim H_1 = 2$ for
all $L_1$ and all $L_2$. Moreover, using the obvious extension of the
vertices' indexing, i.e., by indexing the first row of vertices by $0, \ldots,
{L_{2}-1}$, the second row by ${L_2}, \ldots, {2L_2-1}$, etc., one may choose
\begin{align*}
  c_{\rm h} 
    &= a_{0,1} + a_{1,2} + \cdots + a_{L_{2}-1,0} \eqpunc , \\
  c_{\rm v}
    &= a_{0,L_2} + a_{L_2,2L_2} + \cdots + a_{(L_1-1)L_2,0} \eqpunc ,
\end{align*}
such that
\begin{align}
  \ker\partial_{1}
    &= \bigcup_{\alpha_{\rm h},\alpha_{\rm v} \in \F} \ 
         (
           \alpha_{\rm h} \cdot c_{\rm h} 
           + 
           \alpha_{\rm v} \cdot c_{\rm v}
           +
           \img\partial_{2}
         ) \eqpunc .
           \label{eq:coset:2}
\end{align}

Finally, since $\dim H_0 = 1$ for a connected graph and $\dim H_{2} = \dim
\left( \ker\partial_{2} \right)$ for a $2$-complex, we may summarize our
discussion of the $2$-torus lattice graph by
\begin{align}
	\label{eq:cw-torus-summary}
  \begin{array}{cccccc}
      & C_{2} & \xrightarrow{\partial_{2}} 
      & C_{1} & \xrightarrow{\partial_{1}}  
      & C_{0} \\
    \dim C_i	
      & n     &   & 2n    &    & n \\
    \dim H_i	& 1 &      & 2  &   & 1 \eqpunc .
  \end{array}
\end{align}

We conclude this section with the following observation that will be useful in
Section~\ref{sec:kw:duality:ising:model:1}. Namely, for the $2$-torus lattice graph, we have
\begin{align}
	\label{eq:torus-selfdual}
	\img \partial_2  = \img d_1.
\end{align}
This can be seen by comparing Figs.~\ref{fig:torus}(c) (red, dashed lines) and
(d) (black, solid lines), and is a consequence of the \emph{self-duality} of
the $2$-torus lattice graph. Namely, for any graph $\G = (\V,\A,\s)$ the
\emph{dual graph} $\G_d \defeq (\s, \A_d, \V)$ is defined such that any two
vertices in $\G_d$ are adjacent iff the corresponding surfaces in $\G$ share
an edge. (One can properly define directions for the edges, but we omit the
details.) With this definition, one can verify that the graph $\G$ in
Fig.~\ref{fig:torus}(a) is isomorphic to its dual, which consequently leads to
\eqref{eq:torus-selfdual}.

We conclude this section on $2$-complexes by pointing out
Appendix~\ref{sec:cw-complex} on higher-order complexes.

\section{Statistical Models}
\label{sec:statistical}

In this section we offer a brief introduction to statistical models, in
particular to the Boltzmann distribution and the partition function.

A \emph{statistical model} is defined as a collection of random variables
(RVs) $\{ X_{1}, \dots, X_{n} \}$, where each RV assumes values from a finite
set $\cal X$. (Frequently, $X_{1}, \dots, X_{n}$ are called spins.) With each
configuration $x\in {\cal X}^{n}$ we associate an energy level $E(x)$ such
that the joint distribution of the RVs is the Boltzmann distribution
\begin{align}
  p(x) 
    &\defeq
       \frac{\operatorname{e}^{-\beta E(x)}}{Z} \eqpunc ,
         \label{eq:boltzmann:distribution:1}
\end{align}
where $Z$ is the partition function defined as
\begin{align*}
  Z 
    &\defeq  
        \sum_{x\in {\cal X}^{n}} 
          \operatorname{e}^{-\beta E(x)} \eqpunc ,
\end{align*}
where $\beta\defeq \frac{1}{kT}$ is the inverse temperature, and where $k$ and
$T$ are, respectively, the Boltzmann constant and the temperature.

At first sight, $Z$ is ``only'' a temperature-dependent normalization constant
that appears in~\eqref{eq:boltzmann:distribution:1}. However, the way $Z$
changes with temperature $T$ (and more generally with other parameters like
pressure) tells a lot about how macroscopic properties of a system change with
changing temperature. In particular, in the limit $n \to \infty$ (the
so-called thermodynamic limit), $Z$ or its derivatives might exhibit
discontinuities, thereby delineating phase
transitions~\cite{baxter:exactly}.

Note that
the partition function not only appears in statistical physics in the study of
macroscopic properties induced by the microscopic properties, but, among other
areas, also in the following fundamental problems:
\begin{itemize}

\item the capacity of a constrained channel in information theory (see, e.g.,
  \cite{kato-zeger:capacity}),

\item the number of graph (vertex) colorings in graph theory (see, e.g.,
  \cite{whitney:coloring}),

\item permanents, graph homomorphism, integer flows, etc. (see, e.g.,
  \cite{barvinok:16:1}.

\end{itemize}

Given the importance of the partition function in various fields, several
approaches have been devised to tackle the partition function from different
angles. This has resulted in a variety of methods that have provided:
\begin{itemize}

\item estimates of the partition function, e.g., stochastic approximations
  \cite{potamianos:stochastic}, the renormalization group approximation
  \cite{kadanoff:renorm, wilson:renorm}, and the Bethe approximation
  \cite{bethe};

\item bounds on the partition function \cite{wainwright:bounds};

\item identities that the partition function satisfies, such as the
  Kramers--Wannier duality \cite{kramers-wannier}.

\end{itemize}

A central topic of the remaining sections will be the re-derivation of the
Kramers--Wannier duality for the two-dimensional Ising model with the help of
NFGs and the concepts introduced in earlier sections. The use of some notions
from algebraic topology toward proving the Kramers--Wannier duality is not new
\cite{savit:duality, druhl:duality}, however, the separation of the arguments
leading to the Kramers--Wannier duality into two steps, as discussed below,
seems natural, and concisely points out what is needed for a
Kramers--Wannier-type duality to hold. (See also
Fig.~\ref{fig:kw-duality:simplified:1}.) Moreover, Kramers--Wannier duality is
typically argued in the limit when the size of the square lattice (on the
torus) is large. In contrast, the results provided here hold for any size of
the lattice.

\section{Kramers--Wannier Duality \\
              for the Ising Model}
\label{sec:kw}

In this section we first introduce the Ising model. Afterwards, we suitably
combine the results that we have encountered in earlier sections toward
re-deriving the Kramers--Wannier duality for that model.

\subsection{Ising model}
\label{sec:ising:model:1}

Let $L$ be an arbitrary positive integer and define $n \defeq L^2$. The Ising
model \cite{ising1925} (see, e.g., \cite{baxter:exactly}) is a statistical model with binary
spins, where, w.l.o.g., we will assume ${\cal X}\defeq \F_{2}$.  Moreover, the
spins are arranged over some lattice vertices (called sites). The
two-dimensional nearest-neighbor (ferromagnetic) Ising model is one where the
lattice is chosen as the square lattice of size $L \times L$ (on the plane or
the torus), and the energy of a configuration $x\in \F_{2}^{n}$ is defined as
\begin{align}
  E(x)
    &\defeq 
       -
       \sum_{(i,j)\in \A} 
         \big(
           2\delta_{=}(x_i,x_j)
           -
           1
         \big)\eqpunc,
  \label{eq:energy-ising}
\end{align}
where $\A$ is the set of edges of the lattice.\footnote{The convention taken
  by physicists in adopting the minus sign in (\ref{eq:energy-ising}) is so
  that the configurations with aligned spins (i.e., the configurations where
  all spins are equal) have the lowest energy level.} Subsequently, unless
specified otherwise, we will refer to this model simply as the
(two-dimensional) Ising model. In this case, the partition function can be written
as
\begin{align}
  \label{eq:z-ising-1}
  Z &= \sum_{x\in \F_2^n} 
         \prod_{(i,j)\in \A} 
           \kappa_{\beta}(x_j-x_i) \eqpunc ,
\end{align}
where 
\begin{equation}
\begin{aligned}
	\label{eq:kappa}
  \kappa_{\beta}(0)& \defeq e^{\beta} \eqpunc , \\
  \kappa_{\beta}(1)& \defeq e^{-\beta} \eqpunc .
\end{aligned}
\end{equation}
If $\beta$ is clear from the context, then we will simply write $\kappa$
instead of $\kappa_{\beta}$.

Rewriting the r.h.s. of (\ref{eq:z-ising-1}) more explicitly as
\begin{align}
  \label{eq:z-ising-1.2}
  \sum_{x_{\V}\in \F_2^n} \sum_{x_{\A}\in \F_{2}^{2n}} 
         \prod_{e \defeq (i,j)\in \A} 
			\kappa_{\beta}(x_e) \cdot \delta_{+}(x_e,x_i,-x_j) \eqpunc ,
\end{align}
it is clear that
\begin{align}
	\label{eq:z-pnfg}
  Z
    &= Z_{\Nisingbeta},
\end{align}
where $\Nisingbeta$ is the pairwise interaction NFG in
Fig.~\ref{fig:ising-pnfg} (see Definition~\ref{assumption:nfg:2}).  In other words, the NFG $\Nisingbeta$ represents
the Ising model on the torus at inverse temperature $\beta$, where, as was
discussed in Section~\ref{sec:interaction}, the sites are represented by
equality indicator functions and the parity indicator functions ensure that
the interaction between a pair of adjacent sites depends only on the
difference (modulo $2$) between their spins.  From \eqref{eq:z-pnfg}, and by
recalling the definition of $\code_{\N}$
from~\eqref{eq:set:proj:valid:configuration:1}, we note
that~\eqref{eq:z-ising-1} can also be written as
\begin{align}
  \label{eq:z-ising-2}
  Z &= 2 
         \cdot
         \sum_{x \in \code_{\Nisingbeta}} 
           \prod_{e\in \A}
             \kappa_{\beta}(x_{e}) \eqpunc .
\end{align}
The factor $2$ stems from the fact that for every term in the sum
$\sum_{x \in \code_{\Nisingbeta}} \!\! \cdots$ in~\eqref{eq:z-ising-2}
there are $|\F_{2}|=2$
corresponding terms in the sum
$\sum_{x\in \F_2^n} \cdots$ in~\eqref{eq:z-ising-1}.

\begin{figure}[t]
  \def\s{1}	
  \def\c{.56}	
  \def\L{2}	
  \def\dist{2cm}
  \centering
  \def\yshift{-7cm}
  \def\xshift{9cm}
  \begin{tikzpicture}[scale=\s]
  	\node(a) at (0,0){\input{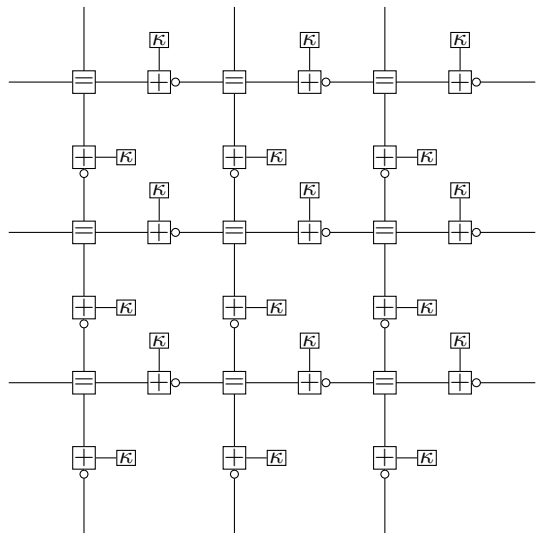}};
  \end{tikzpicture}
  \caption{The NFG $\Nisingbeta$ representing the Ising model on the
    torus. (Note: the leftmost edges are identified with the rightmost edges
    and the topmost edges are identified with the bottommost edges. Moreover,
    $\kappa$ is shorthand notation for $\kappa_{\beta}$.)}
  \label{fig:ising-pnfg}
\end{figure}

\begin{figure}[t]
	\def\s{1}	
	\def\c{.56}	
	\def\L{2}	
	\def\dist{2cm}
	\centering
	\def\yshift{-7cm}
	\begin{tikzpicture}[scale=\s]
		\node(a) at (0,0){\input{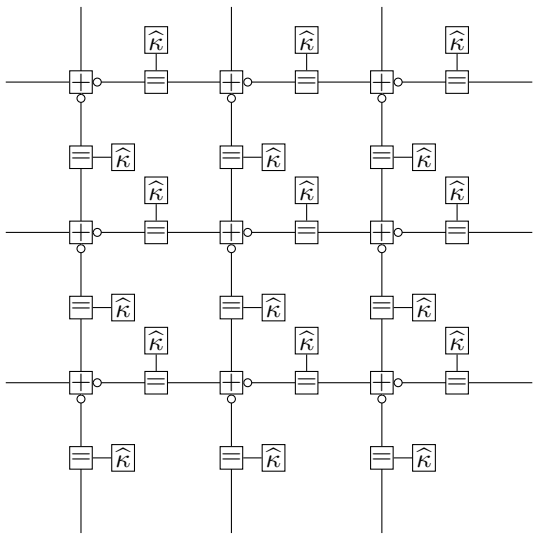}};
	\end{tikzpicture}
	\caption{The Fourier-transformed NFG $\FTNisingbeta$. (Note: the
          leftmost edges are identified with the rightmost edges and the
          topmost edges are identified with the bottommost edges.
          $\widehat{\kappa}$ is shorthand notation for
          $\widehat{\kappa_{\beta}}$.)}
	\label{fig:ising-dnfg}
\end{figure}

\subsection{Fourier-transformed NFG of Ising model}
\label{sec:ft-ising}
In the previous subsection we saw that the NFG $\Nisingbeta$ represents the partition function of the Ising model. From
Section~\ref{sec:fourier}, another
obvious NFG that represents the partition function of the Ising model is the Fourier-transformed NFG
$\FTNisingbeta$ in Fig.~\ref{fig:ising-dnfg}. Namely, we have
\begin{align}
	Z 
	\stackrel{\eqref{eq:z-pnfg}}{=}
	Z_{\Nisingbeta} = 2^{-n} Z_{\FTNisingbeta} \eqpunc,
\end{align}
where
the second equality is by \eqref{thm:ft-z} of Theorem~\ref{thm:nfg-dual} upon
noting that $|\A| - |\V| = n$ for the torus.  This expression of the partition function can be written more
explicitly as
\begin{align*}
	Z = 2^{-n} \sum_{x\in \code_{\widehat{N}_{I(\beta)}}} \prod_{e\in \A}\widehat{\kappa_{\beta}}(x_{e}).
\end{align*}

We now make the following observations about $\FTNisingbeta$.
\begin{itemize}
	\item A degree-one function in $\FTNisingbeta$ is the Fourier transform
		\begin{equation}
		\begin{aligned}
			\label{eq:kappa-hat}
			\widehat{\kappa_{\beta}}(0)
			&= e^{\beta} + e^{-\beta} \eqpunc , \\
			\widehat{\kappa_{\beta}}(1)
			&= e^{\beta} - e^{-\beta} %
		\end{aligned}
		\end{equation}
		of $\kappa_{\beta}$.  Note that at large $\beta$, i.e., low
                temperature, we have $\kappa_{\beta}(0) \gg
                \kappa_{\beta}(1)$. Consequently, the function
                $\prod_{e}\kappa_{\beta}(x_e)$, which is defined on a
                high-dimensional space, is strongly irregular, i.e., it
                contains high peaks and deep valleys, making it harder for
                sampling-based methods to provide an accurate estimate of the
                partition function.  

                In contrast, the Fourier transform shows the opposite
                trend. Namely, at low temperature, we have
                $\widehat{\kappa_{\beta}}(0) \approx
                \widehat{\kappa_{\beta}}(1)$. Consequently, the function
                $\prod_{e}\widehat{\kappa_{\beta}}(x_e)$, which is also
                defined on a high-dimensional space, is almost flat, making it
                easier for sampling-based methods to provide an accurate
                estimate of the partition function.

                This observation is reversed at small values of $\beta$, i.e.,
                high temperature, where sampling estimators based on the
                interaction NFG $\N_{I(\beta)}$ outperform estimators based on
                the Fourier-transformed NFG $\widehat{\N}_{I(\beta)}$. For
                more details, we refer the reader to
                \cite{jerrum:subgraph-world, mehdi:ising, ay:partition}.
	\item
		The above trend carries on beyond the
		Ising model \cite{ay:partition}. However, in the case of the
                Ising model one can make the above argument more explicit.
		Namely, from \eqref{eq:kappa} and \eqref{eq:kappa-hat}, we can write
		$\widehat{\kappa_{\beta}}$ as
		\begin{align}
			\label{eq:kappa-kappa-hat}
			\widehat{\kappa_{\beta}}(x) = \sqrt{2\cdot c_{\beta}} \cdot \kappa_{\widetilde \beta}(x), \ \forall x\in \F_2,
		\end{align}
		where
	  \begin{align}
			c_{\beta}
			&\defeq 2 \cdot \sinh(\beta) \cdot \cosh(\beta)
			\label{eq:def:c:beta:1}
	  \end{align}
	  and the \emph{dual inverse temperature} $\widetilde \beta$ is defined as
		\begin{align}
			\widetilde{\beta}
			&\defeq - \frac{1}{2} \log \left( \tanh\beta \right) \eqpunc .
			\label{eq:def:dual:inverse:temperature:1}
		\end{align}
		Note that $c_{\beta}$ does not depend on $x$ and so can be carried outside the summation when
		computing the partition function. Moreover, the dual inverse temperature is a strictly
		decreasing function of $\beta$, i.e., the function $-\frac{1}{2}\log(\tanh(\cdot))$ maps high
		temperatures to low temperatures, and vice versa.

              \item Although the Fourier transform maps low temperatures to
                high temperatures, the resulting NFG
                $\widehat{\N}_{I(\beta)}$, unlike $\N_{I(\beta)}$, is not a
					 pairwise interaction NFG (see Definition~\ref{assumption:nfg:2}). The task of mapping
                $\widehat{\N}_{I(\beta)}$ to a pairwise interaction NFG at
                high temperature will be discussed in the next
                subsection. (Even though this last step, i.e., from
                $\widehat{\N}_{I(\beta)}$ to a pairwise interaction NFG at
                high temperature, may not be required if one is interested
                in sampling-based approaches for estimating the partition
                function, from a theoretical perspective it can lead to very
                valuable insights.)

\end{itemize}

\subsection{Kramers--Wannier duality}
\label{sec:kw:duality:ising:model:1}

Consider an Ising model on the torus, which can be represented by an NFG
$\Nisingbeta$ as in Fig.~\ref{fig:ising-pnfg}. Kramers and
Wannier~\cite{kramers-wannier} (see also \cite{savit:duality, druhl:duality})
made the observation that the
partition function of the Ising model at inverse temperature $\beta$ can be
expressed in terms of the partition function of the Ising model at the dual inverse
temperature $\widetilde{\beta}$ %
defined in \eqref{eq:def:dual:inverse:temperature:1}.

More precisely, one can make the following statement.

\begin{theorem}
  \label{theorem:kw:duality:limit:n:1}

  For the Ising model on the $2$-torus lattice graph of size $n$, we have
  \begin{align}
    \lim_{n \to \infty}
      \frac{1}{n}
        \log\bigl( Z_{\Nisingbeta} \bigr)
      &= \log(c_{\beta})
         +
         \lim_{n \to \infty}
           \frac{1}{n}
             \log\bigl( Z_{\Nisingbetatilde} \bigr) \eqpunc ,
               \label{eq:theorem:kw:duality:limit:n:1}
  \end{align}
  where $c_{\beta}$ and $\widetilde \beta$ were defined in \eqref{eq:def:c:beta:1} and \eqref{eq:def:dual:inverse:temperature:1}, respectively.
  \qedt
\end{theorem}

Our approach to proving the above result consists of first proving a more
general result that holds for finite $n$ and then to take the limit $n \to
\infty$.

Recall the definition of the cycles $c_{\rm h}$ and $c_{\rm v}$ in Section~\ref{sec:cw-torus}.
In order to state our result, we use the following notation. For any real number $\alpha \geq 0$,
we define $\N^{\rm h}_{I(\alpha)}$ to be the same NFG as $\N_{I(\alpha)}$ but with the interaction
function
\[
	\kappa_{\alpha}(\cdot) \text{ replaced by } \kappa_{\alpha}(1-\cdot) %
\]
along the cycle $c_{\rm h}$.
We similarly define $\N^{\rm v}_{I(\alpha)}$, where the replacement is made along the cycle
$c_{\rm v}$; and $\N^{\rm hv}_{I(\alpha)}$, where the replacement is made along both $c_{\rm h}$ and
$c_{\rm v}$.

\begin{theorem}
  \label{theorem:kw:duality:1}

  For the Ising model on the  $2$-torus lattice graph of size $n$, we have
  \begin{align}
    Z_{\Nisingbeta}
      &= \frac{c_{\beta}^n}{2} 
         \cdot 
         \left(
           Z_{\Nisingbetatilde} 
           + 
           Z_{\Nisingbetatildeh}
           + 
           Z_{\Nisingbetatildev} 
           + 
           Z_{\Nisingbetatildehv} 
         \right) \eqpunc ,
           \label{eq:theorem:kw:duality:1}
  \end{align}
  where $c_{\beta}$ and $\widetilde \beta$ were defined in \eqref{eq:def:c:beta:1} and
  \eqref{eq:def:dual:inverse:temperature:1}, respectively.
  \qedt
\end{theorem}
\begin{proof}
	First note that 
	\begin{align}
		\label{eq:code-im-d1-a}
		\code_{\N_{I(\beta)}} = \img d_1,
	\end{align}
	which is clear by comparing the support NFG of $\N_{I(\beta)}$ in Fig.~\ref{fig:ising-pnfg} and Fig.~\ref{fig:torus}(d) (black, solid lines). 
  We prove the theorem in two steps, namely, we separately show that
  \begin{align}
    Z_{\Nisingbeta}
	 &\!\!=
			2^{-n}
			\! \cdot
         Z_{\FTNisingbeta} \eqpunc ,
           \label{eq:kw:proof:part:1}
           \\
    Z_{\FTNisingbeta}
      &\!\!=
			2^{n-1}
			\! \cdot
			c_{\beta}^n
			\!\cdot
         \bigl(
           Z_{\Nisingbetatilde} 
           \!+\! 
           Z_{\Nisingbetatildeh}
           \!+\! 
           Z_{\Nisingbetatildev} 
           \!+\! 
           Z_{\Nisingbetatildehv} 
         \bigr)\eqpunc.
           \label{eq:kw:proof:part:2}
  \end{align}

  The expression in~\eqref{eq:kw:proof:part:1} follows from \eqref{thm:ft-z}
  of Theorem~\ref{thm:nfg-dual} upon noting that $|\A| - |\V| = n$ for the
  $2$-torus lattice graph.

  The expression in~\eqref{eq:kw:proof:part:2} can be shown as
  follows. Namely,
  \begin{align*}
    &
    Z_{\FTNisingbeta} \\
	   &\quad
		\stackrel{\rm (i)}{=} \!\!
		 \sum_{x\in \ker\partial_{1}} \ 
           \prod_{e\in \A}
             \widehat{\kappa_{\beta}}(x_{e}) \\
      &\quad
		\stackrel{\rm (ii)}{=} \!\!\sum_{x\in \img\partial_{2}} \ 
           \prod_{e\in \A}
             \widehat{\kappa_{\beta}}(x_{e})
         +
	 \cdots 
         +
	 \sum_{x \in c_{\rm v}+c_{\rm h} + \img\partial_{2}} \ 
           \prod_{e\in \A}
             \widehat{\kappa_{\beta}}(x_{e}) \\
      &\quad
       = \!\!\sum_{x\in \img\partial_{2}} \ 
           \prod_{e\in \A}
             \widehat{\kappa_{\beta}}(x_{e}) 
           + 
           \cdots 
           +
           \sum_{x\in \img\partial_{2}} \ 
             \prod_{e\in \A}
               \widehat{\kappa_{\beta}}(x_{e} \! + \! c_{\rm v} \! + \! c_{\rm h}) 
                 \\
      &\quad
		\stackrel{\rm (iii)}{=} \bigl( \sqrt{2 c_{\beta}} \bigr)^{2n}
         \cdot
         \left(
         \sum_{x\in \img\partial_{2}} \ 
           \prod_{e\in \A}
             \kappa_{\widetilde{\beta}}(x_{e}) 
           + 
           \cdots
         \right) \\
      &\quad
		\stackrel{\rm (iv)}{=} %
			 2^{n} c_{\beta}^{n}
         \cdot
         \left(
			\sum_{x \in \code_{\N_{I(\beta)}}}
           \prod_{(i,j) \in \A}
             \kappa_{\widetilde{\beta}}(x_j - x_i) 
           + 
           \cdots
         \right) \\
      &\quad
		\stackrel{\rm (v)}{=} 2^{n-1} \cdot
         c_{\beta}^n
         \cdot
         \left(
			Z_{\N_{I(\tilde \beta)}}
           + 
           \cdots
         \right) \eqpunc ,
  \end{align*}
  where
  equality (i) follows from~\eqref{eq:code-im-d1-a} and \eqref{eq:ker1-img1-nfg-a},
  equality (ii) follows from~\eqref{eq:coset:2},
  equality (iii) follows from expressing $\widehat{\kappa_{\beta}}$ in terms of $\kappa_{\widetilde{\beta}}$ (see~\eqref{eq:kappa-kappa-hat}--\eqref{eq:def:dual:inverse:temperature:1}),
  equality (iv) follows from the self-duality of the $2$-torus lattice graph \eqref{eq:torus-selfdual} and \eqref{eq:code-im-d1-a},
  and equality (v) follows from \eqref{eq:z-ising-2} and the fact that the support NFG is independent of $\beta$, i.e., $\code_{\N_{I(\beta)}} = \code_{\N(\widetilde \beta)}$.
\end{proof}

Here we also give a summary of the proof using Figs.~\ref{fig:torus}(c) and
(d).  The NFG $\N_{I(\beta)}$ is in image representation since its support NFG
is equal to $\N_{\img d_1}$, as can be seen by comparing
Fig.~\ref{fig:ising-pnfg} and Fig.~\ref{fig:torus}(d) (black, solid lines).
  The NFG $\widehat{\N}_{I(\beta)}$ is in kernel representation since its support NFG is equal to $\N_{\ker \partial_1}$, as can be seen by comparing Fig.~\ref{fig:ising-dnfg} and Fig.~\ref{fig:torus}(c) (black, solid lines).
  Note that in Fig.~\ref{fig:torus}, the NFGs in image-representation form are pairwise interaction
  NFGs, while the NFGs in kernel-representation form are not pairwise interaction NFGs.
  The purpose of~\eqref{eq:kw:proof:part:1} is to invert the temperature, which is
  accomplished via the Fourier transform. As a side effect, the image representation form (pairwise
  interaction) is lost since this step causes $\N_{\img d_1}$ (Fig.~\ref{fig:torus}(d) black,
  solid lines) to be replaced by $\N_{\ker \partial_{1}}$ (Fig.~\ref{fig:torus}(c) black, solid lines). 
  The purpose of~\eqref{eq:kw:proof:part:2} is to recover the image representation form
  (pairwise interaction), which is accomplished using \eqref{eq:coset:2}, allowing one to replace
  $\N_{\ker \partial_{1}}$ (Fig.~\ref{fig:torus}(c) black, solid lines) with $\N_{\img \partial_{2}}$ (Fig.~\ref{fig:torus}(c) red, dashed lines).
  Finally, since the $2$-torus lattice graph is self-dual, $\N_{\img \partial_{2}}$ is identical to
  $\N_{\img d_{1}}$, and so the resulting interaction system is an Ising model.

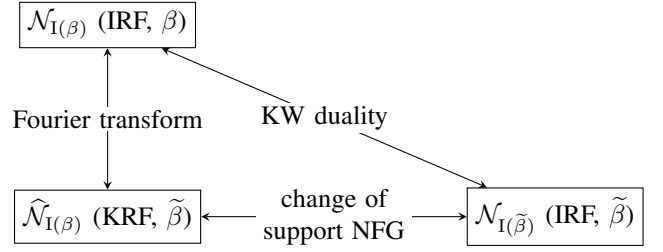
\begin{figure}[t]
  \centering
  \begin{tikzpicture}
    [scale=1.7, every node/.style={draw}]

    \node (il)[align=left] at (0, 1.5)
      {$\Nisingbeta$ (IRF, $\beta$)};

    \node (kh)[align=left] at (0, 0)
    {$\FTNisingbeta$ (KRF, $\widetilde \beta$)};

    \node (ih)[align=left] at (3.5, 0)
      {$\Nisingbetatilde$ (IRF, $\widetilde{\beta}$)};

    \draw[>=stealth] (il)[<->] --node[inner sep=1, fill=white, draw=none]
                                     {KW duality} (ih);

    \draw[>=stealth] (il)[<->] --node[inner sep=1, fill=white, draw=none]
                                     {Fourier transform} (kh);

    \draw[>=stealth] (ih)[<->] --node[inner sep=1, fill=white, draw=none]
                                     {\parbox{2cm}{\centering 
                                                     change of \\ 
                                                     support NFG}
                                     } (kh);
  \end{tikzpicture}
  \caption{Key objects and key steps of the Kramers--Wannier duality. The
    content of every box is ``NFG (form, inverse temperature)'', where ``NFG''
    refers to the relevant NFG, where ``form'' is either ``IRF'' (shorthand
    for ``image-representation form'') or ``KRF'' (shorthand for
    ``kernel-representation form''), and where ``inverse temperature'' refers
    to the inverse temperature appearing in the interaction functions.}
  \label{fig:kw-duality:simplified:1}
\end{figure}

The key objects and key steps of the above comments are summarized in
Fig.~\ref{fig:kw-duality:simplified:1}.

\begin{proof}
  (Proof of Theorem~\ref{theorem:kw:duality:limit:n:1}.) Because
  $\Nisingbetatilde$ and $\Nisingbetatildeh$ differ only in $L$ ($= \sqrt{n}$)
  interaction functions, we have for any configuration that the global
  function value for $\Nisingbetatildeh$ is lower bounded by the global
  function value for $\Nisingbetatilde$ times $\operatorname{e}^{-2\sqrt{n}
    \widetilde{\beta}}$ and upper bounded by the global function value for
  $\Nisingbetatilde$ times $\operatorname{e}^{+2\sqrt{n}
    \widetilde{\beta}}$. Summing over all configurations, we obtain
  \begin{align*}
    Z_{\Nisingbetatilde} 
      \cdot 
      \operatorname{e}^{-2 \sqrt{n} \widetilde{\beta}}
      &\leq
         Z_{\Nisingbetatildeh}
       \leq
         Z_{\Nisingbetatilde} 
           \cdot 
           \operatorname{e}^{+2\sqrt{n} \widetilde{\beta}} \eqpunc ,
  \end{align*}
  i.e.,
  \begin{align*}
    \frac{1}{n}
      \log\bigl( Z_{\Nisingbetatilde} \bigr)
      \! - \!
      \frac{2 \widetilde{\beta}}{\sqrt{n}}
      &\leq
         \frac{1}{n}
           \log\bigl( Z_{\Nisingbetatildeh} \bigr)
       \leq
         \frac{1}{n}
         \log\bigl( Z_{\Nisingbetatilde} \bigr)
         \! + \!
         \frac{2 \widetilde{\beta}}{\sqrt{n}} \eqpunc .
  \end{align*}
  Similarly,
  \begin{align*}
    \frac{1}{n}
      \log\bigl( Z_{\Nisingbetatilde} \bigr)
      \! - \!
      \frac{2 \widetilde{\beta}}{\sqrt{n}}
      &\leq
         \frac{1}{n}
           \log\bigl( Z_{\Nisingbetatildev} \bigr)
       \leq
         \frac{1}{n}
         \log\bigl( Z_{\Nisingbetatilde} \bigr)
         \! + \!
         \frac{2 \widetilde{\beta}}{\sqrt{n}} \eqpunc , \\
    \frac{1}{n}
      \log\bigl( Z_{\Nisingbetatilde} \bigr)
      \! - \!
      \frac{4 \widetilde{\beta}}{\sqrt{n}}
      &\leq
         \frac{1}{n}
           \log\bigl( Z_{\Nisingbetatildehv} \bigr)
       \leq
         \frac{1}{n}
         \log\bigl( Z_{\Nisingbetatilde} \bigr)
         \! + \!
         \frac{4 \widetilde{\beta}}{\sqrt{n}} \eqpunc .
  \end{align*}
  For finite $\widetilde \beta$, the desired result then follows by taking the limit $n \to \infty$ on both
  sides of~\eqref{eq:theorem:kw:duality:1}, along with using the above
  inequalities to simplify the expression.
\end{proof}

\section{Extensions}
\label{sec:ext}

In this section, we discuss the Kramers--Wannier duality for some statistical
models beyond the two-dimensional Ising model.  There are two (independent)
directions of extending the Ising model. Namely, on the one hand, one may look
at higher-dimensional lattices, and, on the other hand, one may consider a
different form of interaction functions between adjacent sites. We discuss the
first direction in Section~\ref{sec:3d:ising:model:1}, which is about
three-dimensional Ising models, and discuss the second direction in
Section~\ref{sec:potts:model:1}, which is about the Potts model \cite{potts1952}.

\subsection{Three-Dimensional Ising Models}
\label{sec:3d:ising:model:1}

In this section we consider different three-dimensional Ising type models and
derive Kramers--Wannier type results. Note that we will first consider the
cubic lattice before moving on to the $3$-torus lattice graph.

A $3$-complex, in addition to vertices, edges, and surfaces, also includes
three-dimensional objects (polyhedra). We will denote the underlying graph by
$\G \defeq (\V,\A,\setS,\setP)$, where $\setP$ collects all polyhedra. The
definition of the boundary of a vertex, edge, or a surface is similar to the
$2$-complex case, see Section~\ref{sec:two:complex:primal:domain:1}. The
boundary of a three-dimensional object $p \in \setP$ is defined as the
weighted sum of the surfaces surrounding $p$, where the weights are chosen
from the set $\left\{ -1, +1 \right\}$ depending on the relative orientations
of $p$ and its surrounding surfaces.

\begin{example}%
  Consider the graph $\G \defeq (\V,\A,\setS,\setP)$ in Fig.~\ref{fig:cube}(a)
  for which $\setP = \{ p_0 \}$, i.e., it consists of a single polyhedron. (One
  can ignore the light blue dashed lines in the figure for the moment.)
  We choose the orientation (not shown in the figure) of the solid cube $p_0$
  so that the front, left, and bottom faces have weight $+1$, and so that the
  back, right, and top faces have weight $-1$ in $\partial_{3}(p_0)$. One may
  verify that this results in the following $3$-complex:
  \begin{align*}
    \begin{array}{cc cc cc cc}
             &C_{3} & \xrightarrow{\partial_{3}} 
             & C_{2} & \xrightarrow{\partial_{2}} 
             & C_{1} & \xrightarrow{\partial_{1}}
             & C_{0} \cr
       \dim C_i	& 1    &  & 6 &  & 12 &  & 8 \\
      \dim H_i	& 0    &  & 0 &  & 0  &  & \eqpunc\eqpunc 1
    \end{array}
  \end{align*}
  Some NFGs that are associated with this complex are shown in
  Figs.~\ref{fig:cube}(b)--(g). (The choice of which NFGs are shown in the
  figure is based on which NFGs will be required for the subsequent
  discussions.)

  The cube in Fig.~\ref{fig:cube}(a) may be extended as part of a cubic lattice
  on a larger number of vertices in the obvious way, e.g., as done in
  Fig.~\ref{fig:cube-more}(a). (The dashed lines in Fig.~\ref{fig:cube}(a)
  show how the cube can be connected as part of a larger lattice.) In this
  case, the NFGs in Fig.~\ref{fig:cube} also extend in the obvious way, for
  instance, the NFGs $\N_{\img \partial_3}$ and $\N_{\img d_{1}}$ are as in
  Figs.~\ref{fig:cube-more}(b) and (c), respectively.
  \qede
\end{example}

\begin{figure}[t]
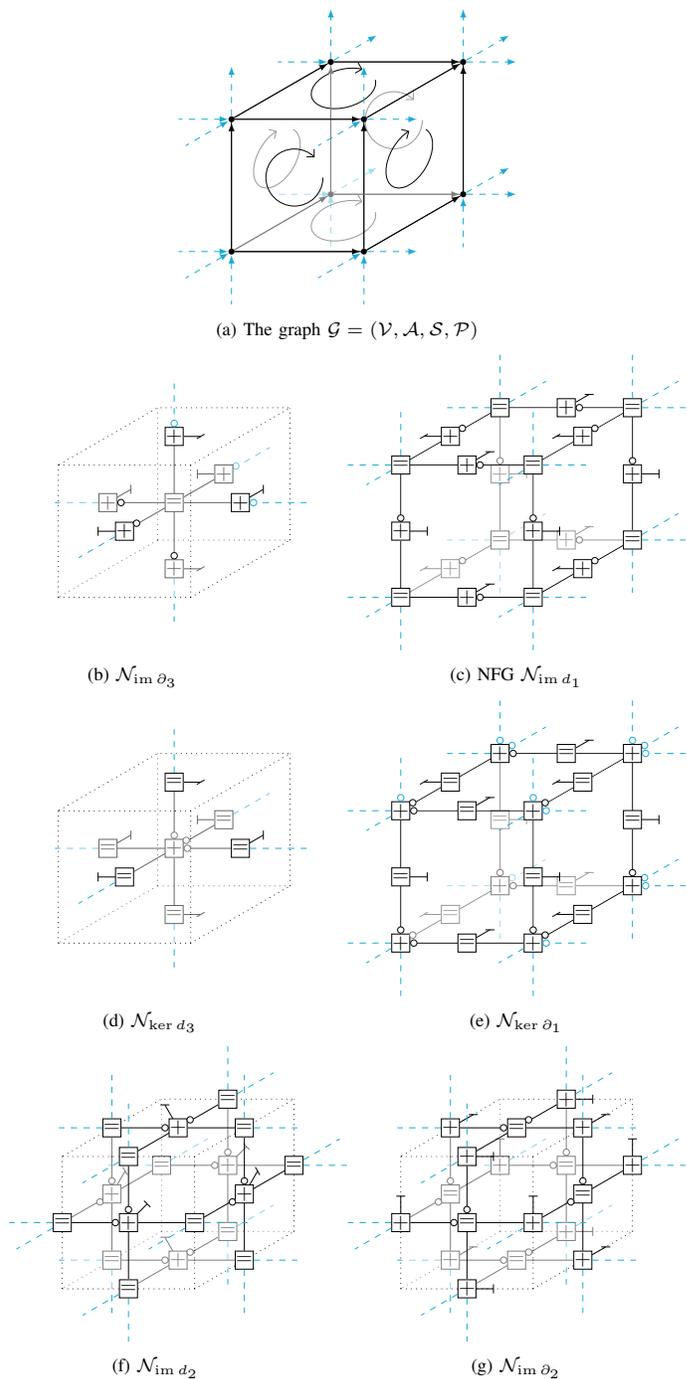

	\def\s{.8}	
	\def\dist{2.2}
	\ifonecol
	\def\xshift{6.5}
	\def\yyshift{-.75}
	\else
	\def\xshift{4.5}
	\def\yyshift{0}
	\fi
	\def\yshift{-4.6}
	\centering
	\begin{tikzpicture}
		\node(a)[scale=\s] at (\xshift/2,\yyshift)
                {\input{figtex/torus/cube-graph.tex}};
		\node[vcaption, below = of a] 
                {(a) The graph $\G = (\V,\A,\setS,\setP)$};

		\node(b)[scale=\s] at (1*\xshift,1*\yshift)
                {\input{figtex/torus/cube-imdif1.tex}};
		\node(b1)[vcaption, below = of b] 
                       {(c) NFG $\N_{\img d_1}$};

		\node(b)[scale=\s] at (1*\xshift,2*\yshift)
                {\input{figtex/torus/cube-kerbd1.tex}};
		\node(b2)[vcaption, below = of b] 
                       {(e) $\N_{\ker \partial_1}$};

		\node(b)[scale=\s] at (1*\xshift,3*\yshift)
                {\input{figtex/torus/cube-imbd2.tex}};
		\node(b3)[vcaption, below = of b] 
                       {(g) $\N_{\img \partial_{2}}$};

		\node(c)[scale=\s] at (0*\xshift,1*\yshift)
                {\input{figtex/torus/cube-imbd3.tex}};
		\node[vcaption, node distance=\xshift*.75, left = of b1] 
                   {(b) $\N_{\img \partial_{3}}$};

		\node(c)[scale=\s] at (0*\xshift,2*\yshift)
                {\input{figtex/torus/cube-kerdif3.tex}};
		\node[vcaption, node distance=\xshift*.75, left = of b2] 
                   {(d) $\N_{\ker d_{3}}$};

		\node(c)[scale=\s] at (0*\xshift,3*\yshift)
                {\input{figtex/torus/cube-imdif2.tex}};
		\node[vcaption, node distance=\xshift*.75, left = of b3] 
                   {(f) $\N_{\img d_{2}}$};

	\end{tikzpicture}
	\caption{A graph $\G = (\V,\A,\setS,\setP)$ representing a cubic
        lattice and associated NFGs. (The dashed lines show how the figures
        extend for a larger lattice.)}
      \label{fig:cube}
\end{figure}

\begin{figure}[t]
	\def\s{.8}	
	\def\dist{2.2}
	\def\xshift{5}
	\def\yshift{-6}
	\centering
	\begin{tikzpicture}
		\node(a)[scale=\s] at (0,0){\input{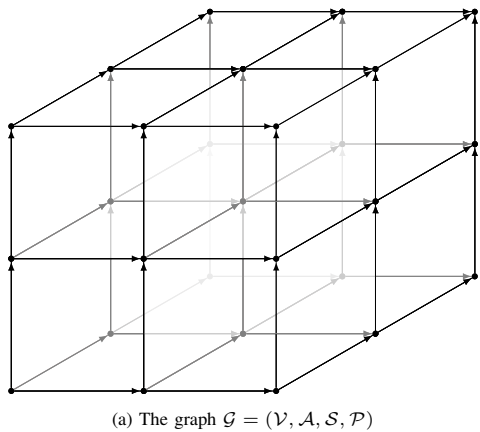}};
		\node[vcaption, below = of a] {(a) The graph $\G = (\V,\A,\mathcal{S},\mathcal{P})$};
		
		\node(imdif1)[scale=\s] at (0*\xshift,2*\yshift){\input{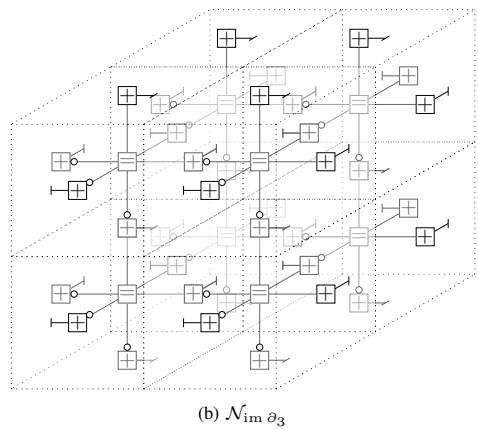}};
		\node[vcaption, below = of imdif1] {(c) $\N_{\img d_{1}}$};

		\node(imbd3)[scale=\s] at (0*\xshift,1*\yshift){\input{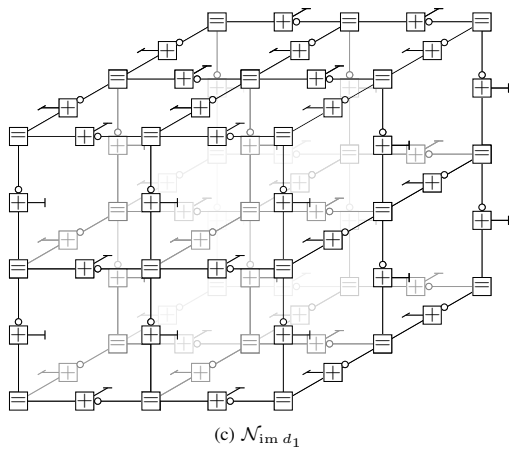}};
		\node[vcaption, below = of imbd3] {(b) $\N_{\img \partial_{3}}$};

	\end{tikzpicture}
	\caption{A demonstration of some of the drawings in
          Fig.~\ref{fig:cube} for a larger cubic lattice.}
	\label{fig:cube-more}
\end{figure}

The Ising model on the three-dimensional cubic lattice is defined in the same
way as the two-dimensional Ising model, see the NFG $\N_{I(\beta)}$ in
Fig.~\ref{fig:cube-ising-pnfg} (for now however without periodic
boundary). Namely, the spins are located on the vertices of the cubic lattice,
and we have the partition function $Z$ as in (\ref{eq:z-ising-1}), where now
$\A$ is the set of edges of the cubic lattice.  The support NFG of $\N_{I(\beta)}$ is the
NFG in Fig.~\ref{fig:cube-more}(c). In subsequent discussions, we may simply
refer to the corresponding NFGs in Fig.~\ref{fig:cube}, where it is understood
that the NFGs are extended to the appropriate number of vertices in the lattice.

\begin{figure}[t]
  \def\s{1}	
  \def\c{.56}	
  \def\L{2}	
  \def\dist{2cm}
  \centering
  \def\yshift{-7cm}
  \def\xshift{9cm}
  \begin{tikzpicture}[scale=\s]
  	\node(b) at (0*\xshift,0){\input{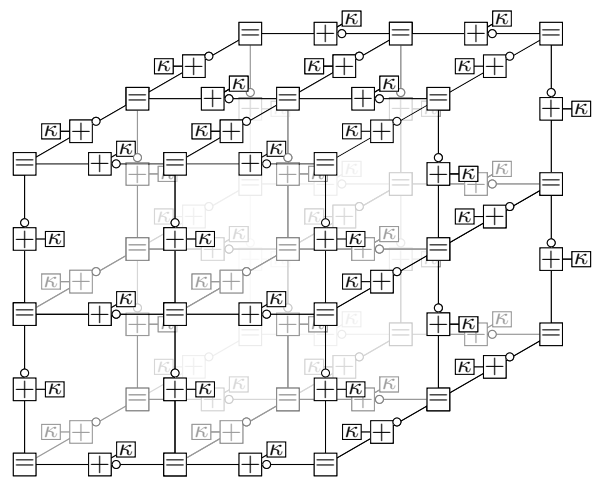}};
  \end{tikzpicture}
  \caption{The NFG $\Nisingbeta$ representing the Ising model on the
    cubic lattice.}
  \label{fig:cube-ising-pnfg}
\end{figure}

\begin{figure}[t]
  \def\s{1}	
  \def\c{.56}	
  \def\L{2}	
  \def\dist{2cm}
  \centering
  \def\yshift{-7cm}
  \def\xshift{9cm}
  \begin{tikzpicture}[scale=\s]
  	\node(b) at (0*\xshift,0){\input{figtex/torus/cube-dnfg-kappa}};
  \end{tikzpicture}
  \caption{The Fourier-transformed NFG $\FTNisingbeta$.}
  \label{fig:cube-ising-dnfg}
\end{figure}

In the three-dimensional Ising model, a spin (away from the borders of the
lattice) has six neighbors, as can be seen in the NFG $\Nisingbeta$ in
Fig.~\ref{fig:cube-ising-pnfg}. Now we can obtain a Kramers--Wannier type
duality for the three-dimensional model as follows. (For reasons of
simplicity, we omit the discussion of scaling factors.)
\begin{itemize}

\item We take the Fourier transform of the NFG $\Nisingbeta$ in
  Fig.~\ref{fig:cube-ising-pnfg} (see Fig.~\ref{fig:cube}(c) for the support
  NFG), which results in the NFG $\FTNisingbeta$ shown in
  Fig.~\ref{fig:cube-ising-dnfg} (see Fig.~\ref{fig:cube}(e) for the support
  NFG). As in Section~\ref{sec:kw:duality:ising:model:1}, we can relate the
  partition function of $\Nisingbeta$ to the partition function of $\FTNisingbeta$.

\item As in Section~\ref{sec:kw:duality:ising:model:1}, the interaction
  functions of $\FTNisingbeta$ are in the desired form, but $\FTNisingbeta$ is
  in kernel-representation form because the support NFG of $\FTNisingbeta$ is
  equivalent to $\N_{\ker \partial_1}$. In order to obtain an NFG in
  image-representation form, we replace the support NFG $\N_{\ker \partial_1}$
  (see Fig.~\ref{fig:cube}(e) by $\N_{\img \partial_2}$ (see
  Fig.~\ref{fig:cube}(g)). Note that for the considered graph $\G$ we have
  $\dim H_1 = 0$, i.e., $\N_{\ker \partial_1}$ is equivalent to
  $\N_{\img \partial_2}$.

\end{itemize}

In summary, we have been able to relate the partition functions of the following
NFGs:
\begin{enumerate}

\item an NFG in image-representation form that represents some statistical
  model at inverse temperature $\beta$;

\item an NFG in image-representation form that represents some other
  statistical model at the dual inverse temperature $\widetilde{\beta}$.

\end{enumerate}

There are two key differences, however, to the analogous result in
Section~\ref{sec:kw:duality:ising:model:1}.
\begin{itemize}
	\item In Section~\ref{sec:kw:duality:ising:model:1}
		the connectivity pattern of the NFG corresponding to 
		Item~2 was, due to the self-duality of the $2$-torus lattice, the same
		as the connectivity pattern of the NFG corresponding to Item~1. This is not the case here
		anymore, i.e., the connectivity pattern of the NFG corresponding to Item~2 is different from
		the connectivity pattern of the NFG corresponding to Item~1. In particular, in this section, whereas for the
		statistical model in Item~1 a spin (away from the borders of the lattice)
		interacts with $6$ other spins, for the statistical model in Item~2 a spin
		(away from the borders of the lattice) interacts with $12$ other spins.
	\item In Section~\ref{sec:kw:duality:ising:model:1} the NFG corresponding to Item~2, since the lattice is
		two-dimensional, is a pairwise interaction NFG. (Like the NFG corresponding to Item~1.)
		This is not the case here anymore, i.e., the NFG corresponding to Item~2 is not 
		a pairwise interaction NFG. (Unlike the corresponding NFG in Item~1) In
		particular, in this section, whereas for the statistical model in Item~1 an interaction function (away from the
		borders of the lattice) is connected to a parity indicator function that involves two spins,
		for the statistical model in Item~2, an interaction function is connected to a parity
		indicator function that involves four spins.

\end{itemize}

Finally, let us point out a different approach to obtain a Kramers--Wannier
type duality result for the three-dimensional Ising model. (In fact, this is
the route taken by Savit in \cite{savit:duality}.) Namely, instead of placing
spins at vertices and letting them interact via edges, we can place spins at
the center of solid cubes and let them interact via their common face. A
Kramers--Wannier type duality result is then obtained by considering the NFGs
in Figs.~\ref{fig:cube}(b)(d)(f) instead of Figs.~\ref{fig:cube}(c)(e)(g). We
omit the details.

Next, we proceed to the Ising model on the $3$-torus lattice graph, where we
start with the following example dedicated to the $3$-complex obtained from
the $3$-torus.

\begin{example}%
  The chain complex of the $3$-torus lattice graph can be described using
  Fig.~\ref{fig:3torus}(a), obtained from Fig.~\ref{fig:cube}(a) (ignoring the
  dashed edges) by identifying the left and right most faces, the top and 
  bottom most faces, and the front and back most faces. Such an identification
  of the faces results in an identification of some corresponding edges and
  some corresponding vertices.  Namely, in this case, all the vertices in
  Fig.~\ref{fig:cube}(a) are identified as one vertex, all the horizontal
  edges as one edge, all the vertical edges as one edge, and all the
  perpendicular (to the page) edges as one edge, as shown in
  Fig.~\ref{fig:3torus}(a). Equivalently, this torus can be drawn as in Fig.~\ref{fig:3torus}(b). (By repeating vertices, edges, and
  faces.) Finally, the boundary operator is
  defined as in the previous example, and one may verify that one obtains the
  following $3$-complex:
  \begin{align*}
    \begin{array}{cc cc cc cc}
        & C_{3} & \xrightarrow{\partial_{3}} 
        & C_{2} & \xrightarrow{\partial_{2}} 
        & C_{1} & \xrightarrow{\partial_{1}}
        & C_{0} \\
      \dim C_i & 1 &  & 3  &  & 3  &  & 1 \\
      \dim H_i & 1 &  & 3  &  & 3  &  & 1
    \end{array}
  \end{align*}
  Some NFGs that are associated with this complex are shown in
  Figs.~\ref{fig:3torus}(c)--(h).

  Let $n \defeq L^{3}$, where $L$ is some positive integer. The cubic lattice
  with $L \times L \times L$ vertices on the $3$-torus can be
  obtained via applying the identification above on the cubic lattice on
  $(L+1)^{3}$ vertices. As we have seen throughout this work, such a lattice
  will have the same dimensions of the homology spaces $H_i$ as above. In
  other words, one may verify that the chain complex arising from the cubic
  lattice on the $3$-torus may be summarized as
  \begin{align*}
    \begin{array}{cc cc cc cl}
        & C_{3} & \xrightarrow{\partial_{3}} 
        & C_{2} & \xrightarrow{\partial_{2}} 
        & C_{1} & \xrightarrow{\partial_{1}}
        & C_{0} \\
      \dim C_i	& n    &  & 3n &  & 3n &  &  n \\
      \dim H_i	& 1    &  & 3  &  & 3  &  &  1
    \end{array}
  \end{align*}
  \qede

  The corresponding NFGs to the ones in Fig.~\ref{fig:3torus}(c)--(h) for any $n$ can be obtained in the
  obvious way, e.g., Figs.~\ref{fig:3torus-more}(a) and (b), respectively, show $\N_{\img \partial_3}$ and
  $\N_{\img d_{1}}$ for $n=8$.
\end{example}

\begin{figure}[t]
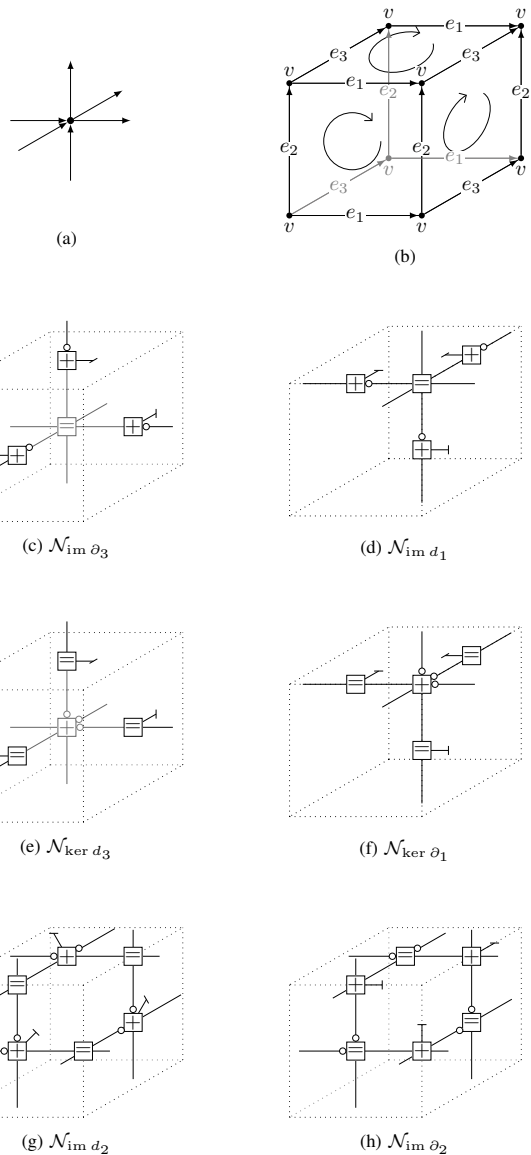

  \def\s{.8}	
  \def\c{.56}	
  \def\L{1}	
  \def\dist{2.2}
  \centering
  \ifonecol
  \def\xshift{6.5}
  \else
  \def\xshift{4.5}
  \fi
  \def\yshift{-4.0}
  \begin{tikzpicture}[scale=1]
  	\node(a)[scale=\s] at (0*\xshift,0)
          {\input{figtex/torus/3torus-graph-n1}};
  	\node[vcaption, node distance=4.5mm, below=of a] 
             {(a)};

  	\node(a)[scale=\s] at (1*\xshift,0)
          {\input{figtex/torus/3torus-graph-n1-2}};
  	\node[vcaption, below=of a] 
             {(b)};

  	\node(d)[scale=\s] at (1*\xshift,1*\yshift)
          {\input{figtex/torus/3torus-imdif1-n1.tex}};
  	\node[vcaption, node distance=1mm,  below = of d] 
             {(d) $\N_{\img d_1}$};

  	\node(f)[scale=\s] at (1*\xshift,2*\yshift)
          {\input{figtex/torus/3torus-kerbd1-n1.tex}};
  	\node[vcaption, node distance=1mm,  below = of f] 
             {(f) $\N_{\ker \partial_1}$};

  	\node(h)[scale=\s] at (1*\xshift,3*\yshift)
          {\input{figtex/torus/3torus-imbd2-n1.tex}};
  	\node[vcaption, below = of h] 
             {(h) $\N_{\img \partial_2}$};

  	\node(c)[scale=\s] at (0*\xshift,1*\yshift)
          {\input{figtex/torus/3torus-imbd3-n1.tex}};
  	\node[vcaption, below = of c]
             {(c) $\N_{\img \partial_{3}}$};

  	\node(e)[scale=\s] at (0*\xshift,2*\yshift)
          {\input{figtex/torus/3torus-kerdif3-n1.tex}};
  	\node[vcaption, below = of e]
          {(e) $\N_{\ker d_{3}}$};

  	\node(f)[scale=\s] at (0*\xshift,3*\yshift)
          {\input{figtex/torus/3torus-imdif2-n1.tex}};
  	\node[vcaption, below = of f]
          {(g) $\N_{\img d_{2}}$};
  \end{tikzpicture}
  \caption{A graph $\G = (\V,\A,\setS,\setP)$ representing a $3$-torus lattice graph and associated NFGs.
	  In (a) and (c)--(h) the left and right most edges, the top and bottom most edges, and the front and
	  back most edges are identified.
  }
  \label{fig:3torus}
\end{figure}

\begin{figure}[t]
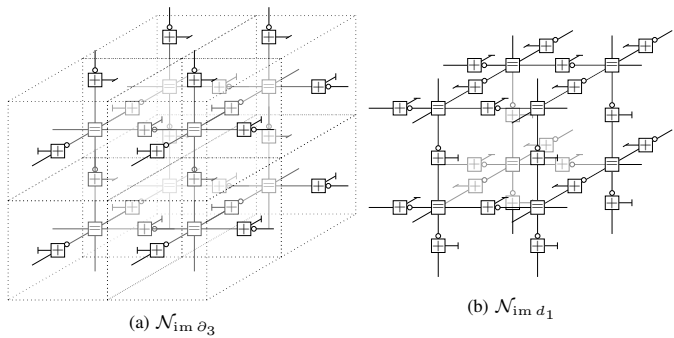

	\def\s{.6}	
	\def\c{.56}	
	\def\L{1}	
	\def\dist{2.2}
	\centering
	\ifonecol
	\def\xshift{6.5}
	\else
	\def\xshift{4.5}
	\fi
	\def\yshift{-4.0}
	\begin{tikzpicture}[scale=1]
		\node(a)[scale=\s] at (0*\xshift,1*\yshift)
                  {\input{figtex/torus/3torus-imbd3.tex}};
						\node[vcaption, below = of a] {(a) $\N_{\img \partial_{3}}$};

		\node(b)[scale=\s] at (1*\xshift,1*\yshift)
                  {\input{figtex/torus/3torus-imdif1.tex}};
						\node[vcaption, node distance=1mm,  below = of b] {(b) $\N_{\img d_1}$};

	\end{tikzpicture}
	\caption{A demonstration of some of the NFGs in Fig.~\ref{fig:3torus}
          for the $3$-torus with $n=8$ vertices, where the edges follow the same identification as in Fig.~\ref{fig:3torus}}
	\label{fig:3torus-more}
\end{figure}

\begin{figure}[t]
	\def\s{1}	
	\def\c{.56}	
	\def\L{2}	
	\def\dist{2.5}
	\centering
	\def\yshift{-7cm}
	\def\xshift{5cm}
	\begin{tikzpicture}[scale=\s]
		\node(b) at (0*\xshift,0*\yshift)
                  {\input{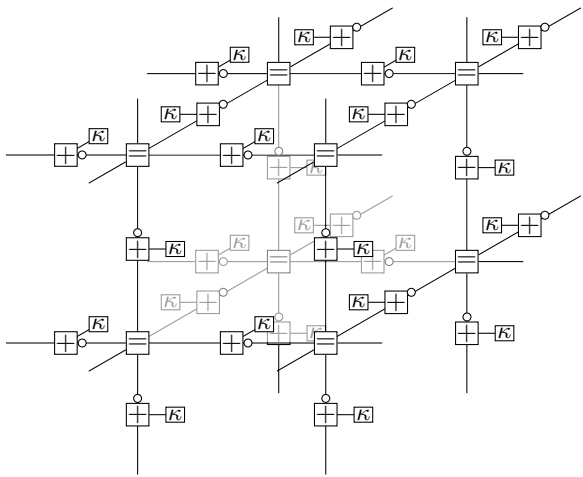}};
	\end{tikzpicture}
	\caption{The NFG $\G$ representing the Ising model on the
          $3$-torus with $n = 8$ sites, where the edges follow the same identification as in Fig.~\ref{fig:3torus}.}
	\label{fig:3d-ising-pnfg}
\end{figure}

The Ising model on the $3$-torus is shown in
Fig.~\ref{fig:3d-ising-pnfg}. The argument for a Kramers--Wannier type duality
result is similar to the cubic lattice. However, while it is still true that
$\img \partial_{2} \subseteq \ker \partial_{1}$, the reverse inclusion does
not hold since $\dim H_{1} = 3$, and we must take care of the cosets of
$\img \partial_{2}$ in $\ker \partial_{1}$. The approach is similar to the
approach in Section~\ref{sec:kw:duality:ising:model:1}, except that here we
have $2^3 = 8$ cosets instead of $2^2 = 4$ cosets. We omit the details.

\subsection{The Potts Model}
\label{sec:potts:model:1}

In the Ising model on the $2$-torus lattice, the spins are
binary and the interaction between a pair of adjacent sites may take one of
two values depending on whether the two spins are equal or not, i.e., the
interaction depends on the Hamming distance between the pair of spins. An
obvious extension is to allow spins to take values from a larger alphabet,
which we will assume to be the ring $\Z_{q}$ of integers modulo $q$,
where $q \geq 2$ is an integer. While we still insist that the interaction may
only depend on the difference between a pair of adjacent spins, there is
certainly more freedom (compared to the Ising model) in choosing such a
dependency, and several well-studied statistical models may be obtained by
further specifying such a dependency.  For instance, insisting on an
interaction that only depends on the Hamming distance results in the standard
Potts model, and specifying the interaction as one that depends on the Lee
distance gives the vector Potts model.
For a detailed discussion of the Potts model \cite{potts1952}, see, e.g., \cite{wu:Potts}.

Note that most techniques from the previous sections can be suitably extended
from a finite field $\F$ to the ring $\Z_q$:
\begin{itemize}

\item The Fourier transform results still hold because they only rely on
  $\langle \Z_q, + \rangle$ being a finite Abelian group.

\item The results involving $m$-complexes need to be suitably generalized
  because vector spaces are replaced by modules.

\end{itemize}

It turns out that for the standard Potts model one can derive Kramers--Wannier
type duality results. (We omit the rather straightforward details.) However,
for the vector Potts model this does not seem to be the case. The issue here
is that $\widehat{\kappa_{\beta}}$ cannot be expressed as some function
$\kappa'_{\widetilde{\beta}}$ such that $\widetilde{\beta}^{-1} \cdot
\log(\kappa'_{\widetilde{\beta}})$ is independent of $\widetilde{\beta}$, as
would be required for the interaction function of a statistical
model. (Exceptions are the special case $q = 3$, where the Lee distance
reduces to the Hamming distance, and the special case $q = 4$, where one can
show that the interaction decomposes into the Kronecker product of two
interactions of binary spins.)

\section{Conclusion}
\label{sec:conclusion:1}

Over the last twenty years, factor graphs have proven to be a powerful
concept. In contrast to many other diagrammatic representations that serve
only a particular limited purpose, factor graphs have turned out to be very
versatile: they can be applied in various contexts, they highlight how
variables are related, they help us to efficiently do exact and approximate
calculations, and they can be transformed. With the present paper, we have
added another tool to this growing toolbox by showing how objects from
algebraic topology can be expressed in terms of factor graphs and by
discussing how algebraic topology can guide the transformation of factor
graphs. We have then used these tools to re-prove the Kramers--Wannier
duality. These tools have also proven to be beneficial in other
contexts~\cite{Li:Vontobel:16:1}.

\section*{Acknowledgment}

The authors gratefully acknowledge discussions with David Forney, Andi
Loeliger, Yongyi Mao, and Mehdi Molkaraie on topics related to this paper. In
particular, the authors appreciate a comment in a paper by Molkaraie
and Loeliger~\cite{mehdi:ising} that started the investigations that led to
the results discussed in the present paper 
and would also like to thank the reviewers and the associate editor for their valuable
comments. 

\appendix
\subsection{Higher-Order Complexes}
\label{sec:cw-complex}

In Sections~\ref{sec:one:complex}--\ref{sec:two:complex} we have introduced
$1$-com\-plexes and $2$-complexes, respectively. The generalization to
$m$-complexes for $m \geq 3$ is fairly straightforward and so
this appendix is rather brief.

Let $m$ be some positive integer. An $m$-dimensional chain complex, or simply
$m$-complex, consists of the following objects:
\begin{itemize}

\item finite-dimensional spaces $C_i$ over $\F$, $i \in \{ -1, 0, 1, \ldots,$
  $m \! + \! 1 \}$;

\item boundary operators $\partial_i: C_i \to C_{i-1}$, $i \in \{ 0, 1,
  \ldots, m \! + \! 1 \}$, where these operators satisfy
  \begin{align*}
    \img \partial_{i+1} 
      &\subseteq
         \ker \partial_{i},
           \quad i \in \{ 0, 1, \ldots, m \} \eqpunc ,
  \end{align*}
  i.e.,
  \begin{align*}
    \partial_{i} \circ \partial_{i+1} 
      &= 0,
           \quad i \in \{ 0, 1, \ldots, m \} \eqpunc .
  \end{align*}

\end{itemize}
These objects are summarized in~\eqref{eq:m:complex:chain}. For $i \in \{ 0,
1, \ldots, m \}$, the $i$-th homology space is defined to be
\begin{align*}
  H_i
    &\defeq 
       \ker\partial_{i} \ / \ \img\partial_{i+1} \eqpunc .
\end{align*}

On the other hand, an $m$-dimensional cochain complex consists of the
following objects:
\begin{itemize}

\item spaces $\widehat{C}_i$, $i \in \{ -1, 0, 1, \ldots, m \! + \! 1 \}$,
  where $\widehat{C}_i$ is the dual space of $C_i$;

\item coboundary operators $d_i: \widehat{C}_{i-1} \to \widehat{C}_{i}$, $i
  \in \{ 0, 1, \ldots, m \! + \! 1 \}$, where these operators satisfy
  \begin{align*}
    \img d_{i}
      &\subseteq
         \ker d_{i+1},
           \quad i \in \{ 0, 1, \ldots, m \} \eqpunc .
  \end{align*}
  i.e.,
  \begin{align*}
    d_{i+1} \circ d_i 
      &= 0,
           \quad i \in \{ 0, 1, \ldots, m \} \eqpunc .
  \end{align*}

\end{itemize}
These objects are summarized in~\eqref{eq:m:complex:cochain}. For $i \in \{ 0,
1, \ldots, m \}$, the $i$-th cohomology space is defined to be
\begin{align*}
  \widehat{H}_i
    &\defeq 
       \ker d_{i+1} \ / \ \img d_{i} \eqpunc .
\end{align*}

The following list contains some further important notions:
\begin{itemize}

\item elements of $C_i$ are called $i$-chains;

\item elements of $\widehat{C}_i$ are called $i$-cochains;

\item elements of $\img \partial_{i+1}$ are called $i$-boundaries;

\item elements of $\ker \partial_i$ are called $i$-cycles;

\item elements of $\img d_i$ are called $i$-coboundaries;

\item elements of $\ker d_{i+1}$ are called $i$-cocycles.

\end{itemize}

\begin{figure}[t]
  \begin{align}
    \label{eq:m:complex:chain}
    &\textcolor{gray}{C_{m+1}} \ 
     \textcolor{gray}{\xrightarrow{\!\!\partial_{m+1}\!\!}} \ 
     C_{m} 
     \xrightarrow{\partial_{m}}
     \cdots 
     \xrightarrow{\partial_{3}} 
     C_{2} 
     \xrightarrow{\partial_{2}} 
     C_{1} 
     \xrightarrow{\partial_{1}}
     C_{0} \ 
     \textcolor{gray}{ \xrightarrow{\partial_{0}}} \ 
     \textcolor{gray}{C_{-1}} \\
     \label{eq:m:complex:cochain}
    &\textcolor{gray}{\widehat{C}_{m+1}} \ 
     \textcolor{gray}{\xleftarrow{\!\!d_{m+1}\!\!}} \
     \widehat{C}_{m} 
     \xleftarrow{d_{m}}
     \cdots 
     \xleftarrow{d_{3}}
     \widehat{C}_{2} 
     \xleftarrow{d_{2}} 
     \widehat{C}_{1} 
     \xleftarrow{d_{1}}
     \widehat{C}_{0} \
     \textcolor{gray}{\xleftarrow{d_{0}}} \ 
     \textcolor{gray}{\widehat{C}_{-1}}
   \end{align}
   \caption{Spaces and mappings associated with an $m$-complex.}
   \label{fig:chain-cochain}
\end{figure}

Finally, let us mention the useful fact that
\begin{align}
  \dim \widehat{H}_{i} 
    &= \dim H_{i},
         \quad i \in \{ 0, 1, \ldots, m \}.
  \label{eq:dim-cohomo}
\end{align}
This is a direct consequence of
\begin{align}
  \label{eq:diffi-ker}
  \ker d_{i} 
    &= \left( 
         \img \partial_{i} 
       \right)^{\perp} \eqpunc , \\ 
  \label{eq:diffi-img}
  \img d_{i} 
    &= \left( 
         \ker\partial_{i} 
       \right)^{\perp} \eqpunc ,
\end{align}
where for any subspace $U \subseteq {C}_{i}$ we have defined its orthogonal space
$U^{\perp}$ to be
\begin{align*}
  U^{\perp}
    &\defeq 
       \bigl\{
         \varphi \in \widehat{C}_{i} 
       \bigm|
         \varphi(x)=0 \ \text{for all $x \in U$}
       \bigr\} \eqpunc .
\end{align*}
Equality \eqref{eq:diffi-ker} follows by noting that $\img \partial_{i} \subseteq C_{i-1}$, and so
\begin{align*}
	(\img \partial_i)^{\perp}
	&= \left\{ \varphi \in \widehat{C}_{i-1} \mid \varphi(\x) = 0 \ \  \forall \x \in \img \partial_{i} \right\}
	\\
	&= \left\{ \varphi \in \widehat{C}_{i-1} \mid \varphi(\partial_{i} \y) = 0 \ \  \forall \y \in C_{i} \right\}
	\\
	&= \left\{ \varphi \in \widehat{C}_{i-1} \mid (d_{i} \varphi)(\y) = 0 \ \  \forall \y \in C_{i} \right\}
	\\
	&= \ker d_{i}.
\end{align*}
(A similar argument leads to \eqref{eq:diffi-img}.)
Now \eqref{eq:dim-cohomo} follows easily from
\begin{align*}
	\dim\widehat H_{i}
	&= \dim(\ker d_{i+1}) - \dim(\img d_{i}) \\
	&= \bigl(
             \dim C_{i} \! - \! \dim(\img \partial_{i+1})
           \bigr) 
           - 
           \bigl(
             \dim C_i \! - \! \dim(\ker \partial_{i})
           \bigr) \\
	&= \dim H_{i}.
\end{align*}

\bibliographystyle{IEEEtran}
\bibliography{bibliolist} %

\begin{thebibliography}{10}
\providecommand{\url}[1]{#1}
\csname url@samestyle\endcsname
\providecommand{\newblock}{\relax}
\providecommand{\bibinfo}[2]{#2}
\providecommand{\BIBentrySTDinterwordspacing}{\spaceskip=0pt\relax}
\providecommand{\BIBentryALTinterwordstretchfactor}{4}
\providecommand{\BIBentryALTinterwordspacing}{\spaceskip=\fontdimen2\font plus
\BIBentryALTinterwordstretchfactor\fontdimen3\font minus
  \fontdimen4\font\relax}
\providecommand{\BIBforeignlanguage}[2]{{%
\expandafter\ifx\csname l@#1\endcsname\relax
\typeout{** WARNING: IEEEtran.bst: No hyphenation pattern has been}%
\typeout{** loaded for the language `#1'. Using the pattern for}%
\typeout{** the default language instead.}%
\else
\language=\csname l@#1\endcsname
\fi
#2}}
\providecommand{\BIBdecl}{\relax}
\BIBdecl

\bibitem{av:kramers--wannier}
A.~Al-Bashabsheh and P.~O. Vontobel, ``The {I}sing model: {K}ramers--{W}annier
  duality and normal factor graphs,'' in \emph{Proc. IEEE Int. Symp. on Inf.
  Theory}, Hong Kong, June 2015, pp. 2266--2270.

\bibitem{bamberg:course}
P.~Bamberg and S.~Sternberg, \emph{A {C}ourse in {M}athematics for {S}tudents
  of {P}hysics}.\hskip 1em plus 0.5em minus 0.4em\relax Cambridge University
  Press, 1991, vol.~2.

\bibitem{hatcher:topology}
A.~Hatcher, \emph{Algebraic {Topology}}.\hskip 1em plus 0.5em minus 0.4em\relax
  Cambridge, 2002.

\bibitem{munkres:algebraic-topology}
J.~R. Munkres, \emph{Elements of {A}lgebraic {T}opology}.\hskip 1em plus 0.5em
  minus 0.4em\relax Addison-Wesley, 1984.

\bibitem{frank:factor}
F.~R. Kschischang, B.~J. Frey, and H.-A. Loeliger, ``Factor graphs and the
  sum-product algorithm,'' \emph{IEEE Trans. Inf. Theory}, vol.~47, no.~2, pp.
  498--519, Feb. 2001.

\bibitem{loeliger:intro}
H.-A. Loeliger, ``An introduction to factor graphs,'' \emph{IEEE Sig. Proc.
  Mag.}, vol.~21, no.~1, pp. 28--41, Jan. 2004.

\bibitem{forney:normal}
G.~D. Forney, Jr., ``Codes on graphs: normal realizations,'' \emph{IEEE Trans.
  Inf. Theory}, vol.~51, no.~2, pp. 520--548, Feb. 2001.

\bibitem{ay:nfg-hol}
A.~Al-Bashabsheh and Y.~Mao, ``Normal factor graphs and holographic
  transformations,'' \emph{IEEE Trans. Inf. Theory}, vol.~57, no.~2, pp.
  752--763, Feb. 2011.

\bibitem{kramers-wannier}
H.~Kramers and G.~Wannier, ``Statistics of the two-dimensional ferromagnet.
  {P}art {I},'' \emph{Physical Review}, vol.~60, no.~3, pp. 252--262, Aug.
  1941.

\bibitem{savit:duality}
R.~Savit, ``Duality in field theory and statistical systems,'' \emph{Reviews of
  Modern Physics}, vol.~52, no.~2, pp. 453--487, April 1980.

\bibitem{druhl:duality}
K.~Dr\"{u}hl and H.~Wagner, ``Algebraic formulation of duality transformations
  for abelian lattice models,'' \emph{Annals of Physics}, vol. 141, no.~2, pp.
  225--253, Feb. 1982.

\bibitem{ising1925}
E.~Ising, ``Beitrag zur {T}heorie des {F}erromagnetismus,'' \emph{Zeitschrift
  f{\"u}r Physik}, vol.~31, no.~1, pp. 253--258, 1925.

\bibitem{maldacena:98:1}
J.~Maldacena, ``The large {$N$} limit of superconformal field theories and
  supergravity,'' \emph{Adv.\ Theor.\ Math.\ Phys.}, vol.~2, pp. 231--252,
  1998.

\bibitem{Li:Vontobel:16:1}
J.~X. Li and P.~O. Vontobel, ``Factor-graph representations of stabilizer
  quantum codes,'' in \emph{Proc.\ 54th Allerton Conf.\ on Communication,
  Control, and Computing}, Allerton House, Monticello, IL, USA, Sep.~28--30
  2016, pp. 1046--1053.

\bibitem{Forney:algebraic:topology:notes}
G.~D. {Forney, Jr.}, ``Codes on graphs: models for elementary algebraic
  topology and statistical physics,'' July 2017, submitted to IEEE Trans.\
  Inf.\ Theory, available online at arXiv:1707.06621.

\bibitem{potts1952}
R.~B. Potts, ``Some generalized order-disorder transformations,'' in
  \emph{Mathematical Proceedings of the Cambridge Philosophical Society},
  vol.~48, no.~1.\hskip 1em plus 0.5em minus 0.4em\relax Cambridge University
  Press, 1952, pp. 106--109.

\bibitem{wu:Potts}
F.-Y. Wu, ``The {P}otts model,'' \emph{Reviews of Modern Physics}, vol.~54,
  no.~1, pp. 235--268, 1982.

\bibitem{forney-pascal:partition-functions}
G.~D. Forney, Jr. and P.~O. Vontobel, ``Partition functions of normal factor
  graphs,'' presented at the Inf. Theory Appl. Workshop, San Diego, CA, Feb.
  2011.

\bibitem{vontobel:electrical}
P.~O. Vontobel and H.-A. Loeliger, ``On factor graphs and electrical
  networks,'' in \emph{Mathematical Systems Theory in Biology, Communications,
  Computation, and Finance (IMA Volumes in Math. and Appl.)}, J.~Rosenthal and
  D.~Gilliam, Eds.\hskip 1em plus 0.5em minus 0.4em\relax New York: Springer
  Verlag, 2003, pp. 469--492.

\bibitem{mao:convolution-fg}
Y.~Mao and F.~R. Kschischang, ``On factor graphs and the {F}ourier transform,''
  \emph{IEEE Trans. Inf. Theory}, vol.~51, no.~5, pp. 1635--1649, 2005.

\bibitem{lidl:finite-fields}
R.~Lidl and H.~Niederreiter, \emph{Finite {F}ields}, ser. Encyclopedia of
  Mathematics and its Applications.\hskip 1em plus 0.5em minus 0.4em\relax
  Cambridge University Press, 1997, vol.~20.

\bibitem{forney:kw-private}
G.~D. Forney, Jr., ``Private communications,'' 2015.

\bibitem{Tanner81}
R.~Tanner, ``A recursive approach to low complexity codes,'' \emph{IEEE Trans.
  Inf. Theory}, vol.~27, no.~5, pp. 533--547, Sept. 1981.

\bibitem{wiberg:95}
N.~Wiberg, H.-A. Loeliger, and R.~K\"otter, ``Codes and iterative decoding on
  general graphs,'' \emph{Euro. Trans. Telecom.}, vol.~6, no.~5, pp. 513--525,
  Sept./Oct. 1995.

\bibitem{kitaev2003fault}
A.~Y. Kitaev, ``Fault-tolerant quantum computation by anyons,'' \emph{Annals of
  Physics}, vol. 303, no.~1, pp. 2--30, 2003.

\bibitem{baxter:exactly}
R.~J. Baxter, \emph{Exactly {S}olved {M}odels in {S}tatistical
  {M}echanics}.\hskip 1em plus 0.5em minus 0.4em\relax Courier Dover
  Publications, 2007.

\bibitem{kato-zeger:capacity}
A.~Kato and K.~Zeger, ``On the capacity of two-dimensional run-length
  constrained channels,'' \emph{IEEE Trans. Inf. Theory}, vol.~45, no.~5, pp.
  1527--1540, July 1999.

\bibitem{whitney:coloring}
H.~Whitney, ``The coloring of graphs,'' \emph{Annals of Mathematics}, pp.
  688--718, Oct. 1932.

\bibitem{barvinok:16:1}
A.~Barvinok, \emph{Combinatorics and Complexity of Partition Functions}.\hskip
  1em plus 0.5em minus 0.4em\relax Cham, Switzerland: Springer, 2016.

\bibitem{potamianos:stochastic}
G.~Potamianos and J.~Goutsias, ``Stochastic approximation algorithms for
  partition function estimation of {G}ibbs random fields,'' \emph{IEEE Trans.
  Inf. Theory}, vol.~43, no.~6, pp. 1948--1965, Nov. 1997.

\bibitem{kadanoff:renorm}
L.~P. Kadanoff, ``Scaling laws for {I}sing models near ${T}_c$,''
  \emph{Physics}, vol.~2, no.~6, pp. 263--272, Feb. 1966.

\bibitem{wilson:renorm}
K.~G. Wilson, ``Renormalization group and critical phenomena. {I}.
  {R}enormalization group and the {K}adanoff scaling picture,'' \emph{Physical
  Review {B}}, vol.~4, no.~9, pp. 3174--3184, Nov. 1971.

\bibitem{bethe}
H.~Bethe, ``Statistical theory of superlattices,'' \emph{Proc. Roy. Soc. London
  A}, vol. 150, no. 871, pp. 552--575, Feb. 1935.

\bibitem{wainwright:bounds}
M.~J. Wainwright, T.~S. Jaakkola, and A.~S. Willsky, ``A new class of upper
  bounds on the log partition function,'' \emph{IEEE Trans. Inf. Theory},
  vol.~51, no.~7, pp. 2313--2335, July 2005.

\bibitem{jerrum:subgraph-world}
M.~Jerrum and A.~Sinclair, ``Polynomial-time approximation algorithms for the
  {Ising} model,'' \emph{SIAM Journal on Computing}, vol.~22, no.~5, pp.
  1087--1116, 1993.

\bibitem{mehdi:ising}
M.~Molkaraie and H.-A. Loeliger, ``Partition function of the {I}sing model via
  factor graph duality,'' in \emph{Proc. IEEE Int. Symp. Inf. Theory},
  Istanbul, Turkey, July 2013, pp. 2304--2308.

\bibitem{ay:partition}
A.~Al-Bashabsheh and Y.~Mao, ``On stochastic estimation of the partition
  function,'' in \emph{Proc. IEEE Int. Symp. on Inf. Theory}, Honolulu, HI,
  June 2014, pp. 1504--1508.

\end{thebibliography}

\begin{IEEEbiographynophoto}{Ali Al-Bashabsheh}
  received a B.Sc.\ (2001) and an M.Sc.\ (2005) in electrical engineering from
  Jordan University of Science and Technology, an M.Sc.\ (2012) in mathematics
  from Carleton University, and a Ph.D.\ (2014) in electrical engineering from
  the University of Ottawa.  His research interests include graphical models,
  information theory, and machine learning.
\end{IEEEbiographynophoto}

\begin{IEEEbiographynophoto}{Pascal O.~Vontobel}
  (S'96--M'97--SM'12) received the Diploma degree in electrical engineering in
  1997, the Post-Diploma degree in information techniques in 2002, and the
  Ph.D.\ degree in electrical engineering in 2003, all from ETH Zurich,
  Switzerland.

  From 1997 to 2002 he was a research and teaching assistant at the Signal and
  Information Processing Laboratory at ETH Zurich, from 2006 to 2013 he was a
  research scientist with the Information Theory Research Group at
  Hewlett--Packard Laboratories in Palo Alto, CA, USA, and since 2014 he has
  been an Associate Professor at the Department of Information Engineering at
  the Chinese University of Hong Kong. Besides this, he was a postdoctoral
  research associate at the University of Illinois at Urbana--Champaign
  (2002--2004), a visiting assistant professor at the University of
  Wisconsin--Madison (2004--2005), a postdoctoral research associate at the
  Massachusetts Institute of Technology (2006), and a visiting scholar at
  Stanford University (2014). His research interests lie in coding and
  information theory, quantum information processing, data science,
  communications, and signal processing.

	Dr.\ Vontobel has been an Associate Editor for the {\sc IEEE Transactions on
	Information Theory} (2009--2012) and the {\sc IEEE Transactions on
	Communications} (2014--2017), an Awards Committee Member of the IEEE
	Information Theory Society (2013--2014), a Distinguished Lecturer of the IEEE
	Information Theory Society (2014--2015), and a TPC co-chair of the 2016 IEEE
	International Symposium on Information Theory. Currently, he is a TPC co-chair
	of the upcoming 2018 IEICE International Symposium on Information Theory and
	its Applications in Singapore and of the upcoming 2018 IEEE Information Theory
	Workshop in Guangzhou, China. He has been on the technical program committees
	of many international conferences and has co-organized several topical
	workshops. Moreover, he has been three times a plenary speaker at
	international information and coding theory conferences and has been awarded
	the ETH medal for his Ph.D.\ dissertation.
\end{IEEEbiographynophoto}

\end{document}